\documentclass[reqno,11pt]{amsart}
\usepackage{amsmath, latexsym, amsfonts, amssymb, amsthm, amscd, stmaryrd, color}
\usepackage[utf8]{inputenc}
\usepackage{xcolor}
\usepackage{graphics,epsf,psfrag}
\usepackage{graphicx}
\numberwithin{equation}{section}

\newtheorem{theorem}{Theorem}[section]
\newtheorem{lemma}[theorem]{Lemma}
\newtheorem{proposition}[theorem]{Proposition}

\newtheorem{rem}[theorem]{Remark}

\newcommand{\overbar}[1]{\mkern 1.5mu\overline{\mkern-1.5mu#1\mkern-1.5mu}\mkern 1.5mu}

\newcommand{\tf}{\textsc{f}}

\newcommand{\lint}{\llbracket}
\newcommand{\rint}{\rrbracket}

\newcommand{\dist}{\mathrm{dist}}

\newcommand{\dd}{\mathrm{d}}
\newcommand{\ind}{\mathbf{1}}

\newcommand{\Z}{\mathbb{Z}}

\renewcommand{\tilde}{\widetilde}

\newcommand{\cB}{{\ensuremath{\mathcal B}} }
\newcommand{\cA}{{\ensuremath{\mathcal A}} }
\newcommand{\cF}{{\ensuremath{\mathcal F}} }
\newcommand{\cG}{{\ensuremath{\mathcal G}} }
\newcommand{\cP}{{\ensuremath{\mathcal P}} }
\newcommand{\cE}{{\ensuremath{\mathcal E}} }

\newcommand{\cC}{{\ensuremath{\mathcal C}} }
\newcommand{\cN}{{\ensuremath{\mathcal N}} }

\newcommand{\cD}{{\ensuremath{\mathcal D}} }

\newcommand{\bP}{{\ensuremath{\mathbf P}} }
\newcommand{\bE}{{\ensuremath{\mathbf E}} }

\DeclareMathSymbol{\leqslant}{\mathalpha}{AMSa}{"36} % nicer `smaller or equal'
\DeclareMathSymbol{\geqslant}{\mathalpha}{AMSa}{"3E} % nicer `larger or equal'
\DeclareMathSymbol{\eset}{\mathalpha}{AMSb}{"3F}     % nicer `emptyset'
\newcommand{\Var}{\mathrm{Var}}
       
\DeclareMathOperator*{\inter}{\bigcap}       % \sum-like symbol for inter
\newcommand{\sumtwo}[2]{\sum_{\substack{#1 \\ #2}}} % sum with 2 lines
\newcommand{\bbE}{{\ensuremath{\mathbb E}} }

\newcommand{\bbG}{{\ensuremath{\mathbb G}} }
\newcommand{\bbH}{{\ensuremath{\mathbb H}} }

\newcommand{\bbN}{{\ensuremath{\mathbb N}} }

\newcommand{\bbP}{{\ensuremath{\mathbb P}} }

\newcommand{\bbR}{{\ensuremath{\mathbb R}} }

\newcommand{\bbZ}{{\ensuremath{\mathbb Z}} }

\newcommand{\ga}{\alpha}
\newcommand{\gb}{\beta}
\newcommand{\gga}{\gamma}            % \gg already exists...
\newcommand{\gd}{\delta}
\newcommand{\gep}{\varepsilon}       % \ge already exists...
\newcommand{\gp}{\varphi}

\newcommand{\gG}{\Gamma}

\newcommand{\gD}{\Delta}

\newcommand{\go}{\omega}

\newcommand{\gl}{\lambda}
\newcommand{\gL}{\Lambda}
\newcommand{\gs}{\sigma}

\newcommand{\supp}{\mathrm{Supp}}
\makeatletter
\def\captionfont@{\footnotesize}
\def\captionheadfont@{\scshape}

\long\def\@makecaption#1#2{%
  \vspace{2mm}
  \setbox\@tempboxa\vbox{\color@setgroup
    \advance\hsize-6pc\noindent
    \captionfont@\captionheadfont@#1\@xp\@ifnotempty\@xp
        {\@cdr#2\@nil}{.\captionfont@\upshape\enspace#2}%
    \unskip\kern-6pc\par
    \global\setbox\@ne\lastbox\color@endgroup}%
  \ifhbox\@ne % the normal case
    \setbox\@ne\hbox{\unhbox\@ne\unskip\unskip\unpenalty\unkern}%
  \fi
  \ifdim\wd\@tempboxa=\z@ % this means caption will fit on one line
    \setbox\@ne\hbox to\columnwidth{\hss\kern-6pc\box\@ne\hss}%
  \else % tempboxa contained more than one line
    \setbox\@ne\vbox{\unvbox\@tempboxa\parskip\z@skip
        \noindent\unhbox\@ne\advance\hsize-6pc\par}%
\fi
  \ifnum\@tempcnta<64 % if the float IS a figure...
    \addvspace\abovecaptionskip
    \moveright 3pc\box\@ne
  \else % if the float IS NOT a figure...
    \moveright 3pc\box\@ne
    \nobreak
    \vskip\belowcaptionskip
  \fi
\relax
}
\makeatother
\def\writefig#1 #2 #3 {\rlap{\kern #1 truecm
\raise #2 truecm \hbox{#3}}}

\newcommand{\cc}{\complement }
\newcommand{\phibar}{H }

\begin{document}

\title[The  disordered lattice free field pinning model approaching criticality]{The  disordered lattice free field pinning model approaching criticality}

\author{Giambattista Giacomin}
\address{Universit\'e de Paris,   Laboratoire de Probabilit{\'e}s, Statistiques  et Mod\'elisation, UMR 8001,
            F-75205 Paris, France}

\address{IMPA - Instituto Nacional de Matem\'atica Pura e Aplicada,
Estrada Dona Castorina 110,
Rio de Janeiro, Brasil 22460-320}
%\email{lacoin@impa.br}
\author{Hubert Lacoin}

\begin{abstract}
We continue the study, initiated in \cite{cf:GL}, of the localization transition of a lattice free field $\phi=(\phi(x))_{x \in \bbZ^d}$, $d\ge3$, in presence of a quenched disordered substrate. The presence of the substrate affects the interface at the spatial sites in which the interface height    is close to zero. This corresponds to the Hamiltonian 
\begin{equation*} \sum_{x\in \bbZ^d }(\beta \go_x+h)\gd_x,\end{equation*}
where $\gd_x=\ind_{[-1,1]}(\phi(x))$, and  $(\go_x)_{x\in \bbZ^d}$ is an IID centered field. A transition takes place when the average pinning potential $h$ goes past a threshold $h_c(\gb)$: from a delocalized phase $h<h_c(\gb)$, where the field is macroscopically repelled by the substrate, to a localized one $h>h_c(\gb)$ where the field sticks to the substrate. In  \cite{cf:GL} the critical value of $h$ is identified and it coincides, up to the sign, with the $\log$-Laplace transform of $\go=\go_x$, that is  $-h_c(\gb)=\gl(\gb):=\log \bbE[e^{\gb\go}]$. Here we obtain  the sharp critical behavior of the free energy approaching criticality:
\begin{equation*}\lim_{u\searrow 0} \frac{  \tf(\gb,h_c(\gb)+u)}{u^2}= \frac{1}{2\,  \Var\left(e^{\gb \go-\gl(\gb)}\right)}.\end{equation*}
Moreover, we give a precise description of the trajectories of the field in the same regime:  the absolute value of the field is $\sqrt{2\gs_d^2\vert\log(h-h_c(\gb))\vert}$   to leading order when $h\searrow h_c(\gb)$  except on a  vanishing fraction of sites ($\gs_d^2$ is the single site variance of the free field). 
%While the behavior of the free energy approaching criticality follows from a two scale decomposition of the field, the  properties of the trajectories require a full multiscale analysis.
\\[10pt]
  \\[10pt]
\end{abstract}

\maketitle

\tableofcontents

\section{Introduction}
\subsection{Disorder and critical phenomena: an overview, till the free field pinning case}
A fundamental issue in statistical mechanics is the effect of  \emph{disorder}, synonymous of \emph{random environment} and (with more old fashioned language) of
 \emph{impurities}, 
 on phase transitions. The issue is very general and applies to any  statistical model that exhibits a phase transition: in mathematical terms, a phase transition happens at a given value, called critical, of a parameter (the temperature, an external field,$\ldots$) when one or more observables on the system  have a singular,  i.e. non-analytic, behavior in the parameter we are considering, at the critical value. 
 The behavior of systems approaching criticality is particularly interesting because of the appearance of \emph{universal behaviors} 
 that are, to a certain extent, characterized in terms of \emph{critical exponents} (e.g. \cite{cf:FFS,cf:FV}).
 Consider now    a \emph{spatially homogeneous} model, i.~e. a model in which the interactions are translation invariant, that has a phase transition.  If we modify the interactions by perturbing them in a spatially random way, we obtain, for every realization of the randomness, a different non homogeneous model (that we call \emph{disordered}): does the phase transition survive to the introduction of this randomness? And, if it does, is the nature of the  transition affected? That is, are the critical exponents  the same as in the homogeneous case?  
 
In spite of the general nature of the problem, phase transitions and critical phenomena are under control only for particular homogeneous models, or classes of homogeneous models. The most famous one, and first (nontrivial) one to be solved (in 1944, by Lars Onsager) is the two dimensional Ising model (on square lattice, with ferromagnetic interactions and  in absence of external field): 
this model has been at the heart of the activity of a large community of researchers since. A part of this community, mostly on the physical side, focused on the issue that interests us, that is whether Onsager's results withstand the introduction of disorder, for example a small amount of disorder. 
And it is precisely in the Ising model context that 
  A.~B.~Harris \cite{cf:Harris} took an approach to this question that turned out to be very successful in the physical community.
  Harris' approach is based on the renormalization group and can be summed up (in a vague but hopefully evocative fashion) by 
  saying that one has to consider what is the effect of the renormalization transformation on the disorder when the system is close to criticality. If the renormalization tends to suppress the disorder then one expects that on large scale the disorder will be \emph{irrelevant}, and the critical phenomenon will not be affected by the disorder. On the other hand, if disorder is enhanced by the renormalization group transformation, one generically expects that the critical behavior is affected by the disorder, that is therefore dubbed \emph{relevant}. 
  The success of Harris' arguments  is in part due to the fact that he was able to make them boil down to a very simple criterion, called \emph{Harris criterion}.  
  
  In spite of the fact that these ideas are around since at least 45 years and that they are commonly applied in physics, from the mathematical viewpoint the understanding of the Harris criterion is very limited, notably for the original example of the two dimensional Ising model (see \cite[Ch.~5]{cf:G} for a review). 
  Only more recently (see \cite{cf:BL,cf:G} and references therein) the Harris criterion prediction has been proven in full  for a class of statistical mechanics models: the pinning models.
  
 Keeping at a very informal level, pinning models can be visualized as \emph{interface} pinning models.  
 An \emph{effective} $d+1$ dimensional  interface is modeled by considering a random function from $\bbZ^d$ to $\bbR$ (or to $\bbZ$): examples include the Lattice Free Field (LFF) or other \emph{gradient} fields like the massless fields  or the Solid On Solid (SOS) models (see e.g.\ \cite[Ch.8]{cf:FV} for an introduction to the LFF or \cite{cf:Fun, cf:Shef,cf:Vel} for more advanced material).
 In $d=1$ these interface models    just reduce to random walk models.
 The pinning potential is a reward that is introduced via an energy term (we are taking a Gibbsian viewpoint of the probability law of the model) that rewards or penalizes the visit to level zero (if the interface takes values in $\bbZ$) or a neighborhood of level $0$
 (if the interface takes values in $\bbR$): we call these visits \emph{contacts}. The intensity of the reward is parametrized by a variable $h$, and it can become a penalization if we change the sign of $h$. In the disordered case we simply make $h$ depend on $x$: the parametrization we choose is $h+\gb \go_x$, where $(\go_x)_{x\in \bbZ ^d}$ are IID centered random variables (with suitable integrability properties), and $\gb\ge 0$.
 
As already understood in \cite{cf:MFisher}, the  $d=1$ model has an intrinsic independence structure that allows in particular a generalization of the model that turns out to be very important in the Harris criterion perspective: in mathematical terms the model can be rewritten (for every value of $d$) just in terms of the point process represented by the location of the contacts and  if $d=1$ the contact set is just a renewal process (if the interface takes values in $\bbZ$, otherwise it is  a Markov renewal process)\cite{cf:GB,cf:G}. This is not only precious in solving the model -- notably, the $d=1$ homogeneous case is exactly solvable -- but it offers an immediate natural generalization to the large class of renewal pinning models. So for $d=1$, in the generalized context we just hinted to, 
one can obtain models for which Harris criterion predicts disorder irrelevance and other ones for which it predicts disorder relevance. We refer to \cite{cf:GB, cf:G} for the large literature on $1+1$ dimensional pinning and renewal pinning. But we want to stress that if the irrelevant disorder results are very satisfactory (and they are proven exactly when Harris criterion predicts irrelevance), relevant disorder results are much weaker. This is not surprising: Harris criterion does not bear information about what the  critical behavior is, if disorder is relevant. Nevertheless, it has been shown that, when the Harris criterion predicts disorder relevance, the critical behavior is not the same as the one of the homogeneous model. What the disordered critical behavior really is remains mathematically a fully open issue.
Substantial progress on this problem has been recently achieved, but not for the pinning models itself:
the critical behavior of a relevant disorder case for one class of \emph{copolymer} pinning models 
and of a simplified version of the hierarchical   pinning model  have been identified respectively in \cite{cf:BGL} and  in \cite{cf:CDHLS}. 
 
The $d>1$ case is a priori more difficult to handle and, above all, the contact set does not enjoy the independence (renewal) structure of the $d=1$ case: it is replaced by a more geometric \emph{spatial} Markov property. %In particular, natural generalizations    of the contact set process are less obvious. 
The problem has been attacked in \cite{cf:CM1,cf:GL,cf:FF2} by using the LFF as interface model: the homogeneous model 
turns out to be  solvable (or, at least, has a certain degree of solvability) and, as far as the questions we raise, in a rather elementary way. In particular, it displays a delocalization to localization transition at a critical value $h_c$, as $h$ grows. 
It is also rather straightforward to see that this transition survives when disorder is introduced, that is for $\gb>0$.
A peculiar feature of this transition is that $h_c$, or $h_c(\gb)$ when $\gb>0$, separates the regime $h<h_c(\gb)$ in which the contact fraction is zero (delocalized regime), and the regime $h>h_c(\gb)$ in which the contact fraction is positive (localized regime).

What is instead much less obvious \cite{cf:GL,cf:FF2} is the identification of   the critical value $h_c(\gb)$, along with  estimates on the contact fraction  of the system  that show that disorder is relevant in all dimensions $d\ge 2$. 
In particular, for $d\ge 3$, the case on which we focus here, we have proven in \cite{cf:GL} that 
the contact fraction approaching criticality, i.e. $h\searrow h_c(\gb)$, is bounded 
above and below by $h-h_c(\gb)$ times a positive constant (different for lower and upper bound). This result has been established only for Gaussian disorder, while for more general disorder a lower bound of $(h-h_c(\gb))^c$, $c>1$ is a constant that depends on $d$. 
Therefore the contact fraction is (Lipschitz) continuous for $\gb>0$ and 
this is sufficient to infer that disorder is relevant. In  fact for $\gb=0$ the contact fraction is discontinuous at the critical value.

The content of the work we present now is: 
\smallskip

\begin{enumerate}
\item showing that, when $h\searrow h_c(\gb)$, the contact fraction behaves like $h-h_c(\gb)$ times a constant that depends on $\gb$ and on the law of the disorder, on which we make only integrability assumptions;
\item providing precise path estimates in the same limit. That is, describing  the trajectories on which the  system concentrates near criticality.
\end{enumerate}   

\smallskip

Precise contact fraction estimates like the ones in point (1)  typically demand at least a heuristic understanding of 
the path behavior of point (2). Therefore  in our context they  demand a good understanding of the \emph{localization mechanism} for the disordered system near criticality. This is one of the main achievements of our analysis. 

On the other hand, the step from (1) to (2) is by no mean evident and, as a matter of fact, it is technically the most demanding part of our analysis, involving in particular a full multiscale analysis.

\subsection{The model's building blocks: Lattice Free Field and disorder} 
\label{sec:building}
We set 
$\gL_N:=\lint 0,N \rint^d$, $d\ge 3$,  $N\in \bbN$,  and consider the centered free field on 
this set. That is, we consider a Gaussian family of centered random variables $(\phi(x))_{x\in \bbZ^d}$, whose law is denoted 
by $\bP_N$, with $\bE_N[\phi(x)\phi(y)]:= G_N(x,y)$ where $G_N(\cdot, \cdot)$ 
is the Green function associated with the simple symmetric random walk on $\bbZ^d$  killed upon exiting $\mathring{\gL}_N:= 
\lint 1,N-1 \rint^d$. 
More explicitly, if $P_x$ denotes the distribution law of
a simple symmetric continuous time random walk  $(S_t)_{t\ge 0}$ with jump rate one in each direction and initial condition $S_0=x$, 
 then
\begin{equation}
\label{zoop}
G_N(x,y)\, =\, E_x\left[ \int_0^{\tau_N} \ind_{\{S_t=y\}} \dd t \right]\, , \  \ \text{ with }  \ \ \tau_N:= \inf\left\{t>0: \, S_t \notin \mathring{\gL}_N\right\}\,.
\end{equation}
% (more generally we write $G_\gL$ for the Green function in $\gL\subset \bbZ^d$ where $\tau_N$ is replaced by the hitting time of the internal boundary of $\gL$).
It is well known that the Green function $G(\cdot, \cdot)$ of the simple random walk without killing (obtained by replacing $\tau_N$ in \eqref{zoop} by $\infty$) exists for $d\ge 3$. We let $\bP$ denote the law of the Gaussian field on $\bbZ^d$ with covariance $G(\cdot,\cdot)$. We also set  $\gs_d^2:= G(0,0)$.

Let us recall from now some well known random walk estimates in transient dimensions that we will repeatedly use. First of all
 $\lim_{x \to \infty}\vert x\vert ^{d-2} G(0,x)$ exists and it is positive, so in particular we can find $C_d>1$ such that for every $x\neq 0$
\begin{equation}
\label{eq:Gest-tail}
 C_d^{-1} \vert x \vert^{2-d} \, \le \, 
G(0, x) \, \le \, C_d \vert x \vert^{2-d}\, .
\end{equation}
Moreover
\begin{equation}
\label{eq:Gest0}
0\le \, G(0,0)- G_N(0,0)\, =\, O \left(1/N^{d-2}\right)\, ,
\end{equation}
and,  always aiming at comparing the Green function and its killed version, we have that for
any sequence  $z_N$ such that $\lim_{N\to \infty}\dist(z_N, \gL^{\cc}_N)=\infty$
(e.g.\
$z_N:=\lceil N/2 \rceil( 1, 1,\dots,1)$) for every $x, y \in \bbZ^d$
\begin{equation}
\label{eq:zoop}
G(x,y)=G(0,y-x)= \lim_{N\to \infty} G_N(x+z_N,y+z_N)\, .
\end{equation}
Of course, for $x\in \bbZ^d$ and $A\subset \bbZ^d$,  $\dist(x,A):= \min_{y\in A}|x-y|$ and we make the choice (irrelevant in most of the cases, but of some importance for some geometric constructions) that 
$|\cdot|$ denotes the $\ell_1$ distance in $\bbZ^d$, that is $\vert x \vert = \sum_{j=1}^d \vert x_j \vert$
(but  for $x \in \bbR^d$ we use $\vert x \vert$ for the Euclidean norm).

\medskip

The disorder, or random environment, $(\go_x)_{x \in \bbZ^d}$ is a family of IID random variables, $\bbP$ is its law.  Free field and disorder are independent. We assume that
\begin{equation}
\label{eq:lambda}
\gl(s)\, :=\, \log \bbE \left[ \exp\left( s \go\right)\right]\, < \, \infty\ \ \text{ for every } s\in \bbR\,,
\end{equation}
and that $\gl'(\cdot)$ is not a constant, i.~e. $\go$ is not a constant.
In \eqref{eq:lambda} 
we have dropped the index $x$ for obvious reasons. Without loss of generality we  assume $\bbE[\go]=0$:
this is largely irrelevant because 
$\go$  appears in the model  in the form $\gb\go -\gl(\gb)$ which is invariant under the transformation $\go \mapsto \go+$constant.
However, $\bbE[\go]=0$ assures that $\gl(\cdot)$ is increasing on the positive semi-axis and decreasing in the negative one, and this is practical.

The generalization of the results to the case in which we assume \eqref{eq:lambda} only, say, for $\vert s \vert$ smaller than a constant is not straightforward. The full hypothesis \eqref{eq:lambda} is used for a cut-off estimate, see Section~\ref{sec:cut-off}, that is probably not necessary but it does not appear to be easy to circumvent. On the other hand, 
a part of the main results (notably, the probability upper bound) can be obtained under very mild hypothesis on the lower (negative) tail of $\go$, in particular for this results \eqref{eq:lambda} is exploited only for $s>0$. 
For sake of readability we will make precise this aspect only in the technical part of the work (see Remark~\ref{rem:weakcond}).

%and that the derivative $\gl'(\cdot)$ is not a constant. In \eqref{eq:lambda}  we have dropped the index $x$ for obvious reasons. On the lower tail of $\go$ we assume only that $\bbE[\go_-]< \infty$, with the notation $\go_- = - \go \ind_{\{\go <0\}}$ for the negative part, and this is of course equivalent to $\bbE[ \vert \go \vert] < \infty$ since the positive part is controlled by \eqref{eq:lambda}. Without loss of generality we set $\bbE [\go]=0$, hence $\gl(0)=0$ and $\gl(s)>0$ for every $s\neq 0$. 

\subsection{The disordered lattice free field pinning model}

The model we consider is the disordered pinning model based on the LFF  with law $\bP_N$.
For $\gb\ge 0$ and $h\in \bbR$ we set
\begin{equation}
\label{eq:themodel}
\frac{\dd \bP_{N, \go, \gb,h}}{\dd \bP_N}(\phi)\, :=\, 
\frac 1{Z_{N, \go, \gb,h}}
e^{ \sum_{x \in \gL_N}(\gb \go_x -\gl(\gb) +h) \gd_x}\, , \ \ \text{ with } \ 
\gd_x:= \ind_{[-1,1]}(\phi (x))\, ,
\end{equation}
where of course $Z_{N, \go, \gb,h}$ is the normalization constant (or \emph{partition function})
\begin{equation}
\label{eq:theZ}
Z_{N, \go, \gb,h}\, :=\, \bE_N \left[ \exp\left( \sum_{x \in \gL_N}(\gb \go_x -\gl(\gb) +h) \gd_x\right)\right]\, .
\end{equation}
We will often use  the  notation 
\begin{equation}
\label{eq:cF}
\cF_A^\phi\, :=\, 
\gs ( \phi(x):\, x \in A)\,.
\end{equation} 
Note that $\bP_{N, \go, \gb,h}$ does not change if we replace summing over $x\in \gL_N$ with $x \in \mathring{\gL}_N$ or any other set
$\gL$,  $\mathring{\gL}_N \subset \gL \subset \gL_N$. 
$Z_{N, \go, \gb,h}$ is affected by such a change, but only in a trivial way and, in particular, the \emph{free energy density} 
(that we will simply call  \emph{free energy})
\begin{equation}
\label{eq:theF}
\tf(\gb, h) \, :=\, 
\lim_{N \to \infty} \frac 1{\vert \gL_N \vert} \bbE \log Z_{N, \go, \gb,h}\, ,
\end{equation}
is clearly not affected either. The proof of the  existence of the limit in \eqref{eq:theF} can be found in \cite{cf:CM1}: the argument is based on the almost super-additive behavior  of the sequence $(\bbE \log Z_{N, \go, \gb,h})_{N\in \bbN}$. We will come back to to this issue (see Proposition \ref{th:superadd}) because a sharp super-additive  behavior for a modified partition function, that gives rise to the same free energy, is going to be important, notably  for the lower bound on  the free energy in Section~\ref{sec:LB}. Here are some basic, but crucial, properties of the free energy (see the introduction of \cite{cf:GL} for full details):

\medskip

\begin{itemize}
\item The map $(\gb, h) \mapsto \tf(\gb, h)$ is convex, moreover it is non decreasing in $h\in \bbR$ for $\gb$ fixed  and in $\gb\ge 0$ for $h$ fixed;
\item The inequality $\tf(\gb, h)\ge 0$ holds because of the rough \emph{entropic repulsion} estimate $\log \bP_1(\phi(x)>1$ for every $x\in \mathring{\gL}_N)= o(\vert \gL_N \vert)$ which is easily derived by exploiting the continuum symmetry of the LFF \cite{cf:LM};
\item  By Jensen's inequality, we have $\tf(\gb, h) \le \tf(0,h)$ (\emph{annealed bound})\, .
\end{itemize}
\noindent The convexity and monotonicity properties in $h$ lead to identifying 
\begin{equation}
h_c(\gb)\, :=\,  \inf \{h:\, \tf(\gb, h)>0\}\,=\, \inf \{ h : \,   \partial_h\tf(\gb,h)>0\}\, , 
\end{equation} 
as a critical point, provided that $h_c(\gb)\not= \pm \infty$. Elementary estimates lead to excluding $h_c(\gb)\not= \pm \infty$, but much more than that is true: 
 in \cite{cf:GL} is shown that  
$h_c(\gb)=0$ for every $\gb\ge 0$. Again we refer to  the introduction of \cite{cf:GL} for full details, 
but we stress that establishing $h_c(\gb)\ge 0$ is a rather straightforward consequence of comparison with the model without disorder: of course
$\tf(0, h)=0$ for $h\le 0$ and a very moderate amount of work leads to 
\begin{equation}
\label{eq:beta0}
\tf(0, h) \stackrel{h \searrow 0}\sim c_d h\, , \ \ \text{ with } 
c_d\, :=\, \bP( \phi( 0) \in[-1,1])\, =\, P(\gs_d \cN \in [-1,1])\, ,
\end{equation}
and  $\cN$ is our notation for a standard Gaussian variable. In particular
\eqref{eq:beta0} yields $h_c(0)=0$ and, in turn, from  the annealed bound we obtain 
$h_c(\gb)\ge  h_c(0)$.  Remark also  that from the annealed bound we extract $\tf(\gb, h) \le c h$, for any $c> c_d$ and $h>0$ small. 
However this result is poor precisely because disorder is relevant for this model:  
the main result in \cite{cf:GL} is that if $\go_x$'s are standard Gaussian variables then for $h\in(0,1)$
\begin{equation}
\label{eq:mainGL}
  C_-(\gb) h^2 \, \le\, \tf(\gb, h) \, \le \, C_+(\gb) h^2\, ,
\end{equation} 
with $C_\pm(\gb)>0$ satisfying $\lim_{\gb \searrow 0}\gb ^2 C_\pm(\gb)= c_{\pm}>0$,  and $c_-< c_+$. 
We remark that, exploiting the convexity of $\tf(\gb, \cdot)$ and the fact that 
$\partial_h  \log Z_{N, \go, \gb,h}= \bE_{N, \go, \gb,h}[ \sum_{x \in \gL_N} \gd_x ]$, 
one readily obtains that the infinite volume contact density 
$\rho (\gb, h) := \lim_N \rho_N(\gb, h)$, with $\rho_N(\gb, h):= 
 \bE_{N, \go, \gb,h}\left[\sum_{x \in \gL_N} \gd_x/ {\vert \gL_N\vert}\right]$, exists 
 except possibly for countably many values of $h$: the non decreasing function $\rho (\gb, \cdot)$ may have jumps and, when it does, 
its value is not well defined at the jump. In order to avoid this nuisance we extend the definition of $\rho(\gb, h)$ 
choosing the right continuous version of $\rho(\gb, \cdot)$  (the results that follow are exactly the same for the left continuous version).  From \eqref{eq:mainGL} one easily obtains 
\begin{equation}
\label{eq:mainGLrho}
2C_-(\gb) h \, \le\, \rho(\gb, h)  \, \le \, 4C_+(\gb) h \, ,
\end{equation} 
for every $h \in (0,1)$ for the lower bound and for every $h \in (0,1/2)$ for the upper bound, see Remark~\ref{rem:Ftorho}.

Also results for general disorder are given in  \cite{cf:GL}, but they are rougher than 
\eqref{eq:mainGL}, in the sense that, in the lower bound in \eqref{eq:mainGL}, $h^2$  is replaced by $h$ to a large power. 
Nevertheless, the results we just cited show that disorder is relevant for the model we consider: the critical behavior changes when disorder is switched on.

\subsection{The main results}

The first main result of this paper is a sharp version of \eqref{eq:mainGL}, valid for general disorder distribution. We prove that $\tf(\gb, h)$ is asymptotically proportional to $h^2$ when $h\searrow 0$ and identify
the value of the constant in front of $h^2$.

\medskip

\begin{theorem}
\label{th:mainF}
For every $\gb >0$ we have 
that 
\begin{equation}
\label{eq:mainF}
\tf (\gb, h) \stackrel{h \searrow 0 } \sim \chi(\gb) h^2\, ,
\end{equation}
with
\begin{equation}
\label{eq:chi}
\chi(\gb)\,:=\, \frac{1}{2\Var(e^{\gb \go_x-\gl(\gb)})}= \frac{1}{2(e^{\gl(2\gb)- 2\gl(\gb)}-1)}\,.
\end{equation}
\end{theorem}
Theorem~\ref{th:mainF} sums up  quantitative upper  and lower bounds (Propositions~\ref{th:final-ub} and \ref{lloobb} respectively) that bear more information on the rate of convergence of \eqref{eq:mainF}. Moreover one easily extracts from \eqref{eq:mainF}
the following asymptotic equivalence on the contact fraction (cf. Remark~\ref{rem:Ftorho})
\begin{equation}
\label{eq:mainFrho}
\rho(\gb, h) \stackrel{h \searrow 0 } \sim 2\chi(\gb) h\,.
\end{equation}
Of course \eqref{eq:mainFrho} gives already a precise information on the behavior of the trajectories with law $\bP_{N, \go, \gb, h}$ in the infinite volume limit and near criticality.  
Our second main result, Theorem~\ref{th:paths} below, goes much farther in this direction.

\medskip

Recall, cf. the end of Section~\ref{sec:building},  that $\sigma_d$ denotes the standard deviation of the one site marginal of the  infinite volume LFF.
The following result shows that asymptotically most of the points in the field are located around height $\sigma_d \sqrt{2 \log (1/h)}$.

\medskip

\begin{theorem}
\label{th:paths}
Given $\gep>0$, there exists $h_0(\gep)>0$  such that for every $h\in (0,h_0(\gep))$ we can find 
$c:=c(\gep,h)>0$ and $N_0:=N_0(\go, \gep, h)$, with $\bbP( N_0 < \infty)=1$, such that $\bbP(\dd \go)$-a.s. we have  for $N>N_0$ 
 \begin{equation}
\bP_{N, \go, \gb, h}\left( \left \vert 
 \left \{
  {x\in \gL_N}:\, 
   \left| 
  {|\phi(x)|}{{(\log (1/h))^{-1/2}}}-\sqrt{2}\,  \sigma_d   
  \right|> \gep 
  \right\} \right \vert
   \, \ge\,  \gep N^d 
  \right)\, \le\,  e^{-c N^d}\, .
 \end{equation}
\end{theorem}

\medskip

This result considerably refines previous estimates obtained on the trajectories. In \cite{cf:GL} it was only proved that typically most point are above height $c_d  \sqrt{ \log( 1/h)}$ (in absolute value) for an explicit  non-optimal constant $c_d$. 

\medskip

\begin{rem}
\label{rem:Ftorho}
The arguments to go  from \eqref{eq:mainGL} to \eqref{eq:mainGLrho} and from \eqref{eq:mainF} to 
\eqref{eq:mainFrho} are standard, but we sketch them here. If $f:[0, \infty) \to [0, \infty)$ is  a convex increasing function  such that $C_1 h^2 \le  f(h) \le C_2 h^2$ for $h\in [0, h_0]$, with $C_1>0$,
then convexity directly yields $\partial f(h) \ge 2 C_1 h$. Here $\partial f (\cdot)$ is either the upper or the lower differential 
of $f$.  Since $\partial f(\cdot) $ is non decreasing we have that, for every $h\in [0, h_0/2]$, $\partial f (u) \ge \partial f(h) \ind_{[h, \infty)}(u)$.
By integrating from $0$ to $2h$ this inequality we obtain $\partial f(h) h \le f(2h) \le 4 C_2 h^2$, that is $\partial f(h) \le 4C_2 h$. On the other hand, 
if $C_2=(1+ \gep^2)C_1$, then by integrating $\partial f(u) \ge 2C_1 u  \ind_{(0, h)}(u)+ \partial f(h) \ind_{[ h, \infty)}(u)$ for $u$ going from $0$ to $(1+ \gep)h \le h_0$ we obtain that  $ \partial f(h)\le 2C_1 h (1+3 \gep)$, for $\gep\in [0, 1]$.
\end{rem}

\bigskip 

\subsection{Discussion of the results, relevant literature and organization of the paper}

\subsubsection{Localization strategy and %high-level 
sketch of proofs}
A key point and one of the main novelty of our work is that we  identify the \emph{localization mechanism}. And this mechanism 
is crucially suggested by the argument for  
the upper bound on the free energy. The upper bound we give  is a universal bound: it holds for an arbitrary random contact set. More explicitly:   
the partition function \eqref{eq:theZ} depends only on the random field $(\gd_x)_{x\in \mathring{\gL}_N}$ or, equivalently,  on the random set $\{x \in \mathring{\gL}_N: \, \gd_x=1\}$, and we give
an upper bound on the free energy density not only for  $\gd_x= \ind_{[-1,1]}(\phi(x))$ with $\phi$ the LFF on $\gL_N$ with zero boundary conditions, but with an arbitrary law of $(\gd_x)_{x\in \mathring{\gL}_N}$. 
%which is equivalent to an arbitrary choice of $\bP_N$.
Moreover this bound is saturated by choosing $(\gd_x)_{x\in \mathring{\gL}_N}$ IID Bernoulli variables of a  parameter $p$
chosen to be the value $p(\gb,h)$ that maximizes the function
\begin{equation}\label{Taylz}
p \mapsto \bbE\log \left(1+p \left( \exp\left(\gb \go -\gl(\gb) +h \right)\right) \right)\, .
\end{equation}
In the limit  $h \searrow 0$ we obtain that $p(\gb,h) \sim 2 \chi(\gb) h$.   
Why should the contact set that we obtain from the LFF be close to saturating this bound too? 
At a heuristic level the reason is a combination of two well known facts on the LFF: 
\begin{enumerate}
\item the continuum symmetry of the LFF that makes 
\emph{rigid} vertical translations of the interface little expensive (this is of course very much in the logic of the \emph{entropic repulsion phenomena}
\cite{cf:BDG,cf:BDZ,cf:D,cf:DG,cf:LM}) 
\item large excursions of the LFF have very mild correlations (an issue underlying \cite{cf:BDZ,cf:DG} and developed in detail in \cite{cf:CCH}). 
\end{enumerate}
This suggests the following behavior for the field when $h>0$ is small: the field shifts away from level $0$, 
precisely it shifts to a level $u(\gb,h)$, or $-u(\gb,h)$, so that the probability that $\phi(x)\sim \cN (\gs_d^2, \pm u(\gb,h))$ belongs to
$[-1,1]$ is equal to $p(\gb,h)$. As the reader might expect in view of Theorem~\ref{th:paths},  it turns out that $u(\gb,h)$ is asymptotically equivalent to the square root 
of $2\gs_d^2 \log (1/h)$.

\medskip

To  substantiate this localization strategy we need  to provide a lower bound. This is achieved by considering the field with boundary conditions $u(\gb,h)$ -- the value of the free energy does not depend on on this choice -- and  via a two step decomposition of the LFF that is in the spirit of several earlier works: in three or more dimension the LFF can be written 
as the superposition of a field with small variance, and spatially power law decaying covariance, and a field that accounts for almost all the variance of the original field, but with exponentially decaying covariance: it is  for example the case of the decomposition of the field theory literature \cite{cf:BGM,cf:FFS} in which 
$G=G_\epsilon+ (G- G_\epsilon)$ where $G_\epsilon$ is the Green function of a walk with a rate of death $\epsilon>0$, see \cite[Sec.~4]{cf:DG} for 
a probabilistic presentation. 
  We propose instead a 
decomposition that is much more geometrically structured: we write $\phi$ as a power law correlated field (with small variance) plus  independent fields that are compactly supported over boxes (i.e. hypercubes). The boxes have edge length  proportional to $h^{-c}$, 
$c>0$ a small constant, and they overlap only near the boundary.  Recalling that the boundary of the LFF is set to height $u(\gb, h)$, hence the mean of the field is $u(\gb, h)$,   in each one of these boxes we typically expect no contact, because  the contact density is proportional to $h$ and the volume of each box is $h^{-cd}$ (and $c$ is small, in particular $c<1/d$). On this scale we are able to perform an accurate analysis that shows  that the leading contribution to the free energy in each of these boxes is given by configurations 
with one or two contacts. The errors introduced by the power law correlated field (with small variance) and by the overlap regions turn out to be higher order corrections.

\medskip

The next step is proving that the trajectories of the field behave like what is suggested by the asymptotic of the free energy and its proof.
This is a matter of proving upper and lower bounds on the height of the field with law $\bP_{N, \go, \gb, h}$, that is 
\begin{equation}
\label{upperandlower}
 \begin{split}
   \text{(A)}& \quad  \bP_{N, \go, \gb, h}\left( \left \vert
 \left \{
  {x\in \gL_N}:\, 
  \frac{|\phi(x)|}{\sqrt{\log (1/h)}}\,<\, \sqrt{2}\,  \sigma_d  - \gep 
  \right\}\right\vert \ge \gep N^d \right) \,\le\,  e^{-c  N^d}\, ,
  \\
   \text{(B)}& \quad \bP_{N, \go, \gb, h}\left( \left \vert
 \left \{
  {x\in \gL_N}:\, 
  \frac{|\phi(x)|}{\sqrt{\log (1/h)}}\,>\, \sqrt{2} \, \sigma_d  + \gep 
  \right\}\right\vert
   \ge \gep N^d \right) \,\le\,  e^{-c  N^d}\,.
 \end{split}
\end{equation}

Showing (A), that is  that the field shifts, except possibly  for $\gep \vert\gL_N\vert$ sites,   to a height of at least $(1-\gep) (2\gs_d^2 \log(1/h))^{1/2}$ is not too difficult.
In fact, again by a two scale argument (but rather standard, in the spirit of the arguments in \cite{cf:BDZ,cf:DG} and already exploited in \cite{cf:GL}) one establishes that if one partitions $\gL_N$
in boxes of edge length $L$ (large but fixed), in most of these boxes there is a point in which the field is close to the correct height 
$(2\gs_d^2 \log(1/h))^{1/2}$.  This is incompatible with having a density of sites on which the field is below $(1-\gep) (2\gs_d^2 \log(1/h))^{1/2}$, because it forces 
a density of sites $x$ to have a neighbor $y$ with $\vert \phi(x)- \phi(y)\vert \ge c (\log(1/h))^{1/2}$ ($c>$ small but not depending on $h$).
And this is highly penalized by the LFF Hamiltonian.

 The lower bound on the trajectory we just outlined  is relatively short and it is just a refinement of  the argument in \cite{cf:GL}.
 On the other hand, showing (B) in \eqref{upperandlower}, that is  that the field shifts, except possibly  for $\gep\vert\gL_N\vert $ sites,  to a height of at most $(1+\gep) (2\gs_d^2 \log(1/h))^{1/2}$ is
 substantially harder and requires novel arguments. This is because being too close to the pinning region, i.e. level zero,  is directly penalized. However
   the fact that the field is too far from the pinning region in a small (but positive) density of sites says that
it cannot collect the expected amount of rewards on those sites, but this  does not exclude, at least not in an obvious way,   that these rewards are collected elsewhere. After all, only rare spikes hit the pinning potential region with the strategy we have outlined. This estimate therefore has to exploit a more collective behavior of the field and the keyword at this stage is certainly \emph{rigidity of the interface}.
 But, in practice, implementing a proof along the reasoning that we just sketched is not straightforward and the control from above of the trajectories  comes via two non trivial and technically demanding estimates:
 \begin{enumerate}
 \item a control of the contact fraction on  \emph{mesoscopic} scales, notably down to the boxes of volumes that are just a bit larger than $h^{-2}$: we stress that these boxes become large when $h\searrow 0$, but they are of constant size with respect to $N$;
 \item an estimate of the rigidity of the field that demands a full multiscale analysis.
 \end{enumerate}  

 \subsubsection{Open problems: sign of $\phi$ and disordered induced symmetry breaking} 
 
An obvious question at this stage is: what about the sign of the field? It is natural to conjecture that for small values of $h$ most sites are located on the same side of the interface. On the other hand, we 
 believe that $N^{-d}\sum_{x\in \gL_N} \ind_{\{\phi_x>0\}}$ converges to $1/2$ for sufficiently large values of $h$, since, in that regime, most sites are favorable to contact. 
This corresponds to the following convergence in law (conjectural) statement on $\bP_{N, \go, \gb, h}$ for $h> 0$:
\begin{equation}
\label{eq:conj1}
 \frac{1}{N^d}\sum_{x\in \gL_N} \ind_{\{\phi_x>0\}} \stackrel{N\to \infty}{\Rightarrow} \rho_\gb(h) \mathrm{Ber}(1/2) +(1-\rho_\gb(h)) 
(1- \mathrm{Ber}(1/2))\, ,
\end{equation}
where $\rho_\gb(h)\in [1/2, 1)$ and approaches $1$ as $h \searrow 0$, while $\rho_\gb(h)=1/2$ when $h$ is above a threshold. %(and $ \mathrm{Ber}(1/2)$ is a Bernoulli random variable of parameter $1/2$).
According to this conjecture, the interface lies above level zero in a majority of sites if Ber$(1/2)=1$ and $\rho_\gb(h)>1/2$: $\rho_\gb(h)$ is precisely the density of these sites. In this case there is  therefore a density $1-\rho_\gb(h)$ of sites below level zero and, by symmetry,
if Ber$(1/2)=0$, i.e. the interface lies below level zero in a majority of sites, the density of sites above zero is $1-\rho_\gb(h)$.
This phenomenon disappears when, for $h$ above a threshold, $\rho_\gb(h)$ becomes $1/2$, and  the right-hand side of 
\eqref{eq:conj1} becomes equal to $1/2$ too.

As we already pointed out, the  value $1/2$ for the parameter of the limiting Bernoulli variable comes from symmetry. We believe that this probability for the field to be \emph{mostly positive} is very sensitive to boundary condition:  
if we replace the centered LFF $\phi$ that defines the model with $\phi+c$, any $c>0$,
we expect  \eqref{eq:conj2} to become
\begin{equation}
\label{eq:conj2}
\lim_{N \to \infty} \frac{1}{N^d}\sum_{x\in \gL_N} \ind_{\{\phi_x>0\}} \, =\,  \rho_\gb(h) \, ,
\end{equation}
in $\bP_{N, \go, \gb, h}$-probability. 

Obtaining a proof of \eqref{eq:conj1} and/or   \eqref{eq:conj2} appears to be very challenging. 
Nevertheless, sidetracking farther, we observe that
 they suggest to the following consideration concerning Gibbs states for the disordered model.
It seems reasonable to expect  that for positive values of $h$  there is a unique translation invariant Gibbs state associated with the homogenous model (we warn the reader that already this step is speculative and it represents in itself  a challenging conjecture). On the other hand, \eqref{eq:conj1}-\eqref{eq:conj2} indicate that for the disordered model there are at least two Gibbs states: one corresponding the limit obtained with positive boundary condition (i.e., $c>0$) and another one corresponding to the limit with negative boundary condition.
 
\subsubsection{Comparison with another interface repulsion phenomenon}
The disorder-induced repulsion phenomenon  highlighted in the present paper bears some analogy with the entropic repulsion phenomenon observed in the SOS model constrained to remain positive and recently studied in detail by one of the authors \cite{cf:Sos1,cf:Sos2}.
The introduction of disorder has in fact effects that are very similar to those induced by the imposing a positivity constraint to the SOS model: the phase transition is smoothened, it vanishes  like of $(h-h_c)^{\nu}$ with $\nu\ge 2$ approaching the critical point (as opposed to linearly for the model without constraint), and 
the interface is  repelled to  a  distance from level zero 
that diverges in this limit.

While the specific mechanisms that triggers these phenomena are different for the two models, two common ingredients can be identified. Firstly, in both models, the contact set  is well approximated (at least at a heuristic level) by an IID Bernoulli field. Secondly, the optimal value for the contact fraction $p$  is obtained by optimizing the balance between a reward which is proportional to $(h-h_c)p$ and a penalty term which takes the form $p^{\mu}$ for some $\mu>1$. Then one can easily conclude that the optimal  balance between penalties occurs for a contact fraction of order $(h-h_c)^{\frac{1}{\mu-1}}$, 
which therefore yields a critical exponent for the free energy equal to $\frac{\mu}{\mu-1}$. We have $\mu=2$ for the disordered pinning and $\mu\in(1,2)$ for SOS (the specific value depends on the lattice, it is equal to $3/2$ on $\bbZ^2$).

One important difference between the two models lies in the origin of the penalty term.
In the disordered model we study here, this penalty term is produced by the second order term in the Taylor expansion of \eqref{Taylz}.
We have
\begin{equation}
\label{elaleq}
\bbE\left[\log \left(1+p \left( \exp\left(\gb \go -\gl(\gb) +h \right)\right) \right)\right]= p h- \frac{p^2}{2} e^{2h} \Var\left(\exp\left(\gb \go -\gl(\gb)\right) \right)+O(p^3).  \
\end{equation}
This quadratic terms in $p$ indicates by how much Jensen's inequality fails to be an equality, in other words it quantifies by  how much the disorder can fail to self average for a fixed contact fraction: recall (observation right after \eqref{Taylz}) that  
$p$ becomes asymptotically proportional to $h$
when $h\searrow 0$, so the first two terms in the right-hand side of \eqref{elaleq} are competing with each other.
For SOS, the penalty term comes from a rewriting of the model, which transforms the wall constraint into a shift of $h$ and an additional penalty for pairs of neighboring contact points, and a similar \emph{first versus second order} competition arises.

\begin{rem}
 Note that the LFF pinning model in presence of a hard wall
 (studied in \cite{cf:BDZwetting, cf:GLw}) presents a different phenomenology, since the critical exponent changes from $1$ to $\infty$ when the hard wall constraint is introduced.
 Informally the reason why this happens is that in that case 
 the penalty term induced by the hard wall constraint is of the form  $p \sqrt{|\log p|}$ instead of $p^{\mu}$, which results in a much smaller optimal value for $p$.
\end{rem}

\subsubsection{Entropic repulsion and critical disordered pinning }
Entropic repulsion models (like \cite{cf:BDZ, cf:D, cf:DG,cf:LM} for LFF and \cite{cf:CLMST} for SOS) have already entered the discussion and there is of course more than a flavor of a connection between our results and  entropic repulsion phenomena. There is however the substantial difference 
that the repulsion phenomenon we observe is to a height that is finite, and diverges only approaching the critical point.
We believe that a direct connection between the wall repulsion studied  in  \cite{cf:BDZ, cf:D, cf:DG,cf:LM}  can be made with the critical disordered pinning model: we quickly develop this next,  just 
 keeping at a heuristic level. 

The main question is: what is the typical value $u(N)$  for $|\phi(x)|$, for $x$ in the bulk of the box, at the critical point, that is, under the measure $\bP_{N, \go, \gb, 0}$.
%Let us briefly formulate a conjecture for the asymptotic of $u(N)$ based on the connection with the model with a wall.
Clearly Theorem \ref{th:paths} indicates that $\lim_n u(N)=\infty$ and, at an intuitive level, the pinning \emph{strength} 
$\gb \go_x -\gl(\gb)$ is negative, i.e. repulsive, on  average, but of course there is a density of sites that are attracting the interface. 
Let us recall the mechanism which induces entropic repulsion for the the measure $\bP_N( \cdot \, |\, \phi(x)\ge 0$ for every $ x\in \gL_N)$. We let 
 $u_{\mathrm{wall}}(N)$ denote the typical bulk height under this measure (note that a softer potential would lead to a similar heuristics). 
The height  $u_{\mathrm{wall}}(N)$  can be understood as 
the one that balances the two penalties $N^{d-2}u$ that  accounts for shifting the field away from level zero in the bulk and
$N^d\exp(-u_N^2/ (2 \gs_d^2))$ that comes from the penalization coming from hitting the penalized (or forbidden) region:
if we want to minimize the sum of these two quantities, we find that $u_{\mathrm{wall}}(N)$ has to be asymptotically equivalent to the square root of
 ${4 \gs_d^2 \log N}$ (see  \cite{cf:BDZ, cf:D, cf:DG}). In the critical disordered model the penalization coming from the potential is weaker:
it is proportional to  $N^d(\exp(-u_N^2/ (2\gs_d^2)))^2$, as it is strongly suggested by the leading $p^2$ in the Taylor expansion of 
\eqref{Taylz} at $h=0$ (see  \eqref{elaleq}). This leads to a (conjectured) repulsion $u(N)^2 \sim {2 \gs_d^2 \log N}$ for the critical disordered pinning model.

On the other hand, we believe that the square root of
${4\sigma_d^2 \log N}$ corresponds to the typical height for negative values of $h$ both in the homogeneous and disordered case. We mention in relation to this problem the disordered entropic repulsion model studied in \cite{cf:BG}: for the repulsion mechanism of  \cite{cf:BG}, quenched and annealed models have, to leading order, the same behavior.

% {I leave this only for a question. Heuristically, that distance should be found by optimizing the 
% contact fraction: on one hand there is a boundary cost proportional to $ u N^{d-1}$ to be at height $u$, and on the other one, and on the other one there is a penalty which is proportional to $N^d p^2$ for the fact that the environment is non-averaging where $p\approx e^{-\frac{u^2}{2\sigma_d}}$.}

\subsubsection{Organization of the paper}
In Section \ref{sec:UB} and \ref{sec:LB}, we prove quantitative upper bound and lower bounds for the free energy respectively. 
The proof of Theorem \ref{th:paths} is also split into two sections:
Section~\ref{sec:LBH} for (A)   in \eqref{upperandlower} and Section~\ref{sec:UBH} for (B) in \eqref{upperandlower}.

The four sections -- upper and lower bounds on the free energy, lower and upper bound on the height of the field --  are almost completely  independent from the technical viewpoint.

% While the proof of Theorem \ref{th:paths} relies on the quantitative bound obtained for the free energy, The sections are mostly independent.

\medskip

 \noindent{\bf Acknowledgements:} This work has been performed in part while G.G.\ was visiting IMPA with the support of the Franco-Brazilian network in mathematics. G.G. also acknowledges support from grant ANR-15-CE40-0020.
 H.L.
 acknowledges support from a productivity grant from CNPq and a Jovem Cientísta do Nosso Estado grant from FAPERJ.

\section{Proof of Theorem \ref{th:mainF}: Upper bound on the free energy}
\label{sec:UB}

Let us first present the quantitative bound proved in this section. As anticipated in the introduction, for the results  in this section   we can sensibly weaken the assumptions on $\go$.  

\medskip

\begin{rem}
\label{rem:weakcond}
The main result of this section, that is  Proposition~\ref{th:final-ub}, holds (and it is proven) assuming less than 
  \eqref{eq:lambda}. More precisely for Proposition~\ref{th:final-ub} we only require  \eqref{eq:lambda} for $s>0$ on  the upper tail of the disorder. For the lower tail
we  assume  that $\bbE[\go_-]< \infty$, with the notation $\go_- = - \go \ind_{\{\go <0\}}$ for the negative part. An analogous standard notation for the positive part is used below.  Like before, we keep the convention 
that  $\bbE [\go]=0$.  
\end{rem}

\medskip

 \begin{proposition}
 \label{th:final-ub}
  For every $\gb>0$ there exists a constant $C_{\gb}$ such that for every $h\in [0,1]$
  \begin{equation}
  \label{eq:final-ub}
    \tf(\gb,h)\, \le\,  \chi(\gb)h^2 +C_{\gb} h^3\, .
  \end{equation}
 \end{proposition}

\medskip

To achieve this bound, we do not rely at all on the fact that $(\gd_x)_{x\in \gL_N}$ is the contact set of  the LFF.
We instead prove a general statement which says that the averaged $\log$ partition function associated to any point process 
$(\gd_x)_{x\in \gL_N}$ -- or any family of Bernoulli random variables (with arbitrary parameters and correlations) -- is always smaller than the one obtained when the $(\gd_x)_{x\in \gL_N}$ are IID Bernoulli variables  with an optimal density. 

\medskip

\begin{proposition} 
\label{th:ub}
Consider $\gL$ a finite non empty set, $(\xi_x)_{x\in \gL}$ a field of IID  random variables satisfying $\bbP( \xi\ge -1 )=1$ and  $\bbE[ (\log (1+ \xi))_+]<\infty$. Moreover we assume that
$\bP_{\gL}$ is the probability distribution of an arbitrary random vector $(\gd_x)_{x\in \gL}$ on $\{0,1\}^{\gL}$. 
Then we have 
\begin{equation}
\label{thegeneral}
 \bbE \log \bE_{\gL}\left[ \prod_{\{x \ : \ \gd_x=1\}}(1+\xi_x) \right]
 \le |\gL|\max_{p\in[0,1]} \bbE \left[ \log \left(1+p \xi \right) \right]\, ,
\end{equation}
with the convention $ \prod_{\{x \in \emptyset\}}(1+\xi_x) =1$.

\end{proposition}

%{The statement is in fact also valid when $\bbE[ (\log (1+ \xi))_+]=\infty$ (modulo the fact that $\bbE \left[ \log \left(1+p \xi \right))\right]$ might not be defined for $p=1$), but is of little interest since the right-hand side \ of \eqref{thegeneral} is equal to infinity in that case (with the maximum attained at any $p\in(0,1)$).}
% Of course by using the \emph{maximum} in the right-hand side we mean that this maximum exists.
% We have excluded the case $\bbE[ (\log (1+ \xi))_+]=\bbE[ (\log (1+ \xi))_-]=\infty$ because the right-hand side
% of \eqref{thegeneral} is not well defined in this case. But the problem is only at $p=1$ and 
% the statement -- if one excludes $p=1$ --  covers also this case simply because 
% $\bbE[ (\log (1+ \xi))_+]=\bbE[ (\log (1+ \xi))_-]=\infty$ directly implies 
% $ \bbE \left[ \log \left(1+p \xi \right) \right]=\infty$ for every $p \in (0,1)$. 

\medskip

Proposition \ref{th:final-ub} then follows from Proposition \ref{th:ub} by  solving the corresponding optimization problem. This is what we do first.

\medskip

\begin{proof}[Proof of Proposition \ref{th:final-ub}]
 As a consequence of \eqref{thegeneral} for $\xi_x\, :=\, e^{\gb\go_x-\gl(\gb)+h}-1$
 we have 
 \begin{equation} 
  \bbE \log Z_{N, \go, \gb,h} \le |\gL_N|\max_{p\in[0,1]} \bbE \left[ \log \left(1+p \xi \right) \right]\, ,
 \end{equation}
and thus for every $\gb>0$ and every $h$ we have 
 \begin{equation}
 \label{eq:ub}
 \tf(\gb,h)\, \le\,  \max_{p\in[0,1]} \bbE \left[ \log \left(1+p\left(e^{\gb\go-\gl(\gb)+h}-1\right)\right) \right]\, .
 \end{equation}
We are now going to expand   the right-hand side for $h \searrow 0$, and everything we are going to require is that  $\gl(3\gb)<\infty$. We use the fact that the maximum is achieved for some $p$ that we call $p_h$ (which is unique as shown in the proof of Proposition \ref{th:ub} although it is not needed in the argument here). 
 First of all remark that   $\lim_{h \searrow 0}p_h=0$. Indeed if we set $p_0=\limsup_{h \searrow 0} p_h $, using dominated convergence and the definition of $p_h$, we have, for some positive sequence $h_n$ tending to $0$
 \begin{equation}
 \lim_{n\to \infty}\max_{p\in[0,1]} \bbE \left[ \log \left(1+p\left(e^{\gb\go-\gl(\gb)+h_n}-1\right)\right) \right]=\bbE[ \log (1+p_0(e^{\gb \go- \gl(\gb)}-1))]\,.
 \end{equation} 
 The right-hand side  is non-negative (set $p=0$ in the left-hand side ) while by the (strict) Jensen's inequality the left-hand side is 
 strictly negative if $p_0>0$ since $\gb>0$ and the distribution of  $\go$ is non degenerate. Therefore $p_0=0$. 
%  $0$ and, since ,  Jensen's inequality is strict so 
%  $\bbE[ \log (1+p_0(e^{\gb \go- \gl(\gb)}-1)]< 0$. But then the same is true for $h>0$ sufficiently small and this is absurd
%  because from the value at $p=0$, $\max_{p\in[0,1]} \bbE \left[ \log \left(1+p\left(e^{\gb\go-\gl(\gb)+h}-1\right)\right) \right]\ge \log 1 =0$.
%  
%  $\tf(\gb, h)\ge 0$. 

To conclude, 
we analyse the asymptotic behavior $\bbE[ \log (1+p(e^{\gb \go- \gl(\gb)+h}-1))]$ in the limit when $p, h \searrow 0$ via Taylor expansion. %{ avant on faisait une expansion en $p$ et $p_h$ ce qui me parait moins pratique pour raisonner ensuite}
This will allow to determine the asymptotics for $p_h$ and subsquently also that of r.h.s of \eqref{eq:ub}.
We use the elementary bound
\begin{equation}
-(x_-)^3 \, \le \, \log(1+x)-x + \frac 12 x^2 \, \le \, \frac 1 3 x^3,
\end{equation}
where the upper bounds holds for every $x>-1$, while the lower bound holds for $x>-0.8156\ldots$. Since we have assumed that $\gl(3\gb) < \infty$  
we  obtain that  
\begin{equation}
\bbE\left[ \log \left(1+p\left(e^{\gb \go- \gl(\gb)+h}-1\right)\right)\right]- p(e^h-1) +\frac { p^2}2 \left(e^{\gl(2\gb)-2 \gl(\gb)+2h}-2e^h+1 \right) \, =\, O\left(p^3\right)\,,
\end{equation}
and this readily entails ($\chi (\gb)$ is given in \eqref{eq:chi})
\begin{equation}
\label{laleq}
\bbE\left[ \log \left(1+p\left(e^{\gb \go- \gl(\gb)+h}-1\right)\right)\right] 
\, = p h - \frac { p^2}{4 \chi (\gb)}+  \, O\left(p^3\right)+O(p^2 h)
+O(p h^2)
\,.
\end{equation}
We stress that an expression (that depends on $p$, $h$ and $\gb$) is $O(p)$ (or $O(h)$,  etc$\ldots$) 
 means that there exists a constant $C_\gb$ such that its absolute value is bounded by $C_\gb p$ (or by
 $C_\gb h$,  etc$\ldots$) for all  $h$ and $p$ sufficiently small. 
It follows then by simple computations (using the fact that $p_h$ is the maximizer and that it tends to $0$) that $p_h \sim 2h \chi(\gb)$ and by bootstrap (using the fact that the left-hand side \ in \eqref{laleq} is $O(h^3)$) we obtain that  
\begin{equation}
\label{optimiz}
p_h \,= \,2h \chi(\gb)+O(h^2)\, .
\end{equation}
Therefore
\eqref{eq:final-ub} follows and the  proof of Proposition~\ref{th:final-ub} is  complete.
\end{proof}

\medskip

\begin{rem}
\label{rem:forfuture}
For future use let us remark that what we have just proven implies that, with
$\xi_x\, :=\, e^{\gb\go_x-\gl(\gb)+h}-1$, we have that for every $\gb>0$ there exists 
$C_\gb>0$ such that 
 \begin{equation}
\max_{p\in[0,1]} \bbE \left[ \log \left(1+p \xi \right) \right]\, \le \, \chi(\gb) h^2 + C_\gb h^3\, ,
 \end{equation}
 for every $h\in [0,1]$.
\end{rem}

\medskip

\begin{proof}[Proof of Proposition \ref{th:ub}]
Let us first observe that we can restrict to the case when $\bbE[\xi]> 0$. When this is not the case, by Jensen's inequality the left hand side of \eqref{thegeneral} is bounded above by $0$ and thus, considering the case $p=0$, the inequality holds. 

 By analyzing the function   $ p \mapsto \log (1+p\xi)$, which is in particular  strictly concave and smooth for $p \in (0,1)$, 
 one readily establishes also that 
 \begin{equation}
 \label{eq:themap}
 [0,1] \ni p   \mapsto 
  \bbE \left[ \log (1+p\xi ) \right]  \in \bbR\cup\{-\infty\}\, ,
  \end{equation}
  is a strictly concave smooth  function  which is continuous up to the boundary points. More precisely, keeping in mind that 
  $(\log(1+p \xi))_+= \log(1+p \xi_+)$ and $(\log(1+p \xi))_-= -\log(1-p \xi_-)$, we have that the function in \eqref{eq:themap} 
  \smallskip
  
  \begin{itemize}
  \item  
   converges to zero  when $p\searrow 0$. This is a consequence of the Dominated Convergence Theorem: 
   $(\log(1+p \xi))_+ \le (\log(1+\xi))_+$ (recall that $\bbE[(\log(1+\xi))_+]<\infty$)
  and  $(\log(1+p \xi))_-$ is bounded for $p$ away from one;
  \item  converges to its boundary values  when $p\nearrow 1 $: the boundary value is
   finite if   $\bbE[ (\log(1+\xi))_-] < \infty$ and it takes values 
  $-\infty$ otherwise. This follows because $(\log(1+p \xi))_+$ is non decreasing  in $p$ to the limit value $(\log(1+ \xi))_+$ that has bounded expectation and because also $(\log(1+p \xi))_-$ is non decreasing in $p$. 
 \end{itemize}
 \medskip
  Let $p_\star$ denote the (unique) value of $p$ for which the maximum in  \eqref{eq:ub} is attained.

 Let us argue first that $p_\star>0$. In fact, $p_\star=0$ is not possible because  the function \eqref{eq:themap} takes value  zero for $p=0$, but its derivative is equal to $\bbE[\xi/(1+p \xi)]$ which approaches 
 $\bbE \xi >0$ (see the beginning of the proof) for $p\searrow 0$: this follows by separating once again the case of $\xi$ positive, for which we apply the Monotone Convergence Theorem, and $\xi$ negative, for which the integrand is bounded. 
%  Therefore the maximum cannot be reached at $p=0$. 
  
 Suppose now that $p_\star\in (0,1)$ (the case $p_\star=1$ is treated at the end). In this case we  exploit the fact that the first derivative of the map \eqref{eq:themap} is zero at $p_\star$ and we obtain 
 \begin{equation}
 \label{lesconds}
  \bbE \left[ \frac{\xi}{1+p_\star\xi} \right]=0 \ \  \text{ which implies } \ \   \bbE \left[ \frac{1}{1+p_\star\xi} \right]=1\, .
 \end{equation}
Reintroducing the dependence on $x$ we set
\begin{equation}
Y_\gL\, :=\, \prod_{x\in  \gL} (1+p_\star \xi_x)\, .
\end{equation}
With the notation 
\begin{equation}
Z_{\gL, \xi}\, :=\,  \bE_{\gL}\left[ \prod_{\{x \ : \ \gd_x=1\}}(1+\xi_x) \right]\, ,
\end{equation} 
we have
\begin{equation}
\label{eq:sotosayfract}
\begin{split}
 \bbE \log Z_{\gL, \xi} \, &=\,  \bbE \log \left[Z_{\gL, \xi} (Y_\gL)^{-1}\right]+ \bbE \log Y_\gL
  \\ 
 &\le \log \bbE \left[Z_{\gL, \xi} (Y_\gL)^{-1}\right]+ \vert \gL \vert \bbE \log (1+p_\star \xi).
\end{split}
\end{equation}
Hence to conclude it suffices  to show that $\bbE \left[Z_{\gL, \xi} (Y_\gL)^{-1}\right]\le 1$.
To establish this we introduce a new law $\tilde \bbP$ for the disorder 
\begin{equation}
  \frac{\dd \tilde \bbP}{\dd \bbP}(\go)\, :=\,  (Y_\gL(\go))^{-1}\, ,
  \end{equation}
where 
$\bbE[Y_\gL^{-1}]=1$ because of the second equality in  \eqref{lesconds}. 
Under this new probability, the variables $\xi_x$ are still IID and  we have 
from \eqref{lesconds}
\begin{equation}
 \tilde \bbE\left[ 1+\xi_x \right]\, =\, 1\, .
\end{equation}
Hence for this reason we have (recall the convention that a product over an empty set is equal to one)
\begin{equation}
\bbE \left[Z_{\gL, \xi} (Y_\gL)^{-1}\right]\, =\,  \tilde \bbE Z_{\gL, \xi}=  \bE_\gL \tilde \bbE \left[ \prod_{x\in  \gL:\,  \gd_x=1} (1+\xi_x) \right]\, =\,1\, .  
\end{equation}

We are left with the case $p_\star=1$. In this case the derivative of the map \eqref{eq:themap} 
must be positive for every $p\in (0,1)$, that is  $\bbE \left[ (1+p\xi)^{-1} \right]<1$ for 
every $p\in (0,1)$, so $\bbE \left[ (1+\xi)^{-1} \right]\le1$: the continuity for $p\nearrow 1$ is established by splitting the expectations according to $\xi\ge 0$, in this case the integrand is bounded,  and $\xi<0$ for which we can apply the Monotone Convergence Theorem.   We use again
\eqref{eq:sotosayfract}, even if in this case $1/Y_\gL$ is not a probability density. But we can argue directly  that
\begin{multline}
\bbE \left[Z_{\gL, \xi} (Y_\gL)^{-1}\right]\, =\, \bE_\gL \bbE \left[\prod_{x \in  \gL} \frac 1{ (1+\xi_x)}  \prod_{x\in  \gL:\,  \gd_x =1 } (1+\xi_x) \right]\\  =\, \bE_\gL\prod_{x \in  \gL: \,  \gd_x=0 } \bbE\left[ \frac 1{ (1+\xi_x)}\right] \, \le \, 1\, .
\end{multline}
This completes the proof of Proposition~\ref{th:ub}.
 \end{proof}

\section{Proof of Theorem \ref{th:mainF}: Lower bound  on the free energy}

\label{sec:LB}

The lower bound we obtain on the free energy is in a sense less precise than the upper bound since the correction we obtain is $h^{2+\gep}$ instead of $h^3$. %Nevertheless, having a quantitative upper bound here is really important since it will have a crucial role in the proof of Theorem \ref{th:paths}.

\begin{proposition}
\label{lloobb}
Choose $\gb>0$.
There exist $\gep=\gep_d$ and
$C_\gb=C_{\gb,d}>0$ such that for every $h>0$
 \begin{equation}
 \label{eq:lloobb}
  \tf(\gb,h)\, \ge\,   \chi(\gb) h^2 -C_{\gb} h^{2+\gep}\, .
 \end{equation}
\end{proposition}

\medskip

Note that we can assume that $h$ is sufficiently small whenever needed.
Indeed if \eqref{eq:lloobb} holds for $h<h_0=h_0(d, \gb)$ and a constant $C_{\gb}$ then it necessary holds for all $h>0$ with a modified  constant
 $C'_\gb=\max (C_{\gb},  \chi(\gb) h_0^{-\gep})$) (this choice makes the  right-hand side \ in \eqref{eq:lloobb} negative for $h>h_0$).

% in fact assume that \eqref{eq:lloobb} is proven for $h< h_0=h_0(d, \gb)$. If 
% the right-hand side of \eqref{eq:lloobb} evaluated at $h_0$ is non positive,
% then the same expression is negative for larger values of $h$ and, since the left-hans side is non negative,
%  \eqref{eq:lloobb} holds for every $h>0$. On the other hand if the right-hand side of \eqref{eq:lloobb} evaluated at $h_0$ is positive then we can make $C_\gb$ larger so that the right-hand side of \eqref{eq:lloobb} evaluated at $h_0$ is zero (it suffices to set $C_\gb=  \chi(\gb) h_0^{-\gep}$). With this new $C_\gb$ we can apply the first part of the argument and   \eqref{eq:lloobb} holds for every $h>0$.

\medskip

The proof of Proposition~\ref{eq:lloobb} is essentially self-contained except for a result  of super-additivity connected to the existence of 
the free energy that we cite from  \cite{cf:GL}. 
For this result we introduce
\begin{equation}
\tilde \gL_N\, :=\, \lint 1,N \rint^d\,,
\end{equation}
and we let
  $\partial \gL_N:= \gL_N \setminus \mathring{\gL}_N$ denote the internal boundary of $\gL_N$ and set
\begin{equation}
\gd^u_x\, :=\, \ind_{ [u-1,u+1]}(\phi(x))\, .
\end{equation}

\medskip

\begin{proposition}[{\cite[Prop.~4.2]{cf:GL}}]
\label{th:superadd}
For every $\gb>0$,  every $h\in \bbR$ and every $u \in \bbR$  we have 
that 
\begin{equation}
\label{eq:superadd1}
\tf (\gb, h)\, =\, \lim_{N\to \infty} \frac 1{N^d}
\bbE \bE \left[\log \bE \left[e^{\sum_{x\in \tilde \gL_N}(\gb \go_x-\gl(\gb)+h)\gd^u_x } \Big | \,
 \cF^\phi_{\partial \gL_N} \right] \right]\, .
\end{equation} 
Moreover for every $N$
\begin{equation}
\label{eq:superadd2}
\tf (\gb, h)\, \ge \,  \frac 1{N^d}
\bbE \bE \left[\log \bE \left[e^{\sum_{x\in \tilde \gL_N}(\gb \go_x-\gl(\gb)+h)\gd^u_x } \Big | \,
 \cF^\phi_{\partial \gL_N} \right] \right]\, .
\end{equation}
\end{proposition}

\medskip

Proposition~\ref{th:superadd} deserves some discussion. The point is that it introduces a different 
partition function: let us discuss first the case $u=0$. The main difference in this case is that we are not considering  
$0$ boundary conditions, but boundary conditions that are random and that they are sampled from 
a LFF, so they are zero only in some averaged sense (there is also the milder difference that the contacts are only the ones in the box $\tilde \gL_N$ which is slightly smaller than 
$\gL_N$). When we introduce $u\neq 0$ we can think of this new partition function as 
the partition function of the model in which the boundary conditions are not sampled from a LFF with
mean zero, but with mean $-u$: we have then written the partition function by exploiting the continuum symmetry of the LFF and we have  translated 
the region in which the contact potential acts  up by $u$ and reset the boundary mean to zero. %, but it is more intuitive to think of the model as the original one (so contact intervals centered at level zero) and boundary with mean $-u$. 
 Proposition~\ref{eq:superadd1} states two facts:
 \smallskip
 
 \begin{itemize}
 \item [(A)] The free energy associated to this new model coincides with the free energy of the original model, and this regardless of the value of $u$;
 \item [(B)] The free energy  dominates its finite $N$ approximation, if we choose the modified partition function we have just introduced for the  finite $N$ approximation: this is proven in \cite{cf:GL} as a direct consequence of the fact 
 that the logarithm of the modified partition function forms a super-additive sequence. 
 \end{itemize}
\smallskip

\medskip
The proof of Proposition~\ref{lloobb}, which involves several steps, is given in Section \ref{lloobb}. 
Before going through it we provide (in Section \ref{sketech} below) a quick exposition the main underlying ideas, and introduce in Section \ref{sec:splitting} a decomposition of the field which serves as an important technical tool for the proof.

\subsection{Sketch of proof for Proposition~\ref{lloobb}}\label{sketech}

The intuition behind our proof of Proposition \ref{lloobb} comes from the inequality 
\eqref{thegeneral}, which implies that for every choice of $u$ and $N$
\begin{equation}
\label{zineq}
\log \bE \left[e^{\sum_{x\in \tilde \gL_N}(\gb \go_x-\gl(\gb)+h)\gd^u_x } \Big | \,
 \cF^\phi_{\partial \gL_N} \right]\le |\tilde \gL_N|\max_{p\in[0,1]} \bbE \left[ \log \left(1+p [e^{\gb \go-\gl(\gb)+h}-1]\right) \right].
\end{equation}
We observe that this inequality is an equality if 
$(\gd^u_x)_{x\in \bbZ^d}$ is replaced by a field of Bernoulli variable with parameter $p(\gb,h)$ which is the maximizer of the right-hand side  of \eqref{zineq} . Hence our strategy relies on fixing $N$ large so that the influence of the boundary condition vanishes, and the value of $u$ in such a way that $(\gd^u_x)_{x\in \bbZ^d}$ 
resembles an IID Bernoulli field with optimal density. This is achieved by fixing $u=u(\gb, h)$ in such a way so that $\bbE[ \gd^u_x]=2h\chi(\beta)$ (recall that from \eqref{optimiz} this ensures that the density is close to optimal). Moreover, at a heuristical level, when $h\searrow 0$, the dependence between the variables  $(\gd^u_x)_{x\in \tilde \gL_N}$  vanishes: indeed as our fixed density vanishes, we have $u(\gb,h)\to \infty$, and high peaks of the LFF are known to display some asymptotic independence (see \cite{cf:CCH} for an illustration).
Most of the challenge is then to transform this
intuition of asymptotic independence into a quantitative statement.

\medskip

The strategy of proof is the following: we split $\tilde \gL_N$ into smaller boxes of edge length $M$ with $1\ll M\ll N$, and we wish to consider the contribution of each box separately. To do so we write $\phi$ as a sum of ``local'' fields whose compact supports corresponds roughly to a box, plus a negligible rest  (see Proposition \ref{decomp}). It is not possible for the support of the local field to match exactly with boxes and they must display some overlap, but we play with an extra parameter  $1 \ll L \ll M$ to make the total area of overlapping region negligible.

\medskip

Once this decomposition is made, we need to show the following two estimates:
\begin{itemize}
\item[(i)] The contribution per site to the free energy inside the region where there is no overlap is given to leading order in $h\searrow 0$  by
$\chi(\gb)h^2$. This is the content of Lemma \ref{pticube}.
\item[(ii)] The contribution per site to the free energy in regions where the support of different local fields intersect is larger than $-h^{2-\gep}$. While the second estimate might seem very rough, it turns out to be sufficient for our purpose since 
the overlap of the supports of local fields only accounts for a small portion of the box $\tilde \gL_N$.
\end{itemize}

\subsection{A finite range decomposition of the free field}
\label{sec:splitting}

Let us explain in this section our decomposition of $\phi$ into a sum of 
random field supported cubic boxes plus a random field of much smaller amplitude (which \emph{contains} all the long range correlations).

\medskip

Throughout the text we use cube for hyper-cube.
When $d\ge 3$, given $L\ge 1$ ($L$ will later be chosen as a function  of $h$ that diverges in the limit $h\searrow 0$, so  we can think of $L$ as a large integer). 
We choose the support of the local fields to be cubes of edge length $M=L^2+L$ while the length $L$ corresponds to the width of 
the overlap region between the support of two neighboring local fields.

% We let $|\cdot|$ and $\| \cdot\|
% $\vert x \vert= \sqrt{x_1^2+\ldots+ x_d^2}$, 
% $\Vert x\Vert= \max_{j=1, \ldots, d} \vert x_j \vert$ and $\dist (x,A):=\min_{y\in A} \vert x-y\vert$. 
% \medskip

\medskip

\begin{proposition}
\label{decomp}
If $\phi=(\phi(x))_{x\in \bbZ^d}$ is a LFF on $\bbZ^d$,
then one can construct a collection of independent  random fields 
$\{\phi_0,(\phi^{(z)})_{z\in \bbZ^d}\}$ (with non negative covariance entries) which satisfy the following properties
\begin{itemize}
 \item [(i)] We have %the following identity in law
 \begin{equation}
\label{ladeco}
\phi\stackrel{\mathrm{law}}= \phi_0+\sum_{z\in \bbZ^d} \phi^{(z)}\,,
\end{equation}
  \item [(ii)] The  field $\phi_0$ satisfies
  \begin{equation}\label{lioub}
  \sup_x \Var(\phi_0(x))\, = \begin{cases} 
  O\left( L^{-1} (\log L)\right) 
  & \text{ if } d=3,
   \\ O\left( L^{-1}\right) & \text{ if } d\ge  4.
   \end{cases}
  \end{equation}
  \item [(ii)] The fields $\phi^{(z)}$ are  identically distributed up to a lattice translation and they are supported in a box of diameter $M+L$. More precisely we have 
  $\phi^{(z)}(\cdot)\stackrel{\mathrm{law}}{=} \phi^{(0)}(\cdot- Mz)$ for every $z$ and %we have 
  almost surely 
  \begin{equation}
   \supp\left(\phi^{(z)}\right)\, =\, M z+\lint 1,M+L-1 \rint^d.
  \end{equation}
  where    $\supp\left(\phi^{(z)}\right):=\{ x\in \bbZ^d \ : \  \phi^{(z)}(x)\ne0\}$.
 \end{itemize}
\end{proposition}

\medskip

One way to picture the decomposition is thinking that the support of the $z$-local field 
are the integer points in  $(0, M+L)$ translated by $Mz$. The supports of the fields 
$\phi^{(z)}$ and $\phi^{(z')}$, $z\neq z'$, do not overlap if $| z-z' | \ge 2$ (recall that $|\cdot|$ denotes the $l_1$ norm).
If  $| z-z' | =1$ they do overlap, but at most on $(M+L-1)^{d-1} (L-1)$ sites. There are regions
(of size $(L-1)^{d}$ sites: the \emph{corners}) where $2^d$ local fields overlap and this is the maximal number 
of overlapping fields.  

\medskip

\begin{proof}
We first decompose $\phi(x)$ into $(L+1)^d$ independent fields
\begin{equation}
\phi(x)= (L+1)^{-d/2} \sum_{y\in \lint 0,L \rint^d} \varphi^{(y)}(x)\, ,
\end{equation}
 where $\left(\varphi^{(y)}\right)_{y\in \lint 0,L \rint^d}$ are IID infinite volume LFF.
Then we  introduce the grid $\bbH_M$ 
\begin{equation}
\bbH_M\, :=\, \{ x\in \bbZ^d : \, \textrm{ther exists }  i\in \lint 1,d \rint  \textrm{ such that } x_i\in M \bbZ \}\, ,
\end{equation}
which splits the lattice $\bbZ^d$ into cubic boxes of edge length $M-1$. Let us also introduce  the translations of $\bbH_M$: 
\begin{equation}
\bbH^{(y)}_M\,:=\, y+\bbH_M\, , 
\end{equation}
which will be used for 
   $y\in \lint 0,L \rint^d$ and  the boxes (with boundaries) delimited by $\bbH^{(y)}_M$ are:
\begin{equation}
 B_M^{(y,z)}\, :=\, \prod_{i=1}^d\lint y_i+Mz_i,y_i+M(z_i+1)\rint \quad \text{ for }  y\in \lint 0
, L \rint^d  \text{ and } \ z \in \bbZ^d\,.
\end{equation}
 We let   $H^{(y)}$ be the harmonic extension of the restriction of $\varphi^{(y)}$ to $\bbH_M^{(y)}$, 
 which is the solution of the system
 \begin{equation}
  \begin{cases}
   \left(\gD H^{(y)}\right)(x)=0 & \text{ for } x\in \bbZ^d \setminus \bbH^{(y)}_M\, ,\\
    H^{(y)}(x)=\varphi^{(y)}(x) &  \text{ for } x\in \bbH^{(y)}_M\, .
  \end{cases}
 \end{equation}
Recall that $H^{(y)}(\cdot)$ is the conditional expectation of  $\varphi^{(y)}$  knowing its value on $\bbH^{(y)}_M$, that is
\begin{equation}
\bbE\left[ \varphi^{(y)}(x) \, \bigg| \, \cF^{\gp^{(y)}}_{\bbH^{(y)}_M} \right]\, =\,  H^{(y)}(x)\, . 
\end{equation}
Now we define the fields
\begin{equation}
 \psi^{(y,z)}(x)\, :=\,  \left(\varphi^{(y)}(x)-H^{(y)}(x)\right)\ind_{B_M^{(y,z)}}(x)\, ,
\end{equation}
and it follows from the spatial Markov property that, for every $y$, the random fields $\{\psi^{(y,z)}\}_{z\in \bbZ^d}$ are independent, and 
$\psi^{(y,z)}$ is  a free field on  $B_M^{(y,z)}$ with zero boundary conditions.
% 
% Now if we let $(G^{(y,z)})_{z\in \bbZ^d}$ be independent free fields with zero boundary condition defined on 
% the box \end{equation}B^{(y,z)}=\prod_{i=1}^d\lint y_i+Mz_i,y_i+M(z_i+1)\rint\end{equation},
% we have the following identity in distribution
% \end{equation}\varphi^{(y)}= H^{(y)}+\sum_{z\in \bbZ^d}G^{(y,z)}\end{equation}
% where $G^{(y,z)}$ is extended to $\bbZ^d$ by setting it equal to zero outside of $B^{(y,z)}$.
We are now ready to make the fields $\phi_0$ and $\phi^{(z)}$ explicit:
\begin{equation}
\label{eq:phi0phiy}
 \begin{split}
\phi_0(x)&:=(L+1)^{-d/2}\sum_{y\in \lint 0,L \rint^d} H^{(y)}(x),\\
  \phi^{(z)}(x)&:= (L+1)^{-d/2}\sum_{y\in \lint 0,L \rint^d} \psi^{(y,z)}(x).
 \end{split}
\end{equation}
The support property of $\phi^{(z)}$ is evident, so 
it remains to show that \eqref{lioub} holds.
This is a consequence of the bound (Lemma~\ref{th:var-est0lem} for a proof)
\begin{equation}
\label{eq:var-est0}
\Var \left(H^{(y)}(x)\right)\, \le\, C_d \left(\left(\dist\left(x,\bbH^{(y)}\right)+1\right)^{2-d}\right)\,.
\end{equation}
In fact since $\left(\varphi^{(y)}\right)_{y \in  \lint 0,L \rint^d}$ is a family of independent fields, from \eqref{eq:var-est0} we have
\begin{equation}
\label{eq:abnd}
\begin{split}
\Var \left(\phi_0(x) \right)\, &=\, \frac 1{(L+1)^d} \sum_{y \in  \lint 0,L \rint^d} \Var \left( H^{(y)}(x)\right)
\\ 
&\le \,   \frac {C_d}{(L+1)^d} \sum_{y \in  \lint 0,L \rint^d} \frac 1{\left(\dist\left(x,\bbH^{(y)}\right)+1\right)^{d-2}}\, ,
\end{split}
\end{equation}
and simple symmetry arguments, assuming $L>1$, show that the last expression is bounded by the case in which
$-x_j$ is equal to the  (upper or lower) integer part of $L/2$ for every $j$:
this corresponds to  summing on $y$ over a cube of edge length $L+1$ centered, up to parity issues, at $x$. We now assume $L$ even to simplify the notations: the last expression in \eqref{eq:abnd} is therefore bounded for every $x$ by
\begin{equation}
\begin{split}
 \frac {2^d C_d}{(L+1)^d} \sum_{y \in  \lint 0,L/2 \rint^d} \frac 1{\left(\dist\left(0,\bbH^{(y)}\right)+1\right)^{d-2}}
 \, &\le \,  \frac {C_d 2^{2d}}{(L+1)^d} \sumtwo{y \in  \lint 0,L/2 \rint^d:}{y_1\le y_2 \le \ldots \le y_d}
 \frac 1{\left(y_1+1\right)^{d-2}}
 \\
 & \le \, \frac {C_d 2^{2d} \left( \tfrac L 2 +1\right)^{d-1}}{(L+1)^d} \sum_{y=0}^{L/2}
 \frac 1{\left(y+1\right)^{d-2}}\, ,
 \end{split}
\end{equation}
which, separating the cases $d=3$ and $d>3$, directly yields the desired estimate. This completes the proof of Proposition~\ref{decomp}.
\end{proof}

\subsection{Proof of Proposition~\ref{lloobb}}

\subsubsection{Step 1: Choice of the finite size parameters}
Recall that by Proposition~\ref{th:superadd} it suffices to show that there exists $h_0>0$ such that for every $h\in (0, h_0)$ there exist $N$ and $u$ such that 
\begin{equation}
\label{toestim}
 \bbE \bE \left[\log \bE \left[e^{\sum_{x\in \tilde \gL_N}(\gb \go_x-\gl(\gb)+h)\gd^u_x } \Big | \,
 \cF^\phi_{\partial \gL_N} \right] \right]
 \, \ge\, N^d\left(\chi(\gb)h^2 -C_{\gb} h^{2+\gep}\right)\, .
\end{equation}
So $N$ may (and will) depend on $h$ as well as $u$. In view of exploiting the finite range decomposition  of Proposition \ref{decomp}
we introduce also an $h$ dependent quantity 
 $L$ and recall that $M=L^2+L$:
\begin{equation}
\label{parami}
 L\, :=\, h^{-\frac{1}{\kappa_d}}, \quad k\, =\, h^{-10}, \quad \text{and } N\, =\, k M +L \, =\,  k L^2+(k+1)L\,
 \stackrel{h\searrow 0}\sim\, {h^{-10-\tfrac 2{\kappa_d}}} \, ,
\end{equation}
where $\kappa_d \ge 4d$ is a positive integer that depends only on  $d$ (the choice is made just after \eqref{eq:choosekappa} below).
We drop integer parts in the notation for the sake of readability.
Finally  we fix $u_h$ (we will often omit the subscript for better readability) to be the unique positive solution $u$ to the equation
\begin{equation}
\label{defu}
\bP\left( \gs_d \cN \in [u-1,u+1]\right) 
%\, =\, \frac1{\sqrt{2 \pi \gs_d^2}} \int_{u_h-1}^{u_h+1} e^{-\frac{v^2}{2\sigma^2_d}}\dd v
\, = \, 2h\chi(\gb)\, .%=:\, p_{\gb}\, . 
\end{equation}
% We could have replaced $2h\chi(\gb)$ by $p_d(\gb,h)$ but using a first order approximation (cf. \eqref{optimiz}) is sufficient for our purpose.
Note that the left-hand side of \eqref{defu} decreases from  $\bP( \gs_d \cN \in [-1,+1])$ to zero when $u$ goes from $0$ to $\infty$. Hence a unique solution $u_h$ to \eqref{defu}
exists if (and only if) $h \in [0,  \bP( \gs_d \cN \in [-1,+1])/(2\chi(\gb))]$ (which we assume). Lemma~\ref{th:uh} below provides a sharp asymptotic expression for $u_h$ along with a useful technical estimate.
\medskip

\noindent We need a preliminary notation: for  $a>0$ and $h>0$ small we set
\begin{equation}
\label{eq:u}
 u(a,h)\,:=\, \gs_d \sqrt{2 \log(1/h)}+1- \frac{\gs_d}2 \frac{\log  \log(1/h) }{\sqrt{2 \log(1/h)}}- 
\gs_d  \frac{\log \left(2a \sqrt{\pi}\right)}{\sqrt{2 \log(1/h)}}\, .
\end{equation}

\medskip

\begin{lemma}
\label{th:uh}
For $h \searrow 0$
\begin{equation}
\label{eq:uh}
u_h \,= \,  u\left( 2\chi(\gb), h\right) + o \left(1/\sqrt{\log (1/h)}
\right)\, ,
\end{equation}
and if $0 \le r(h)=o(1/\sqrt{\log (1/h)})$ then both for $I_h= [ 0, r(h)]$ and
$I_h= [-r(h), 0]$ we have
\begin{equation}
\label{eq:uh2}
\bP\left( \gs_d \cN
- u_h+1 \in I_h
 \right)  \, =\, 
\frac{2 \chi(\gb)}{\gs _d^2} \, u_h h\,  r(h) (1+o(1))\, .
\end{equation}
On the other hand, for every choice of two positive constants $c$ and $C$ we have for $h$ sufficiently small
\begin{equation}
\label{eq:uh2b}
\bP\left( \gs_d \cN \ge 
 u_h-1+c 
 \right)  \, \le \, h \left( \log (1/h)\right)^{-C}\, .
 \end{equation}  
\end{lemma}

\noindent
\emph{Proof of Lemma~\ref{th:uh}}.
Everything is based on  the well known asymptotic ($x \nearrow \infty$) estimate 
\begin{equation}
\label{eq:asymptZ}
\bP(\cN>x) \,=\, \frac 1{x\sqrt{2\pi}} 
\exp\left(-\frac{x^2}{2}\right) \left( 1+ O\left( \frac1{x^2}\right) \right)\, ,
\end{equation}
which in particular implies (via a relatively cumbersome computation) 
 that
\begin{equation}
\label{eq:uah}
\bP\left( u(a,h)+\gs_d \cN \in [-1,1]\right)\stackrel{h \searrow 0} \sim  ah \, , 
\end{equation}
and we point out that the result is the same if $[u_h-1,u_h+1]$ is replaced by 
$[u_h-1,u_h-1+c]$, any $c>0$: that is, the contribution to the asymptotic behavior is all near 
$u_h-1$.
Using \eqref{eq:uah} together with \eqref{defu} we readily extract  \eqref{eq:uh}.

At this point it is rather straightforward to realize that \eqref{eq:uh2b} holds (this is 
is just a  quantitative version of 
the 
observation that we just made that  the contribution to the asymptotic behavior is all near 
$u_h-1$): the leading effect generated by a shift by $c$ in the Gaussian term  is $\exp(-(2c/\gs_d) \sqrt{2 \log(1/h)})$,
that is a factor that vanishes faster than any power of $1/ \log (1/h)$ and this is the content of \eqref{eq:uh2b}.

The estimate for 
\eqref{eq:uh2} requires more care.    For $I_h= [-r(h), 0]$ (the argument  for  $I_h= [ 0, r(h)]$ is essentially the same) we have
\begin{equation}
\label{eq:uh2.1}
\bP\left( \gs_d \cN -u_h+1 \in [-r(h),0]\right)\, =
\,
\int_{u_h-1-r(h)}^{u_h-1} g_{\gs_d}(z) \dd z \, ,
\end{equation}
with $g_\gs(\cdot)$ the density of $\gs \cN$. 
For $r(h)=o(1/\sqrt{\log (1/h)})$ one directly verifies that 
\begin{equation}
\lim_{h \searrow 0} \sup_{z:\, \vert z-u_h+1 \vert \le r (h)} \left \vert \frac{g_{\gs_d}(z)} {g_{\gs_d}(u_h-1)} \, - \, 1 \,\right \vert \, =0 \, ,
\end{equation}
 so that
\begin{equation}
\int_{u_h-1-r(h)}^{u_h-1} g_{\gs_d}(z) \dd z \, = \, r(h) g_{\gs_d}(u_h-1)(1+o(1)).
\end{equation}
Using the asymptotic equivalence \eqref{eq:asymptZ} and \eqref{defu} we have, when $h\searrow 0$
\begin{equation}
g_{\gs_d}(u_h-1)= \frac{u_h}{\sigma^2_d} \bP \left(  \gs_d \cN  \in [u_h-1,u_h+1]\right)(1+o(1))
=\,   \frac{2 \chi (\gb) h u_h}{\gs_d^2}(1+o(1)).\, 
\end{equation}
Then we can conclude that \eqref{eq:uh2} holds exploiting also the asymptotic expression 
\eqref{eq:uh}
for $u_h$. The proof of Lemma~\ref{th:uh} is therefore complete.
\qed

\subsubsection{Step 2: field decomposition and boundary control estimate}
Let us use the decomposition \eqref{ladeco} of Proposition \ref{decomp} for the LFF. Using the information we have concerning the support of $\phi^{(z)}$, we see that
the value of   $\phi(x)$ for $x \in \partial \gL_N$ is not affected by the realization of $(\phi^{(z)})_{z\in \bbZ^d \setminus \lint 0,k-1 \rint^d}$.
Hence letting 
\begin{equation}
\bP^0_N\,,\  \bP^1_N \textrm{ and } \bP^2_N\, ,
\end{equation}
 denote, respectively, the distribution of 
 \begin{equation}
 \phi_0\,, \   \  \left(\phi^{(z)}\right)_{z\in \lint 0,k-1 \rint^d}
\textrm{ and }  (\phi^{(z)})_{z\in \bbZ^d\setminus \lint 0,k-1 \rint^d}\, ,
\end{equation}
 we obtain from Jensen's inequality that
\begin{multline}
\label{toestima}
 \bE \left[\log \bE \left[e^{\sum_{x\in \tilde \gL_N}(\gb \go_x-\gl(\gb)+h)\gd^u_x } \, \Big | \,\cF^\phi_{ \partial \gL_N} \right]\right]\, \ge \\
  \bE^0_N \otimes \bE^2_N \left[ \log \bE^1_N \left[e^{\sum_{x\in \tilde \gL_N}(\gb \go_x-\gl(\gb)+h)\gd^u_x }\right]\right].
\end{multline}

\begin{figure}[htbp]
\centering
\includegraphics[width=14.5 cm]{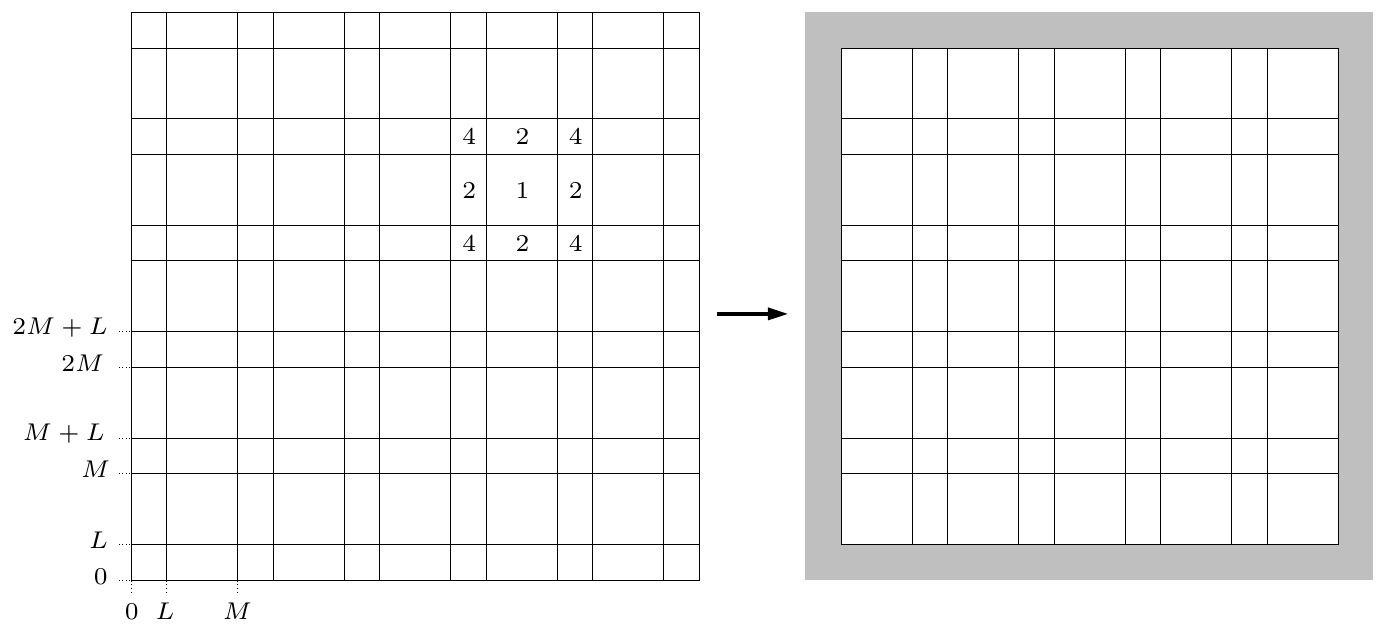}
\vskip-.2cm
\caption{\label{fig:step2} 
On the left there is the set $\gL_N$, covered by partially superposed boxes of edge length $M+L$ (the picture is drawn for $d=2$ just for the sake of visualization purpose: $d2$  is not a relevant dimension for our  problem ). One of this boxes is highlighted by the numbers $1$, $2$ and $4$: one (and only one) of the  fields $\phi^{(z)}$ is supported, for a well chosen $z=z_{\star}\in \bbZ^d$, exactly on this box. On the region marked by $1$, of edge length $M-L$,  only one field
$\phi^{(z)}$ is non zero, that is $\phi^{(z)}(x)=0$ for every $z \neq z_{\star}$ and $x$ in the region marked by $1$. In the regions, belonging to the frame of region $1$, marked by $2$ and $4$, there are respectively $2$ and $4$ values of $z$ (one is $z_\star$) for which the  
$\phi^{(z)}$ fields are non zero. Of course this frame is much smaller than the box ($M \sim L^2$). moreover in dimension $d$ there are regions in which up to $2^d$ fields superpose. Step 2 in the proof consists in
showing that we can erase the contacts in the shadowed frame of the large box $\gL_N$ at little price: we will see that in the end this fully deals with the boundary effects because the long range correlations of the field are all carried by the field $\phi_0$ and we get rid of this field in Step 3. 
}
\end{figure}

\medskip

\noindent Now let us show that one can replace $\tilde \gL_N$ in \eqref{toestima} by a smaller set in order to avoid boundary effects.
Set 
\begin{equation}
\gL'_N=\lint L+1, kM \rint^d \quad \text{ and } \quad \gG_N:= \tilde \gL_N\setminus \gL'_N\, .
\end{equation}
We are going to prove:

\medskip

\begin{lemma}
\label{th:gLprime}
 For $h\ge 0$ we have $\bP^0_N \otimes \bP^2_N$ a.s.\
\begin{equation}
\label{lok}
\bbE \log \bE^1_N \left[e^{\sum_{x\in \tilde \gL_N}(\gb \go_x-\gl(\gb)+h)\gd^u_x } \right]\, 
\ge\, \bbE  \log \bE^1_N \left[e^{\sum_{x\in \gL'_N}(\gb \go_x-\gl(\gb)+h)\gd^u_x } \right]-\gl(\gb)|\gG_N|\, .
\end{equation}
\end{lemma}
\medskip 

Since the cardinality of $\gG_N$ is $O(N^{d-1}L)$, our choice of parameter makes the last term negligible with respect to $N^d h^{2+\gep}$, and thus 
combining \eqref{toestima} and \eqref{lok}, the inequality \eqref{toestim} follows if $\gep <1/\kappa_d$ (which we therefore assume) and if we show that  
\begin{equation}
\label{toestim2}
 \bbE \bE^0_N \log \bE^1_N \left[e^{\sum_{x\in \gL'_N}(\gb \go_x-\gl(\gb)+h)\gd^u_x } \right]
\,  \ge \, N^d\left(\chi(\gb)h^2 -C_{\gb} h^{2+\gep}\right).
\end{equation}
Note that the expectation w.r.t.\ to $\bP^2_N$ is not displayed because $\phi$ restricted to $\gL'_N$ is completely determined by $\phi_0$ and 
$\left(\phi^{(z)}\right)_{z\in \lint 0,k-1 \rint^d}$.

\medskip

\noindent
\emph{Proof of Lemma~\ref{th:gLprime}}.
 We let $\mu^{A}_N$ be the probability measure defined by 
\begin{equation}
 \frac{\dd \mu^{A}_N}{\dd \bP^1_N}\left( (\phi^{(z)})_{z\in \lint 0,k-1\rint}\right)= 
 \frac{e^{\sum_{x\in A}(\gb \go_x-\gl(\gb)+h)\gd^u_x }}{\bE^1_N \left[e^{\sum_{x\in A}(\gb \go_x-\gl(\gb)+h)\gd^u_x }\right]},
\end{equation}
that is the distribution of $(\phi^{(z)})_{z\in \lint 0,k-1\rint}$ when interactions with sites in $A$ are taken into account
(this measure depends on the realization of $\phi_0$ and possibly also of that of $\phi^{(z)}$ for $z\notin \lint 0,k-1\rint$).
We have
\begin{multline}
 \log \bE^1_N \left[e^{\sum_{x\in \tilde \gL_N}(\gb \go_x-\gl(\gb)+h)\gd^u_x } \right]\\
 = \log \bE^1_N \left[e^{\sum_{x\in \gL'_N}(\gb \go_x-\gl(\gb)+h)\gd^u_x }  \right]
 +\log \mu^{\gL'_N}_N \left[e^{\sum_{x\in \gG_N}(\gb \go_x-\gl(\gb)+h)\gd^u_x } \right],
\end{multline}
And thus \eqref{lok} follows by Jensen's inequality as follows:
\begin{equation}
\label{lajens}
 \bbE \log \mu^{\gL'_N}_N \left[e^{\sum_{x\in \gG_N}(\gb \go_x-\gl(\gb)+h)\gd^u_x } \right]
 \,\ge\,   \bbE  \mu^{\gL'_N}_N \left[ \sum_{x\in \gG_N}(\gb \go_x-\gl(\gb)+h)\gd^u_x   \right]
 \,\ge \,   -\gl(\gb)|\gG_N|\, ,
\end{equation}
and in the last step we have used the fact that $ \mu^{\gL'_N}_N $ does not depend
on $\{\go_x\}_{x \in \Gamma_N}$ and then we have just used that $\sum_{x\in \gG_N} \gd^u_x \le 
|\gG_N|$. The proof of Lemma~\ref{th:gLprime} is complete.
\qed

\subsubsection{Step 3: getting rid of the base field $\phi_0$}
Our next step is to get rid of the dependence in $\phi_0$ in \eqref{toestim2}. We can do this combining two facts. Firstly from \eqref{lioub} and our choice of the parameters we know that with an overwhelming large probability $\phi_0$ is small everywhere. 
Secondly, from \eqref{eq:uh2}, we know that the expectation $\bE^1_N[\gd^u_x]$ for $x\in \gd^u_x$ is not much affected by small variations of $\phi_0$.
 Let us define the event $\cC_N$ by
\begin{equation}
\label{eq:cC}
\cC_N\, :=\, \left\{   \left|\phi_0(x)\right|\,\le\,  L^{-1/3}
\text{ for every } x\in \gL'_N \right\}\, .
\end{equation}
A simple union bound using the estimate on the variance \eqref{lioub} 
yields immediately for $h$ small 
\begin{equation}
\label{labrob}
 \bP^0_N \left( \cC^{\cc}_N\right)\, \le\,  \exp\left(-L^{1/4}\right)\, .
\end{equation}
By applying Jensen's inequality  we have also that  for every realization of $\phi_0$
\begin{equation}
\label{eq:prepretoestim3}
\bbE \log \bE^1_N \left[e^{\sum_{x\in \gL'_N}(\gb \go_x-\gl(\gb)+h)\gd^u_x } \right]\, \ge\, 
 \bE^1_N \bbE\left[{\sum_{x\in \gL'_N}(\gb \go_x-\gl(\gb)+h)\gd^u_x } \right]
\, \ge  \,  -|\gL'_N|\gl(\gb),
\end{equation}
and therefore
\begin{multline}
\label{eq:pretoestim3}
\bbE \bE^0 \log \bE^1_N \left[e^{\sum_{x\in  \gL'_N}(\gb \go_x-\gl(\gb)+h)\gd^u_x } \right]\\
\ge \bbE \bE^0_N \left[ \log \bE^1_N \left[e^{\sum_{x\in \gL'_N}(\gb \go_x-\gl(\gb)+h)\gd^u_x } \right]\ind_{\cC_N}\right]-\gl(\gb) 
N^d\, \bP^0_N\left(\cC^{\cc}_N\right)\, .
\end{multline}
Now recalling \eqref{labrob} and \eqref{parami}, this implies that \eqref{toestim2} follows if one proves 
that for every $\phi_0\in \cC_N$
\begin{equation}
\label{toestim3}
 \bbE \log \bE^1_N \left[e^{\sum_{x\in \gL'_N}(\gb \go_x-\gl(\gb)+h)\gd^u_x } \right]
\,  \ge \, N^d\left(\chi(\gb)h^2 -C_{\gb} h^{2+\gep}\right)\, .
\end{equation}

Note that expressions like for example  the leftmost sides of \eqref{eq:prepretoestim3}--\eqref{toestim3}
 are now  random variables: they are measurable with respect to 
the $\gs$-field generated by $\phi_0$. 
We are going to see that 
 not averaging over $\phi_0$ is typically not a problem, at least as long as $\phi_0\in \cC_N$.
The key result in this direction says that the contact density 
is not very much affected by conditioning with respect to $\phi_0$, as long as $\phi_0\in \cC_N$.
It will be repeatedly used in the remainder. Here it is:

\medskip 

\begin{lemma}
\label{th:pertinho}
For every $\phi_0\in \cC_N$,  $x\in \gL'_N$ and $h$ sufficiently small we have 
\begin{equation}
\label{pertinho}
 \left|\bE^1_N[\gd^u_x]- 2\chi(\gb)h\right|\, \le\,    h L^{-1/4}\, .
\end{equation}
\end{lemma}
\medskip

\medskip

\noindent
\emph{Proof of Lemma~\ref{th:pertinho}}.
We introduce the practical notation 
\begin{equation}
\label{eq:phi1}
\phi_1(x)\, :=\, \sum_{z\in \lint 0,k-1\rint^d} \phi^{(z)}(x)\, 
\end{equation}
 (note that the summation defining $\phi_1(x)$ contains between one and $2^d$ non-zero terms),
and 
for  $x\in \gL'_N$ we have that
\begin{equation}\label{lapro}
\begin{split}
\bE^1_N\left[\gd^u_x\right]\, &=\, 
\bP^1_N\left( \phi_1(x) \in  [u-1-\phi_0(x),u+1-\phi_0(x)]\right)
\\
&=\, \bP\left( \sqrt{\gs_d^2- \Var \left( \phi_0(x)\right) } \, \cN \in {[u-1-\phi_0(x),u+1-\phi_0(x)]}\right)\, ,
\end{split} 
 \end{equation}
 where $\bP$ in the last line is just the law of $\cN\sim \cN(0,1)$, which is of course the only random variable in the expression. 
Now, for small $h$  and assuming that $\cC_N$ holds, by monotonicity  the right-hand side \ in \eqref{lapro} is maximized when $\phi_0(x)=L^{1/3}$ and minimized when $\phi_0(x)=-L^{1/3}$.
 If we set $\tilde \sigma:= \sqrt{\gs_d^2- \Var \left( \phi_0(x)\right)}$,  \eqref{lioub} (and the fact that $u$ is asymptotically proportional to   $\log L$) implies that 
for $h$ sufficiently small we have
\begin{equation*}\begin{split}
               \frac{\sigma_d}{\tilde \sigma}[u-1+ L^{-1/3},u+1+L^{1/3}] \supset [u-1+ 2L^{-1/3},u+1],\\
                 \frac{\sigma_d}{\tilde \sigma}[u-1- L^{-1/3},u+1-L^{1/3}] \subset [u-1-L^{-1/3},u+1],
        \end{split}
\end{equation*}
so that
\begin{equation}
 \bP\left( \gs_d\,  \cN \in [u-1+ 2L^{-1/3},u+1]\right)\, \le \,
\bE^1_N\left[\gd^u_x\right]\, \le \,  \bP\left( \gs_d\,  \cN \in {[u-1- L^{-1/3},u+1]}\right)\, .
\end{equation}
Now, using \eqref{eq:uh2} we obtain that 
\begin{equation}
\label{eq:forpertinho}
 \left|\bE^1_N[\gd^u_x]- 2\chi(\gb)h\right|\, \le\,    h
 \bP\left( \gs_d\,  \cN \in {[u-1- L^{-1/3},u-1+ 2L^{-1/3}]}\right) \, \le \, C h u_h   L^{-1/3}
 \, ,
\end{equation}
 with $C=C(\gb, d)$ which can be easily read out  from \eqref{eq:uh2}.
 The proof of Lemma~\ref{th:pertinho} is therefore complete.
 \qed

\subsubsection{Step 4: reducing to estimates on $M$-boxes.}
We are now going to state two technical results. The first (Lemma \ref{pticube}) is   
 \eqref{toestim3}, but with $\gL'_N$ replaced by the set of vertices $x$ which are in the support of a unique $\phi^{(z)}$
 \begin{equation}\gL''_N:= \left\{ x\in \gL'_N \ : \ \exists z, \ \phi_1(x)=\phi^{(z)}(x),\  \bP^1_N-\text{a.s.} \right\}. \end{equation}
 Note that our condition $L\ll M$ ensures that  $\vert \gL''_N\vert \sim\vert \gL'_N\vert$. The reason why the quantity with $\gL''_N$ is easier to handle is that by independence of the $\phi^{(z)}$, the partition function can be factorized (and its $\log$ becomes simply a sum).
 For technical reason, we prove in fact \eqref{toestim3} with $\gL''_N$ but also with an additional constraint of the fields $\phi^{(z)}$.
 
 The second result (Lemma \ref{ctunlemme}) states  that the portion  of the domain that we leave out, i.e. $\gL'_N\setminus \gL''_N$, gives a negligible contribution, this is where the constraint %technical restriction
 imposed in the statement of the first result are used. 

\medskip

Both results are needed for our proof, but while the second is of a technical nature, the first contains the key second moment argument on which  the proof relies. 
\begin{figure}[htbp]
\centering
\includegraphics[width=13 cm]{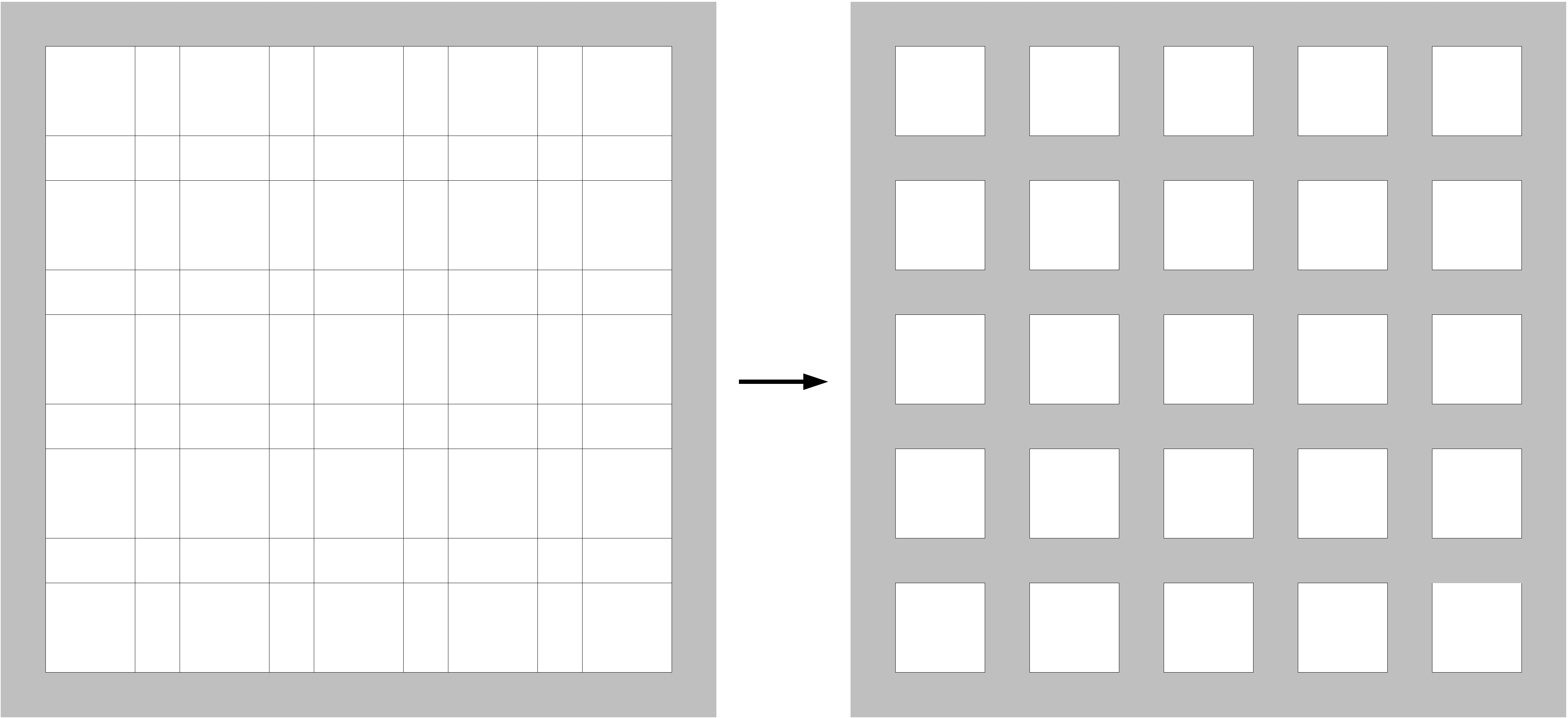}
\vskip-.2cm
\caption{\label{fig:step4} 
In Step 4 to Step 6 we show that it is sufficient to consider the system in which the contacts in the shadowed 
grid-like region are erased. The independence properties of the $\phi^{(z)}$ fields reduce the problem to estimate
the contributions to the free energy in each of the non shadowed boxes of edge length $M-L$. One important ingredient is also the fact that, for every such a box, we introduce a constraint on the $\phi^{(z)}$ fields whose support intersects the box or its frame (the box with the frame is a box of edgelenth $M+L$). This constraint limits the number of contacts that can take place in this box and its frame. A second moment argument turns out to be very efficient thanks to this control on the number of contacts.   
}
\end{figure}
 Let us observe that $\gL''_N$ is a disjoint union of cubes of diameter $M-L$
 \begin{equation}
 \gL''_N\, = \, \bigcup_{z\in \lint 0,k-1 \rint^d} \cB^{(z)}_{L}
 \, ,
 \end{equation}
 where 
 \begin{equation}
 \label{eq:BLz}
 \cB^{(z)}_{L}\, :=\, \lint L+1, M\rint^d  + z M\, .
 \end{equation}
 The space between cubes being much smaller than the cube diameter, 
 $\gL''_N$ covers most of the original box.
 Furthermore let us also introduce  
 \begin{equation}
 \cD_N\, :
 =\, \inter_{z\in \lint 0,k-1\rint^d}\cD^{(z)}_N\ \ \textrm{ with } \ \
 \cD^{(z)}_N\, :=\, \left
\{ \phi^{(z)}:\,\left \vert \left\{ x \, : \, \phi^{(z)}(x)\ge 2^{-d}(u-2)\right\}\right\vert  \le \kappa_d \right\} \,. 
 \end{equation}
where $\kappa_d$ is fixed in such a way that 
\begin{equation}
\label{wop}
1- \bP^1_N \left( \cD^{(z)}_N \right) \, =\, 
\bP^1_N\left(\left \vert\left\{ x \ : \ \phi^{(z)}(x)\ge 2^{-d}(u-2)\right\}\right\vert > \kappa_d\right)\, \le\,  h^2\, .
\end{equation}
Note that $\cD_N$ limits the number of contact points.
Let us show that we can find $\kappa_d$ such that \eqref{wop} holds.

\begin{lemma}
\label{th:pqs4}
For any $c>0$ %$c \in (0, 2^{-d}\gs_d)$ 
we can find $\kappa=\kappa(d, c)\in \bbN$ such that for $h$ small
\begin{equation}
\label{eq:pqs4}
\bP^1_N\left(\left \vert\left\{ x \ : \ \phi^{(z)}(x)\ge c \sqrt{2 \log (1/h)}\right\}\right\vert > \kappa\right)\, \le\,  h^2\, .
\end{equation}
\end{lemma}

\medskip

Of course  \eqref{wop} holds 
by Lemma~\ref{th:pqs4} by choosing a value of
$c \in (0, 2^{-d}\gs_d)$ and then fixing $\kappa_d=\max(\kappa(d, c), 4d)$
(the requirement $\kappa_d \ge 4d$ is due to 
  \eqref{eq:cond2kappa} below). 

\medskip

\begin{proof}
By a union bound the probability on the left-hand side of \eqref{eq:pqs4}  
is not larger than
\begin{multline}
(M+L)^{d\kappa} \max_{x_1, \ldots, x_ {\kappa}} \bP^1_N\left( \phi^{(z)}(x_j)\ge c \sqrt{2 \log (1/h)} 
\textrm{ for } j=1, \ldots, \kappa \right)
\, \le \\
 (2M)^{d\kappa} \max_{x_1, \ldots, x_ {\kappa}} \bP^1_N\left( \sum_{j=1}^{\kappa}  \phi^{(z)}(x_j)\ge c \kappa  \sqrt{2 \log (1/h)}\right)\, ,
\end{multline}
where of course the $x_j$'s are in the support of $ \phi^{(z)}$. We are therefore reduced to estimating the variance 
of $\sum_{j=1}^{\kappa}  \phi^{(z)}(x_j)$. The first observation is that we can replace  $\bP^1_N$ with $\bP$, that is we can work with the free field in $\bbZ^d$ (this is just because $G(x,y)$ is larger than the covariance of $\phi^{(z)}$ for every choice of $x$ and $y$). The next step is realizing that this variance is maximal when the $x_j$ are \emph{closely packed}: more precisely that the maximum of the variance is bounded above by a constant that depends only on the dimension 
time the variance of the case in which the set $\{x_1,\ldots, x_{\kappa}\}$ is  replaced by the cube of edge length equal to $\lceil \kappa^{1/d} \rceil$ (of course this cube contains more than $\kappa$ points, unless $\kappa^{1/d}$ is an integer number: the full argument is left to the reader and can be found for example in 
\cite[Lemma~6.11]{cf:GL} where one can refer also for an explicit constant (that is however largely overestimated). The computation for the cube is straightforward and yields  a variance which behaves for 
$\kappa$ large like $C_d \kappa^{(d+2)/d}$. Therefore
\begin{equation}
\begin{split}
\bP^1_N\left( \sum_{j=1}^{\kappa}  \phi^{(z)}(x_j)\ge c \kappa  \sqrt{2 \log (1/h)}\right) \, &\le \, \bP\left(
\gs_d \cN \, \ge \, c\, C_d^{-1/2}  \kappa^{1-(d+2)/(2d)} \sqrt{2 \log (1/h)}\right) \\
 & \le \, h^{c^2 \kappa^{(d-2)/d}/ (C_d\gs_d^2)}\, =\, h^{c'_d\kappa^{(d-2)/d}}\,.
\end{split}
\end{equation}
Therefore 
\begin{equation}
\label{eq:choosekappa}
\bP^1_N\left(\left \vert\left\{ x \ : \ \phi^{(z)}(x)\ge 2^{-d}(u-2)\right\}\right\vert > \kappa\right)\, \le\, 
 %(2M)^{d\kappa} h^{c'_d\kappa_d^{(d-2)/d}}\,=\, 
 2^{d\kappa}
 h^{-2d+c'_d\kappa^{(d-2)/d}}\, \le \, h^2\, ,
\end{equation}
with the last step that holds if  $\kappa$ is chosen so that
$c'_d\kappa^{(d-2)/d} > 2d+2$ 
 and if $h$ is sufficiently small. Therefore $\kappa(d,c)$ is identified and the  proof of Lemma~\ref{th:pqs4} is 
 complete.
\end{proof}

\medskip

Here is the first of the two results that we announced:
\medskip

\begin{lemma}
\label{pticube}
 For every $\phi_0\in \cC_N$ we have (note that
 $\vert \cB^{(z)}_{L} \vert=(M-L)^d=L^{2d}$)
 \begin{equation}
 \label{eq:pticube}
  \bbE \log \bE^1_N \left[e^{\sum_{x\in \cB_L^{(z)}}(\gb \go_x-\gl(\gb)+h)\gd^u_x }\ind_{\cD^{(z)}_N}\right]\,\ge\,  L^{2d}(\chi(\gb)h^2-C_\gb h^{2+\gep}).
 \end{equation} 
As a consequence for $\phi_0\in \cC_N$ we have 
  \begin{equation}
  \label{eq:pticube2}
  \begin{split}
  \bbE\log \bE^1_N \left[e^{\sum_{x\in \gL''_N}(\gb \go_x-\gl(\gb)+h)\gd^u_x }\ind_{\cD_N}\right]\, 
  &= \sum _{z\in \lint 0,k-1\rint^d}\bbE\log \bE^1_N \left[e^{\sum_{x\in \cB_L^{(z)}}(\gb \go_x-\gl(\gb)+h)\gd^u_x }\ind_{\cD^{(z)}_N}\right]
  \\
  &\ge\, N^d(\chi(\gb) h^2-C_\gb h^{2+\gep})\, .
 \end{split}
 \end{equation}
 \end{lemma}

 Of course this is not yet sufficient to conclude as there is no direct way to show that adding the sites of $\gL'_N\setminus\gL''_N$ has a positive contribution on the free energy.
 
 \medskip
 
We will in fact content ourselves with showing that this contribution is not too negative. This is the object of our second result. It   requires the introduction of some further notation.
 We define for $A\subset \gL'_N$,  $\tilde Z_{A, \go}$ as a partition function restricted to $\cD_N$ for which the interaction is present only for sites in $A$
 \begin{equation}
 \label{eq:ZAgo}
 \tilde Z_{A, \go}\, := \, \bE^1_N \left[e^{\sum_{x\in A}(\gb \go_x-\gl(\gb)+h)\gd^u_x }\ind_{\cD_N}\right]\, .
 \end{equation}
Moreover we define
\begin{equation}
\label{eq:Mh}
K_{h}\, :=\, \inf\{ K\ge \sqrt{\log 1/h}  :\, 
\bbE[ |\go_x| \, ; \, |\go_x|> K]\le h^3\}\,.
\end{equation}
The condition  $K\ge \sqrt{\log 1/h}$ is artificial, but it is convenient for us to have $\lim_{h\searrow 0} K_h= \infty$. Since we have assumed $\bbE[\exp(t \go_x)]< \infty$ for every $t\in \bbR$ we have that
$K_{h}=o( \log(1/h))$. We set
\begin{equation}
\overbar{\go}_x\, :=\, \go_x\ind_{\{ |\go_x|\le K_h\}} .
\end{equation} 
Here is the second result.

 \medskip
 
 \begin{lemma}
 \label{ctunlemme}
For  every $\gga>0$ there exists $h_0$ such that for all $h\in [0,h_0]$ we have
\begin{equation}\label{zladif}
\bbE \log \tilde Z_{B, \overbar{\go}} - \bbE \log \tilde Z_{A,\overbar{\go}}\, \ge\,  - h^{2-\gga} \left \vert B\setminus A\right\vert\, ,
\end{equation}
for every $A$ and $B$ with  $A\subset B\subset \gL'_N$
 \end{lemma}

 \medskip
 
 Let us show that combining Lemma~\ref{pticube} and Lemma \ref{ctunlemme} we obtain \eqref{toestim3}
 (and the proof of Proposition \ref{lloobb} is therefore reduced to proving the two lemmas).
 First of all just by restricting the expectation to the event $\cD_N$ we have
  \begin{equation}
   \label{eq:intcut-off-1}
   \bbE \log \bE^1_N \left[e^{\sum_{x\in \gL'_N}(\gb \go_x-\gl(\gb)+h)\gd^u_x } \right]\, \ge 
     \bbE \log \bE^1_N \left[e^{\sum_{x\in \gL'_N}(\gb \go_x-\gl(\gb)+h)\gd^u_x }  \ind_{ \cD_N}\right]
 \end{equation}
 The next step is to remark  that we can replace $\go_x$ by $\overbar \go_x$ without changing much the expectation. In fact, by applying Lemma~\ref{th:theboundz} with $K=K_h$ we see 
 that for any $A\subset \gL'_N$ we have 
 \begin{equation}\label{theboundz}
     \left|\bbE \left[\log \tilde Z_{A,  \go}\right]-
          \bbE \left[\log \tilde Z_{A,\overbar\go}\right]\right|\le \beta |A|h^3.
 \end{equation}
 In particular, using \eqref{theboundz} for $A=\gL'_N$ we see that it is sufficient to prove \eqref{toestim3} for  
\begin{equation}
          \bbE \left[\log \tilde Z_{\gL'_N,\overbar \go}\right] =  \bbE \left[\log \tilde Z_{\gL''_N,\overbar \go}\right]+
         \bbE \left[\log \tilde Z_{\gL'_N,\overbar \go}- \log \tilde Z_{\gL''_N,\overbar \go} \right]
\end{equation}
We now  use Lemma \ref{pticube} and \eqref{theboundz} to bound the first term and Lemma \ref{ctunlemme} for the second one and 
 %But one can also see\eqref{eq:pticube2}  holds also with $\go$ replaced by $\overbar \go$ (possibly with the need of changing the constant $C_{\beta}$) as a direct consequence of  \eqref{theboundz} with $A=\gL''_N$. Hence using \eqref{eq:pticube2} (for $\overbar \go$) and and \eqref{zladif} in Equation \eqref{eq:intcut-off-1}
  we obtain that for $\gamma>0$ and  $h$ sufficiently small (allowed to depend on $\gamma$)
  \begin{equation}
   \bbE \log \bE^1_N \left[e^{\sum_{x\in \gL'_N}(\gb \go_x-\gl(\gb)+h)\gd^u_x } \right]
   \ge\,  N^d(\chi(\gb) h^2-C_\gb h^{2+\gep}) - |\gL'_N \setminus \gL''_N|  h^{2-\gga}\, .
   \end{equation}
  %   \begin{equation}
%   \begin{split}
%    \bbE \log \bE^1_N \left[e^{\sum_{x\in \gL'_N}(\gb \go_x-\gl(\gb)+h)\gd^u_x } \right]
%    \, & \ge \,  \bbE \log \bE^1_N \left[e^{\sum_{x\in \gL'_N}(\gb \overbar{\go}_x-\gl(\gb)+h)\gd^u_x } \right]
%    \\
%    &\ge\,  \bbE \log \tilde Z_{\gL'_N, \overbar{\go}}
%    \\
%     \, &=\,   \bbE \log \tilde Z_{\gL''_N, \overbar{\go}}
%    +\left(\bbE \log \tilde Z_{\gL'_N, \overbar{\go}}-\bbE \log \tilde Z_{\gL''_N, \overbar{\go}} \right)
%  \\  
%    &\ge  \bbE \log \tilde Z_{\gL''_N, \go} - \gb \left\vert \gL''_N \right\vert h^3 
%   - 2d N^{d-1}L h^{2-\gga} 
%    \\
%    &\ge\,  N^d(\chi(\gb) h^2-C_\gb h^{2+\gep})- \gb N^d h^3  - 3d N^{d} h^{12-\gga+1/\kappa_d}\, ,
%  \end{split}
% %    \end{equation}
%   where the first inequality is an evident monotonicity, the second is \eqref{eq:intcut-off-1}, the one before the last is  
%     Lemma \ref{ctunlemme} and the last is \eqref{eq:pticube2} (with the notation \eqref{eq:ZAgo}.
  Now it suffices to recall the choice of the parameters and to choose $\gamma= 1/2\kappa_d$ and  
  \eqref{parami} allows to  conclude that \eqref{toestim3} holds for some $\gep>0$.
  More precisely we have 
 \begin{equation} |\gL'_N \setminus \gL''_N|\le d (k+1) L N^{d-1} \le 2d N^d L^{-1}= 2d N^d h^{\frac 1 \kappa_d}.\end{equation}
 As announced, this means that we are \emph{just} left with proving Lemma~\ref{pticube} and  Lemma~\ref{ctunlemme}. This is the content of the next two steps of the proof.

 \subsubsection{Step 5: the second moment estimate (proof of  Lemma~\ref{pticube})}

To see that \eqref{eq:pticube2}
 follows from \eqref{eq:pticube}, one simply observes that
once $\phi_0$ is fixed, $(\gd^u_x)_{x\in  \cB_L^{(z)}}$ is determined by $\phi^{(z)}$.
 Hence, since  $\bP^1_N$ is a product measure (recall the support properties of $\phi^{(z)}$ 
 and the definition \eqref{eq:BLz} of $\cB^{(z)}_{L}$), the expectation can be factorized. 
 Applying \eqref{eq:pticube} to each term of the sum thus obtained we have
 $$ \bbE \log \bE^1_N \left[e^{\sum_{x\in \gL'_N}(\gb \go_x-\gl(\gb)+h)\gd^u_x } \right]\ge   |\gL''_N|(\chi(\gb)h^2-C_\gb h^{2+\gep}),$$
 and  we conclude (modifying the value of $C_{\beta}$ if necessary) by using the fact that our choice of parameters implies $|\gL''_N|\ge N^d(1-h^{\gep})$.
 
 \medskip

 Let us now turn to the important part which is the proof of \eqref{eq:pticube}.
Since the result does not depend on $z$ let us assume that $z=0$ and write $\cB$ for $\cB^{(z)}_L$. 
With this notational simplification we remark the splitting:
\begin{multline}\label{eq:splitLB}
  \bbE \log \bE^1_N \left[e^{\sum_{x\in \cB}(\gb \go_x-\gl(\gb)+h)\gd^u_x }\ind_{\cD^{(0)}_N}\right]\,=
  \\
   \bbE \log \frac{\bE^1_N \left[e^{\sum_{x\in \cB}(\gb \go_x-\gl(\gb)+h)\gd^u_x }\ind_{\cD^{(0)}_N}\right]}
   {\bbE\bE^1_N \left[e^{\sum_{x\in \cB}(\gb \go_x-\gl(\gb)+h)\gd^u_x }\ind_{\cD^{(0)}_N}\right]} +
   \log \bbE\bE^1_N \left[e^{\sum_{x\in \cB}(\gb \go_x-\gl(\gb)+h)\gd^u_x }\ind_{\cD^{(0)}_N}\right]
  \\
  =\, 
  \bbE \log \bE^{1,h}_N \left[e^{\sum_{x\in \cB}(\gb \go_x-\gl(\gb))\gd^u_x }\, \Big\vert \, {\cD^{(0)}_N}\right]
  + \log  \bbE\bE^1_N \left[e^{h\sum_{x\in \cB}\gd^u_x }\ind_{\cD^{(0)}_N}\right]\, ,
  \end{multline}
  and in the last step we have introduced the \emph{homogeneous pinning} Gibbs measure  $\bP^{1,h}_N$.
Let us first estimate the last addendum in \eqref{eq:splitLB}. We have  
 \begin{equation}
 \begin{split}
    % \log \bbE \bE^1_N \left[e^{\sum_{x\in \cB}(\gb \go_x-\gl(\gb)+h)\gd^u_x }\ind_{\cD^{({0})}_N}\right]=  
 \log \bE^1_N \left[e^{\sum_{x\in \cB}h\gd^u_x }\ind_{\cD^{({0})}_N}\right]
\,     &\ge \,    h\,  \bE^1_N \left[ \sum_{x\in \cB}\gd^u_x \, \bigg| \, \cD^{({0})}_N \right]+\log \bP^1_N(\cD^{({0})}_N)
\\
     &\ge\,  h\,  \bE^1_N \left[ \left(\sum_{x\in \cB}\gd^u_x\right) \ind_{ \cD^{({0})}_N} \right]+\log \bP^1_N(\cD^{({0})}_N)
 \\
     &=\,  h\,  \bE^1_N \left[ \sum_{x\in \cB}\gd^u_x \right] -h
      \bE^1_N \left[ \left(\sum_{x\in \cB}\gd^u_x\right) \ind_{ \left(\cD^{({0})}_N\right)^\complement} \right]+\log \bP^1_N(\cD^{({0})}_N)
      \\
     &\ge\, L^{2d}
     \left(
  2 \chi(\gb) h^2 + O\left(h^2 L^{-1/4}\right)\right) -  L^{2d}h^3 -2h^2
\\
     &\ge\, L^{2d} \left(
   2 \chi(\gb)h^2 -C_\gb h^{2+\gep}\right)\, , 
 \end{split}
 \end{equation}
 where in the step before the last one 
the three terms correspond to the three terms in the previous line and we have used, in order, \eqref{pertinho},
then 
 \eqref{wop} together with $\sum_{x\in \cB}\gd^u_x \le L^{2d}$ and finally again   \eqref{wop}.
 
We are left with estimating 
first addendum in \eqref{eq:splitLB}.
Note that since $\bP^{1,h}_N( \cdot \ | \, \cD^{({0})}_N )$ is  a small modification of $\bP^1_N$, the probability of making one contact in $\cB$ is small also under the 
former measure. 
In fact
\begin{equation}
\begin{split}
\bP^{1,h}_N \left( \sum_{ x\in \cB} \gd^u_x\, \ge \, 1 
\, \bigg| \, \cD^{({0})}_N \right) \, &=\, 
\frac{
 \bE^{1}_N\left[
  e^{h\sum_{x\in \cB}\gd^u_x }\ind_{\left \{\sum_{ x\in \cB} \gd^u_x\, \ge \, 1 \right\} \cap\cD^{(0)}_N} \right]
  }{\bE^{1}_N\left[
  e^{h\sum_{x\in \cB}\gd^u_x }\ind_{\cD^{(0)}_N} \right]}
  \\
  &\le \, \frac{ e^{h \kappa_d}
   \bP^{1}_N\left(
  {\left \{\sum_{ x\in \cB} \gd^u_x\, \ge \, 1 \right\} \cap\cD^{(0)}_N} \right)
  }
  {
  \bP^{1}_N\left(
  {\cD^{(0)}_N} \right)
  }
  \\
  & \le  \,  (1+2\kappa_d h)
     \bP^{1}_N\left(
  {\left \{\sum_{ x\in \cB} \gd^u_x\, \ge \, 1 \right\} \cap\cD^{(0)}_N} \right)
  \\
  & \le  \,  (1+2\kappa_d h)
      \bP^{1}_N\left(
  \sum_{ x\in \cB} \gd^u_x\, \ge \, 1 \right)
  \\
  & \le  \,  (1+2\kappa_d h) \sum_{ x\in \cB} \bE^{1}_N\left[ \gd^u_x\right] \, \le\, C_1(\gb) L^{2d} h\, ,
  \end{split}
  \end{equation}
  where in the second inequality we used \eqref{wop} and in the last one we have used \eqref{pertinho}
  and $C_1(\gb):=4\chi(\gb)$.
As a consequence we have  for every $\go$ %almost surely
 \begin{equation}
 \bE^{1,h}_N \left[e^{\sum_{x\in \cB}(\gb \go_x-\gl(\gb))\gd^u_x } \, \bigg| \, \cD^{({0})}_N \right]\, \ge\,  1- C_1(\gb)  L^{2d} h\, .
 \end{equation}
Hence using the formula $\log (1+y)\ge y-\frac{1}{2(1-\eta)}y^2$, valid for all $y\ge -\eta$  and  $\eta\in(0,1)$,
with $\eta=C_1(\gb)  L^{2d} h$
 we obtain that 
 \begin{multline}
 \label{decoop}
   \bbE \log  \bE^{1,h}_N \left[e^{\sum_{x\in \cB}(\gb \go_x-\gl(\gb))\gd^u_x } \, \Big| \, \cD^{({0})}_N \right]\\
   \ge\,   -\frac1{2 (1-C_1(\gb)L^{2d}h)} \left( \bbE\left(\bE^1_N \left[e^{\sum_{x\in \cB}(\gb \go_x-\gl(\gb))\gd^u_x }\, \Big| \, \cD^{({0})}_N\right]^2\right) -1\right)
   \\
   \ge \, -\frac 12  \left(1+2C_1(\gb)L^{2d}h\right) \left(  \left(\tilde \bE^1_N\right)^{\otimes 2} \left[e^{\sum_{x\in \cB}(\gl(2\gb)-2\gl(\gb))\gd^{(1)}_x\gd^{(2)}_x} \right] -1 \right)\, ,
   \end{multline}
   where the last inequality  it is simply the fact that $1+2 x \ge 1/(1-x)$ pour $x\in [0,1/2]$, so we are assuming 
   $C_1(\gb)L^{2d}h\le 1/2$,  and we have introduced 
 the notation 
 \begin{equation}
 \tilde \bP^1_N(\cdot)\, :=\,  \bP^1_N\left(\cdot\, \Big\vert \cD^{({0})}_N\right)\, .
 \end{equation}
 The notation $\gd^{(i)}_x$, $i=1,2$ denote the set of contact point for the two marginals $\phi^{(i)}$,  $i=1,2$ of the product measure $\left(\tilde \bP^1_N\right)^{\otimes 2}$.
 Now taking advantage of the fact that, because of the restriction to $\cD^{(0)}_N$, under $\left(\tilde \bP^1_N\right)^{\otimes 2}$ we have 
 $\sum_{x\in \cB}\gd^{(1)}_x\gd^{(2)}_x \le \kappa_d$, we deduce that 
   \begin{multline}
   \label{eq:2mom-magic}
  \left(\tilde \bE^1_N\right)^{\otimes 2} \left[e^{\sum_{x\in \cB}(\gl(2\gb)-2\gl(\gb))\gd^{(1)}_x\gd^{(2)}_x} \right]
  \, 
  \le \, 1 -  \left(\tilde \bP^1_N\right)^{\otimes 2} \left( \sum_{x\in \cB}\gd^{(1)}_x\gd^{(2)}_x \ge 1\right) +\\ 
   e^{\gl(2\gb)-2\gl(\gb)}  \left(\tilde \bP^1_N\right)^{\otimes 2} \left(\sum_{x\in \cB}\gd^{(1)}_x\gd^{(2)}_x=1\right)
  +e^{\kappa_d(\gl(2\gb)-2\gl(\gb))} \left(\tilde \bP^1_N\right)^{\otimes 2} \left(\sum_{x\in \cB}\gd^{(1)}_x\gd^{(2)}_x\ge 2\right)
  \\
  =\,  1+\left(e^{\gl(2\gb)-2\gl(\gb)}-1 \right)  \left(\tilde \bP^1_N\right)^{\otimes 2} \left(\sum_{x\in \cB}\gd^{(1)}_x\gd^{(2)}_x=1\right)\\
  +
  \left(e^{\kappa_d(\gl(2\gb)-2\gl(\gb))}-1 \right) \left(\tilde \bP^1_N\right)^{\otimes 2} \left(\sum_{x\in \cB}\gd^{(1)}_x\gd^{(2)}_x\ge 2 \right) 
  .
   \end{multline}
   In the following lemma (whose proof we postpone), we compute a sharp upper bound for the second term in the sum, and show that the third one is negligible.

\medskip

\begin{lemma}
\label{th:lemformagic}
Set $\gep=1/(4\kappa_d)$. There exists $h_0=h_0(\gb, d)$ such that for 
 $h\in (0, h_0(\gb, d))$  and  $\phi_0 \in \cC_N$ we have
 \begin{equation}
 \label{eq:lemformagic1}
   \left(\tilde \bE^1_N\right)^{\otimes 2} \left(\sum_{x\in \cB}\gd^{(1)}_x\gd^{(2)}_x\right)\, \le\, 
   \left( (2\chi(\gb)h)^2+ h^{2+\gep}\right) L^{2d}\,,
\end{equation}
and
\begin{equation}
\label{eq:lemformagic2}
    \left(\tilde \bP^1_N\right)^{\otimes 2} \left(\sum_{x\in \cB}\gd^{(1)}_x\gd^{(2)}_x\ge 2\right) \,\le\,   h^{2+\gep}L^{2d}\, .
 \end{equation}
\end{lemma}

\medskip

The application of Lemma~\ref{th:lemformagic} to 
\eqref{eq:2mom-magic} yields (for adequate choice of constants)
\begin{multline}
 \left(\tilde \bE^1_N\right)^{\otimes 2} \left[e^{\sum_{x\in \cB}(\gl(2\gb)-2\gl(\gb))\gd^{(1)}_x\gd^{(2)}_x} \right] -1
 \, \ge \\
  \frac 1{2 \chi(\gb)}  \left( (2\chi(\gb)h)^2+C h^{2+\gep}\right) L^{2d} -
 Ch^{2+\gep}\left(e^{\kappa_d(\gl(2\gb)-2\gl(\gb))}-1 \right)   L^{2d}\\
 =\,  
\left( 2\chi(\gb)h^2 -C' h^{2+\gep} \right) L^{2d}
 \, ,
\end{multline}
and going back to \eqref{decoop} we get (for a different constant $C$)
\begin{equation}
 \bbE \log  \bE^{1,h}_N \left[e^{\sum_{x\in \cB}(\gb \go_x-\gl(\gb))\gd^u_x } \, \Big| \, \cD^{({0})}_N \right]
 \, \ge\, \left( \chi(\gb)h^2 -C h^{2+\gep} \right) L^{2d}\, ,
\end{equation}
and we are done with the proof of  Lemma~\ref{pticube}.
\qed

\medskip

\noindent
\emph{Proof of Lemma~\ref{th:lemformagic}.} 
The proof relies essentially on controlling the first and second moment of the sum $\sum_{x\in \cB}\gd^{(1)}_x\gd^{(2)}_x$. 
For what concerns \eqref{eq:lemformagic1} we observe that
\begin{equation}
   \left(\tilde \bE^1_N\right)^{\otimes 2} \left(\sum_{x\in \cB}\gd^{(1)}_x\gd^{(2)}_x\right)\, =\, 
  \sum_{x\in \cB} \left( \tilde \bE^1_N \left[ \gd^u_x\right] \right)^2\, ,
\end{equation}
and that 
\begin{equation}
\begin{split}
\tilde \bE^1_N \left[ \gd^u_x\right] \, =\, 
\frac{ \bE^1_N \left[ \gd^u_x \ind_{\cD_N^{(0)}}\right] }{ \bP^1_N \left( \cD_N^{(0)}\right)}
\, \le \, (1+2h^2) \bE^1_N \left[ \gd^u_x\right]\, &\le \, 
 (1+2h^2) \left( 2 \chi(\gb) h + h L^{-1/4}
 \right)
 \\
 &\le \, 2 \chi(\gb) h + h^{1+b}\, ,
 \end{split}
\end{equation}
and in the last step $b$ is any positive number smaller than $1/(4\kappa_d)$ and the step holds for $h$
smaller that a constant that depends on $\gb$, $d$ and on the choice of $b$.
Therefore the square of this last expression is bounded by 
$(2 \chi(\gb) h)^2+h^{2+b}$, for example $b=1/(5 \kappa_d)$ for $h$ smaller than a constant that depends only on $\gb$ and $d$. The proof of \eqref{eq:lemformagic1} is therefore complete.

\medskip

For what concerns \eqref{eq:lemformagic2} we start like for  \eqref{eq:lemformagic1}, that is we 
get rid of the conditioning with respect to $\cD_N^{(0)}$ and then we proceed with a union bound:
\begin{equation}
\label{eq:m24}
\begin{split}
  \left(\tilde \bP^1_N\right)^{\otimes 2} \left(\sum_{x\in \cB}\gd^{(1)}_x\gd^{(2)}_x\ge 2\right) 
  \,&\le\,
  (1+2h^2)^2  \left(\bP^1_N\right)^{\otimes 2} \left(\sum_{x\in \cB}\gd^{(1)}_x\gd^{(2)}_x\ge 2\right) 
  \\
   \,&\le\,
  (1+2h^2)^2 \sum_{x,y\in \cB: \, x\neq y}  \left(\bP^1_N\right)^{\otimes 2} \left(
  \gd^{(1)}_x\gd^{(2)}_x \gd^{(1)}_y\gd^{(2)}_y=1
  \right)
  \\
   \,&=\,
  (1+2h^2)^2 \sum_{x,y\in \cB: \, x\neq y}  \bP^1_N \left( \gd^u_x \gd^u_y=1\right)^2
  \\
   \,&\le \,
  (1+2h^2)^2 L^{4d} \max_{x\neq y} \bP^1_N \left( \gd^u_x \gd^u_y=1\right)^2\, .
\end{split}
\end{equation}
We proceed by observing that (recall \eqref{eq:phi1})
\begin{equation}
\label{eq:m24.2}
\begin{split}
 \bP^1_N \left( \gd^u_x \gd^u_y=1\right) \, &\le \,
 \bP^1_N \left( \phi_1 (x) \ge u-1- \phi_0(x), \, 
  \phi_1 (y) \ge u-1- \phi_0(x)\right) 
  \\
  & \le \, 
   \bP^1_N \left(  \phi_1 (x) +  \phi_1 (y) \ge 2u-3\right)\, .
   \end{split}
\end{equation}
This is just a Gaussian tail estimate for a centered Gaussian variable with variance 
equal to $\Var(\phi(x)+\phi(y)- \phi_0(x)-\phi_0(y))= \Var(\phi(x)+\phi(y))- 
\Var(\phi_0(x)+\phi_0(y))$ and, for $x \neq 0$, it is therefore smaller than
$\Var(\phi(x)+\phi(y)= 2 (1+p(d)) \gs^2_d$, where $p(d)\in (0,1)$ is the probability that
the the simple random walk, starting from the origin, revisits the origin. The constant 
$p(d)$ has an expression in terms of an integral involving a Bessel function  (see for example \cite[Section~5.9]{cf:Finch}), in particular $p(d)$ decreases to zero as $d$ becomes large: here we will just use
$p(d) \le p(3)= 0.3405\ldots < 7/20$).
Hence we are left with estimating
\begin{equation}
\label{eq:m24.3}
 \bP^1_N \left( \gd^u_x \gd^u_y=1\right) \, \le \,
 \bP\left( \sqrt{2 (1+p(d))} \gs_d \cN \, \ge\, 2u-3\right)\, , 
\end{equation} 
and \eqref{eq:asymptZ}, along with the fact that $u_h$ is asymptotically equivalent to
$\gs_d \sqrt{2 \log(1/h)}$, readily yields that 
\begin{equation}
\label{eq:m24.4}
 \bP^1_N \left( \gd^u_x \gd^u_y=1\right) \, \le \,
h^{1+b}
\, , 
\end{equation} 
for every $b\in (0, (1-p(d))/(1+p(d)))$, and $h$ sufficiently small. Going back to \eqref{eq:m24}we conclude that
\begin{equation}
\label{eq:cond2kappa}
  \left(\tilde \bP^1_N\right)^{\otimes 2} \left(\sum_{x\in \cB}\gd^{(1)}_x\gd^{(2)}_x\ge 2\right)\, \le \, 
  L^{4d} h^{2+2 b}\, =\, L^{2d} h^{2+2b-2d/\kappa_d}
  \,,
 \end{equation}
 for every $b\in (0, (1-p(d))/(1+p(d)))$ and $h$ sufficiently small. Since $(1-x)/(1+x)> (1-2x)$ for every 
 $x>0$ we can choose $b=1-2p(d)>1-7/10=3/10$. Since $2b -2d/\kappa_d >3/5 - 2 d/\kappa_d\ge 1/10> 1/(4 \kappa_d)$ (for the last inequality and the one before the last one, recall that $\kappa_d$ is chosen to be at least 
 $4d$) 
 and this 
 completes the proof of Lemma~\ref{th:lemformagic}.
\qed

 \subsubsection{Step 6: the proof of Lemma \ref{ctunlemme}}
 Of course it is sufficient to prove the result for $B=A\cup\{x\}$, because then the general results can be obtained by adding vertices one by one.
We just have to prove then that 
\begin{equation}
\label{letr}
 \bbE \log \tilde\nu^{A, \overbar{\go}} \left(e^{(\gb \overbar{\go}_x-\gl(\gb)+h)\gd^u_x}\right)\ge -C h^{2-\gga}.
\end{equation}
where $\tilde\nu^{A, \overbar{\go}}$ denotes to the distribution of $(\phi^{(z)})_{z\in \bbZ^d}$ associated with the partition function $\tilde Z_{A, \overbar{\go}}$, cf. \eqref{eq:ZAgo}: note that it depends only on $(\overbar{\go}_y)_{y\in A}$.
We are going to show a stronger statement than \eqref{letr}: in fact we are going to show that  \eqref{letr} holds
also if we average only   
 with respect to $\overbar{\go}_x$  (we use the notation $\bbE_x$) and we freeze the  realization of $(\overbar{\go}_y)_{y\neq x}$.
Setting $p_{A, \overbar{\go}}=p_{A, \overbar{\go}, \gb ,h}:=\tilde\nu^{A, \overbar{\go}}(\gd^u_x=1)$ -- it is a random variable measurable with respect to 
$(\overbar{\go}_y)_{y\in A}$ --
 we  have
\begin{equation}
\begin{split}
\bbE_x \log \tilde\nu^{A, \overbar{\go}} \left[ e^{(\gb \overbar{\go}_x-\gl(\gb)+h)\gd^u_x}\right]
\,& =\,  \bbE_x\log\left(1+p_{A, \overbar{\go}}\left(e^{\gb\overbar{\go}_x-\gl(\gb)+h}-1\right) \right)\,  
\\ 
&\ge\,  - p_{A, \overbar{\go}}^2 \left(e^{\gl(2\gb)-2\gl(\gb)}-1 \right)\, ,
\end{split}
\end{equation}
where in the last step holds if $p_{A, \overbar{\go}} \le 1/2$, so  we can use  the inequality $\log(1+y) \ge y-y^2$ that holds for $y \ge -1/2$: in fact $p_{A, \overbar{\go}} \le 1/2$ yields   $p_{A, \overbar{\go}}(e^{\gb\overbar{\go}-\gl(\gb)}-1) \ge -1/2$ for every value of $\overbar{\go}$ and $\gb$. 

Hence we are reduced to proving that for every $\gb$, $d$ and every $\gga>0$ there exists $h_0$ such that for every 
finite subset $A$ of $\bbZ^d$, every $h\in (0, h_0)$, every $\phi_0 \in \cC_N$  and every realization of $\overbar{\go}$
\begin{equation}
\label{lapreuve}
 %\tilde\nu^A[ \gd^u_x]
 p_{A, \overbar{\go}}\, 
 \le\,   h^{1-\gga/2}\, .
\end{equation}

Recalling 
\eqref{pertinho}, the strategy is  now to show that $\tilde \nu_A$ is not too different from $\bP^1_N$.
And, since $\phi_0$ is fixed, $\phi(x)$ is determined by the realization $\phi^{(z)}$ for at most $2^d$ values of $z$.
We will prove that conditioned to all the rest, the marginal distribution of  $\phi^{(z)}$ has a bounded Radon-Nikodym derivative. Namely (recall that $\vert\overbar{\go}_x \vert \le K_h$, with $K_h$ given in \eqref{eq:Mh}): 

\medskip

\begin{lemma}
\label{derivs}
 For every $z$, $A$ and every measurable subset $B$ of $\bbR^{\bbZ^d}$
 we have that when  $\phi^{(y)}\in \cD^{(y)}_N$ for $y\ne z$,
 \begin{equation}\label{low}
  \tilde\nu^{A, \overbar{\go}}
  \left( \phi^{(z)}\in B \ \big| \ \phi^{(y)}, y \ne z\right)\,\le\,  2 e^{ \gb K_h 3^{d+1} \kappa_d}\bP^1_N
  \left(\phi^{(z)}\in B\right)\, . 
 \end{equation}
\end{lemma}

\medskip

We want to apply Lemma~\ref{derivs} to estimate from above  $\tilde\nu^{A,\overbar{\go}}\left( \gd^u_x=1\right)$
and for this we observe that the event 
$\gd^u_x=1$ relies only on the realization of finitely many $\phi^{(z)}$. More precisely
\begin{equation}
\label{eq:forRN2d}
\left\{(\phi^{(z)})_{z \in \bbZ^d}:\, \gd^u_x=1\right\}\, =\, 
\left\{(\phi^{(z)})_{z \in \bbZ^d}:\,\sum_{z\in I_x} \phi^{(z)}(x) \in [u-1, u+1]- \phi_0(x)\right\}
\, ,
\end{equation}
where the set $I_x:=\{ z : x\in \supp (\phi^{(z)})\}$, contains at least one point and at most $2^d$.
An immediate consequence of Lemma~\ref{derivs} is that 
 the Radon-Nykodym derivative with repect to $\bP^1_N$ of the distribution of 
 $(\phi^{(z)})_{z\in I_x}$ under 
$\tilde\nu^{A, \overbar{\go}}$  is bounded above by 
\begin{equation}(2 e^{\gb K_h 3^{d+1} \kappa_d})^{|I_x|}\le(2 e^{ \gb K_h 3^{d+1} \kappa_d})^{2^d}\end{equation}
As the event in \eqref{eq:forRN2d} depends only on $(\phi^{(z)})_{z\in I_x}$ we have in particular (recall that $K_h=o(\log(1/h))$)

% It is therefore possible to write 
% this event as a countable disjoint union of product events, that is events of the form 
% $\{(\phi^{(z)})_{z \in \bbZ^d}:\, \phi^{(z)}(x)\in B_{z}$ for $z \in I_x\}$, for a proper choice of Borelian subsets 
% $B_z$ of $\bbR$ and Lemma~\ref{derivs} implies
% \begin{equation}
%  \tilde\nu^{A, \overbar{\go}}\left( \phi^{(z)}(x)\in B_{z} \textrm{ for } z \in I_x\}
%  \right) \, \le \, (2 e^{\gb K_h 3^{d+1} \kappa_d})^{\vert I_x \vert}\bP^1_N
%  \left( \phi^{(z)}(x)\in B_{z} \textrm{ for } z \in I_x\}
%  \right)\,,
% \end{equation}
% and therefore the same estimate extends to the countable union \eqref{eq:forRN2d}.
% If we  we see that from
% Lemma~\ref{derivs} we obtain that for every $x$
\begin{equation}
   \tilde\nu^{A,\overbar{\go}}\left( \gd^u_x=1\right)\le (2 e^{\gb K_h 3^{d+1} \kappa_d})^{2^d}\bP^1_N\left( \gd^u_x=1\right)\,
   \le  \, (2 e^{\gb K_h 3^{d+1} \kappa_d})^{2^d} 3 \chi(\gb) h \, \le 
   \,  h^{-\gga/2} h\, ,
\end{equation}
where 
we have applied  \eqref{pertinho} in the inequality before the last one. 
\qed

\begin{proof}[Proof of Lemma \ref{derivs}]
% First let us note
% one can prove \eqref{low} with $\bP^1_N(\phi^{(z)}\in B)$ replaced by  $\bP^1_N(\phi^{(z)}\in B \ | \ \cD_N)$,
% which coincides, by independence, with $\bP^1_N(\phi^{(z)}\in \cH \ | \ \cD^{(z)}_N)$. 
% This is simply because $\bP^1_N ( \cD^{(z)}_N)$ has probability close to one as $h$ becomes small, cf. \eqref{wop},
% so 
% \begin{equation}
% \bP^1_N \left(\phi^{(z)}\in B\, \big \vert \, \cD_N^{(z)} \right) \, \le \frac{\bP^1_N \left(\phi^{(z)}\in B\right)}{\bP^1_N\left( \cD_N^{(z)}\right)}\, \le \, \frac{3}{2} \, \bP^1_N \left(\phi^{(z)}\in B\right)\,.
% \end{equation}
We recall that
\begin{equation}
\tilde \cB_z\, =\,  (L^2+L)z+ \lint 1,L^2+2L-1\rint^d\, =\, Mz+ \lint 1,M+L-1\rint^d\, .
\end{equation}
is the support of $\phi^{(z)}$.
Now the conditional Radon-Nikodym derivative given $\phi^{(y)}$, $y\ne z$  is equal to
\begin{equation}
\label{eq:lvmq2}
  \frac{e^{\sum_{x\in \tilde \cB_z \cap A}(\gb \overbar{\go}_x-\gl(\gb)+h)\gd^u_x}\ind_{\cD^{(z)}_N}}{\bE^{1,(z)}_N\left[e^{\sum_{x\in \tilde \cB_z \cap A}(\gb \overbar{\go}_x-\gl(\gb)+h)\gd^u_x}\ind_{\cD^{(z)}_N}\right]}
\end{equation}
where the superscript $(z)$ in the expectation underlines that the average is taken only w.r.t. $\phi^{(z)}$.
We consider the sum over $\tilde \cB_z \cap A$ and not $A$ because terms coming from $x\in (\tilde \cB_z)^{\cc}$ are completely determined by $\phi^{(y)}$, $y\ne z$ and therefore cancel out in the numerator and denominator.

Because of the restriction to $\cD_N$ there are at most $3^d\kappa_d$ contacts in $\tilde \cB_z$.
Indeed if $\gd^u_x=0$ for some $x\in \tilde \cB_z$ then there must exist  $z'$ with $|z-z'|\le 1$ for which $\phi^{(z')}(x) \ge  2^{-d}(u-2)$ (because otherwise   the sum $\sum_{y\in \bbZ^d} \phi^{(y)}$ which contains at most $2^d$ non-zero terms is smaller than $u-2$). We conclude using the constraint $\cD^{(z')}_N$ and the fact that there are $3^d$ choices for $z'$. 
Thus using our uniform bound on $\overbar \go$ we have
\begin{equation}
\label{eq:lvmq3}
\left|\sum_{x\in \tilde \cB_z\cap A}(\gb \overbar{\go}_x-\gl(\gb)+h)\gd^u_x\,\right| \le\,   3^d \kappa_d \left( \gb K_h + |h| +\gl(\beta) \right) \le \frac{1}{2} 3^{d+1} \kappa_d \gb K_h  \, ,
\end{equation}
 where the last inequality is valid for $h$ small enough (recall that $K_h$ diverges when $h\to 0$). 
 Using this bound on the numerator and denominator we obtain
 that the Radon Nykodym is bounded above by 
\begin{equation} 
\frac{1}{\bP^{1}_N\left(\cD^{(z)}_N\right)}  e^{3^{d+1} \kappa_d \gb K_h}\, ,
\end{equation}
 which yields the desired estimates because $\bP^{1}_N(\cD^{(z)}_N)\ge 1/2$.
 \end{proof}

\section{Proof of Theorem \ref{th:paths}: Lower bound on the height} 
\label{sec:LBH}

The main object of this Section is to prove inequality (A) in  \eqref{upperandlower} holds. 
It can be reformulated as follows:

\medskip

\begin{proposition}
\label{lescontactz}
Given $\gep>0$, for all $h\le h_0(\gep)$, there exists $c(\gep,h)>0$
such that almost surely for $N$ sufficiently large (depending on $\go$) we have
 \begin{equation}
 \bP_{N, \go, \gb, h}\left( \sum_{x\in \mathring{\gL}_N} \ind_{\{ |\phi(x)| \le (1-\gep)\sigma_d\sqrt{2\log (1/h)}\}}\ge \gep N^d \right) \,\le\,  e^{-c N^d}\, .
 \end{equation}
\end{proposition}

\medskip

For simplicity we redefine from now till the end, that is the entire Sections \ref{sec:LBH} and \ref{sec:UBH},     the value of $u_h$  by keeping only the leading behavior, that is we set 
 \begin{equation}
\label{def2uh}
 u_h\,:= \, \sigma_d\sqrt{2\log (1/h)}\, .
\end{equation}
Recalling \eqref{defu}-\eqref{eq:uh}, we see that the newly defined $u_h$ coincides  to first order with the one used in Section~\ref{sec:LB}.

\medskip
 The proof of Proposition~\ref{lescontactz} is divided into three steps. To each step is devoted one of the subsections that follow.

% The aim is showing that for every $\gep\in (0,1)$ and $\eta\in (0,1)0$ we can find $h_0$ such that for every $h\le h_0$  
% \begin{equation}
% \label{eq:mainLB}
% \lim_{N \to \infty}
% \bbE \left[ \bP_{N, \go, \gb, h} \left( 
%  \sum_{x \in \mathring \gL_N} \ind_{\vert \phi(x) \vert \le (1- \eta) \gs_d \sqrt{2\log (1/h)}}
% \, \ge \, \gep N^d
% \right) \right] \, =\, 0\,  .
% \end{equation}

\subsection{Step 1: Upper bound on  the contact fraction}

The first important ingredient to prove Proposition \ref{lescontactz} is a quantitative control on the contact density under
$\bP_{N, \go, \gb, h}$. We know the contact density  is close to the optimal density $p_h$ when  $h \searrow 0$, and in this limit $p_h\sim 2 \chi(\gb) h$
 (see \eqref{optimiz}). But  we can extract from Theorem \ref{th:mainF} also
 upper and lower  Large Deviations estimates:
 for our arguments we just need  a control from above, and this is what we are going to develop next.
Given $\eta>0$ we set 
\begin{equation}
\label{beneta}
B_{N, \eta}\, :=\, \left\{ \phi\in \bbR^{\mathring\gL_N} \, :\,\frac{1}{N^d}
\sum_{x\in\mathring \gL _N} 
\gd_x
\,
\le  \, \left (2 \chi(\gb) +\eta\right) h \right\}\, .
\end{equation}
Then we have the following: 

\medskip

\begin{lemma}
\label{lalem}
 For every $\eta>0$, there exists constant $c(\beta,\eta)>0$ and $h_0$ such that for all $h\in(0,h_0)$,we have almost-surely for $N$ sufficiently large 
 \begin{equation}
 \bP_{N, \go, \gb, h} \left( B^{\cc}_{N, \eta} \right)\le e^{-c h^2 N^d}.
 \end{equation}
\end{lemma}

\medskip

\begin{rem}
Actually we only need to show the result for one positive value of  $\eta$, no need to choose it arbitrarily small.
In fact, we are going to apply   Lemma~\ref{lalem} with $\eta = \chi(\gb)$.
Nevertheless,  we feel that a  precise result  gives more intuition about the proof.
\end{rem}

\medskip

\begin{proof}
Given an event $A$, we use the notation  $Z_{N, \go, \gb, h}(A)$ for the partition function restricted to the set $A$ that is
\begin{equation}\label{restriparti}
 Z_{N, \go, \gb, h}(A)= \bP_{N, \go, \gb, h}(A)Z_{N, \go, \gb, h}.
 \end{equation}
For every $v>0$, using the convexity of $v\mapsto \log Z_{N, \go, \gb, h+v}\left(B^{\complement}_{N, \eta}\right)$ and the fact that its derivative at the origin is
\begin{equation}  
\bE_{N, \go, \gb, h}\left[\sum_{x \in \mathring \gL _N} \gd_x  \ \Bigg\vert \ B^{\complement}_{N, \eta} \right]\,\ge\,   (2\chi(\gb)+\eta) h
N^d \, ,\end{equation}
we have
\begin{equation}
\label{eq:forcontr}
\log \frac{ Z_{N, \go, \gb, h+v}}{ Z_{N, \go, \gb, h}}\, \ge\, 
\log \frac{ Z_{N, \go, \gb, h+v}\left(B^{\complement}_{N, \eta}\right)}{ Z_{N, \go, \gb, h}}\, \ge\,  (2\chi(\gb)+\eta) h v 
N^d
%\left\vert \mathring{\gL}_N\right\vert
+\log \bP_{N, \go, \gb, h} \left( B^{\complement}_{N, \eta} \right)\,,
\end{equation}
where we have used that
 \begin{equation}
\bP_{N, \go, \gb, h} \left( B^{\complement}_{N, \eta} \right) \, =\frac { Z_{N, \go, \gb, h}\left(B^{\complement}_{N, \eta}\right)}{ Z_{N, \go, \gb, h}}.
\end{equation}
Dividing by $N^d$ and taking the limit in \eqref{eq:forcontr}
we obtain that $\bbP( \dd \go)$-a.s.
\begin{equation}
\limsup_{N\to \infty} \frac{1}{N^d}\log \bP_{N, \go, \gb, h} \left( B^{\complement}_{N, \eta} \right)\,\le\,  \tf(\beta,h+v)-\tf(\beta,h)- (2\chi(\gb)+\eta) h v \, .
\end{equation}
Choose now $v=b h$ with $b:= \eta/(2 \chi(\gb))$. 
By applying the precise asymptotic results of Theorem \ref{th:mainF}  we obtain
\begin{equation}
\tf(\beta,h+v)-\tf(\beta,h)- (2\chi(\gb)+\eta) h v \stackrel{h\searrow 0} \sim h^2 
\left( \chi(\gb) b^2 - \eta b\right) \, =\,  - h^2  \frac{\eta^2}{4 \chi(\gb)}\, ,
\end{equation}
which completes the proof of Lemma~\ref{lalem}.
\end{proof}

\subsection{Step 2:  lower bound on harmonic averages}

We introduce a length $L\in \bbN$ sufficiently large (how large is specified below) and divide $\gL_N$ into disjoint cubes of edge length $L$. We introduce the disjoint cubes 
$C_L^z:= \lint 0,L \rint^d + z L $ and 
their \emph{ centers} $x_L(z):= zL  +(1, \ldots,1) \lfloor L/2 \rfloor$.
We choose $L$ so that the variance of a zero-boundary free field on $C^z_L$ at the center of the cube is close to variance of the infinite volume field
 (recall \eqref{eq:Gest0})  

 \begin{equation}\label{conditL}
 \sigma^2_{d,L}:=G_{L}(x_L(0),x_L(0))\ge \sigma^2_d(1-\gep/2).
\end{equation}
 We consider only  $z \in \lint 0, \lfloor N/L\rfloor-1 \rint^d=: I_{N,L}$ meaning that we consider cubes for which $C_L^z \subset  \gL_N$. 
 We let $\phibar_L(z)$ denote the harmonic average, at the center $x_L(z)$ of the cube $C_L^z$, 
 of the field  on the boundary of $C_L^z$
\begin{equation}\label{lharmo}
\phibar_L(z)\, :=\, \sum_{x \in \partial C_L^z} p_{L, z}(x) \phi(x)\, .
\end{equation}
where $\partial C_L^z:=  C_L^z \setminus  (\lint 1,L-1 \rint^d + z L) $ is the internal boundary of  $C_L^z$
and $p_{L, z}(x)$ is the probability that a simple symmetric random walk issued from the center $x_L(z)$
hits $\partial C_L^z$ at $x$.

We are going to show that for most $z$'s, $\phibar_L(z)$ lies above height 
$(1-\gep/2)u_h$.
We introduce the event
\begin{equation}
F_{N, \gep}\, :=\, 
\left\{ 
\phi\in \bbR^{\mathring \gL_N}:\,
\left \vert 
\Xi_{N, L}(\phi)
 \right \vert 
 \, \le \, \frac {\gep }{2 L^d} N^d
 \right\}\,, 
\end{equation}
where
\begin{equation}
\Xi_{N, L}(\phi)\, :=\, 
\left\{ z \in I_{N,L}:\,
 \left \vert \phibar_L(z) \right \vert
 \, \le (1-\gep/2) u_h  
 \right\}\, ,
\end{equation}
is a  random subset of $I_{N,L}$.

\begin{lemma}\label{lalem2}
Given $\gep>0$, the exists $h_0(\gep)$ such for all $h\in(0,h_0)$, there exists $c(\gep,h)>0$ for which 
for all $N$ sufficiently large we have
\begin{equation}
\label{eq:cncl1.1}
 \bP_{N, \go, \gb, h} \left( F^{\cc}_{N, \gep} \right)\, \le e^{-c N^d}.
\end{equation} 

\end{lemma}

\begin{proof}
 Using Lemma \ref{lalem} we can in fact look only at the probability of
 $F_{N,  \gep}^{\cc} \cap B_{N, \eta}$ for some arbitrary value of $\eta$.
 We choose $\eta=\chi(\beta)$ and simply denote the corresponding event by $B_N$.
 Now we have for any event $A$  and $h>0$  
 \begin{equation}
 \label{lacompute}
 \begin{split}
\bbE \left[\bP_{N, \go, \gb, h} \left( A \cap B_N\right)\right]
 \, &\, = \bbE\left[ \frac{Z_{N, \go, \gb, h}\left(A \cap B_N\right)}
{Z_{N, \go, \gb, h}}\right]\\
&\le \,
 \bbE\left[Z_{N, \go, \gb, h}\left(A \cap B_N\right) \right]/  \bP_N\left(  \forall x \in \mathring \gL _N, \ \phi(x)>1 \right)
 \\
 &=\,  
 \bE_N\left[ e^{h \sum_{x\in \mathring \gL_N}\gd_x} \ind_{A\cap B_N}\right]
 \Big/  \bP_N\left(   \forall x \in \mathring \gL _N, \ \phi(x)>1\right)\, ,
 \end{split}
 \end{equation}
 where we have simply bounded the denominator by the contribution of trajectories with no contacts.
 Given that the probability in the denominator  behaves  sub-exponentially in the volume $N^{d}$ (recall the entropic repulsion estimates discussed below \eqref{eq:theF})
 and that the number of contact is bounded above by $3\chi(\beta)hN^d $, we have for $N$ sufficiently large
 \begin{equation}\label{lacompute2}
  \bbE \left[\bP_{N, \go, \gb, h} \left( A \cap B_N\right)\right]
  \le e^{4\chi(\beta)h^2 N^d}\bP_N\left(A\cap B_N\right)\,. 
 \end{equation}
Hence to prove 
\eqref{eq:cncl1.1}
 it is sufficient to show that
\begin{equation} 
\bP_N\left(F_{N, \eta, \gep}^{\cc}\cap B_N \right) \, \le \, e^{-5\chi(\beta) h^2 N^d}\, ,
\end{equation}
and then apply the Borel-Cantelli via a Markov inequality bound. 
We are in fact going to show that for any realization of $\phibar_L$ for which 
$\vert \Xi_{N,L}(\phi)\vert > \frac{\gep}{2 L^2} N^d,$
 \begin{equation}
 \label{aprovee} 
  \bP_N\left( B_N \, \big| \, \phibar_L(z), \,  z\in I_{N,L}\right)\,\le\,   e^{-5\chi(\beta) h^2 N^d}\, .
 \end{equation}
The Markov property for the LFF states that under $\bE_N$  
the random variables 
$\phi (x_L(z))-\phibar_L(z)$ are IID Gaussian variables with variance $\sigma^2_{d,L}$, which are indendent of $(\phibar_L(z))_{z\in I_{N,L}}$.
In particular, conditioned to $\phibar_L$,  $(\gd_{x_L(z)})_{z\in \lint 0, \lfloor N/L \rfloor -1\rint^d}$ are independent Bernoulli variables with respective parameters 
\begin{equation}
q_z\left(\phibar_L\right)\,:=\, \frac{1}{\sqrt{2\pi}}\int_{-1}^1 \exp\left(-\frac{(t-\phibar_L(z))^2}{2 \sigma^2_{d,L}}\right)\dd t\,.
\end{equation}
Note in particular that the above expression is  decreasing in $|H(z)|$ and using also \eqref{conditL}
we see that  for $z\in \Xi_{N,L}$
\begin{equation}
\begin{split}
q_z\left(\phibar_L\right)\, &\ge\,  \frac{1}{\sqrt{2\pi}}\int_{-1}^1 \exp\left(-\frac{(t-(1-\gep/2)u_h)^2}{2 \sigma^2_{d,L}}\right)\dd t
\\
& = \, P\left( 
\gs_d \cN \in \left[a_\gep  u_h - 1/a_\gep
,
a_\gep  u_h - 1/a_\gep 
\right]\right)
=:\overbar q (h,d,L, \gep)\, ,
\end{split}
\end{equation}
where in the intermediate step we used $a_\gep:=(1-\gep/2)^{1/2}$. Replacing $u_h$ by its value \eqref{def2uh} and using 
\eqref{eq:asymptZ} in a  rough way, we obtain that 
$ \overbar q\ge h^{1-\gep/2}$, at least for $h$ sufficiently small. 

\medskip

Hence in particular, considering within $\mathring{\gL}_N$ only the points of the form $x_L(z)$ with $z\in \Xi_{N,L}$ we obtain that, conditioned to $H$, when  $\vert\Xi_{N,L}\vert> \frac{\gep}{2 L^2} N^d$, the quantity $\sum_{x\in \mathring \gL_N} \gd_x$ stochastically dominates 
a binomial random variable of parameters $\overbar q$ and $\lceil \gep/(2 L^2) \rceil N^d$ (we then omit the integer part for notational convenience.
We have thus 
 \begin{equation}
 \bP_N \left(B_N \ \big| \ \phibar_L (z), z\in I_{N,L} \right)\, \le\,  P\left( \textrm{Bin}( \gep  N^d/ (2L^d), h^{1-\eta/2} ) \, >\,  3 \chi(\gb) hN^d\right),
\end{equation}
where $\textrm{Bin}(n,p)$ denotes a binomial random variable of parameters $n$ and
$p$. By the first inequality in Lemma~\ref{th:binomial}
applied with  $n=  \gep  N^d/ (2L^d)$ and  $\gD= 6 \chi(\gb) h^{\gep/2} L^d / \gep$, it is just a matter of choosing $h$ suitably small to get to 
\begin{equation}
  \bP_N\left(B_N \ | \ \overbar\phi(z), z\in I_{N,L}\right)\,  \le \, 
\exp\left(-  \frac{\gep h^{1- \gep/2}}{2L^d} N^d\right) 
\, , 
\end{equation}
which largely proves \eqref{aprovee}.
\end{proof}

\subsection{Step 3:   positive density of low sites is incompatible 
with
harmonic average lower bound}

Now let us consider the event whose probability we wish to bound  in Proposition \ref{lescontactz} which is 

\begin{equation}
\label{eq:mainE} A_{N,\gep}:=
\left \{ \phi\in \bbR^{\mathring \gL_N}:\,
 \sum_{x\in \mathring \gL_N} \ind_{\{\vert \phi(x) \vert \le (1- \gep)u_h\}}
\, \ge \, \gep N^d
\right\} \, .
\end{equation}
Using Lemma \ref{lalem} and Lemma \ref{lalem2}, it is sufficient to prove that the probability of $A_{N,\gep}\cap B_{N}\cap F_{N,\gep}$ decays exponentially with the volume.
From \eqref{lacompute2} (with $A= A_{N,\gep}\cap F_{N,\gep}$)
we deduce that 
\begin{equation}
 \bbE \left[\bP_{N, \go, \gb, h} \left(A_{N,\gep}\cap B_N\cap F_{N,\gep}\right)\right]\le e^{4\chi(\beta)h^2 N^d}\bP_N\left(A_{N,\gep}\cap B_N\cap F_{N,\gep}\right)\,.
\end{equation}
Hence we are left with showing that 
\begin{equation}
\label{eq:lwst63}
 \bP_N\left(A_{N,\gep}\cap F_{N,\gep}\right)\, \le\,  e^{-5\chi(\beta)h^2 N^d}.
\end{equation}
Now to conclude we need observe that on the event $A_{N,\gep}\cap F_{N,\gep}$
the Hamiltonian $H_N(\phi)$ is anomalously large.
Indeed, on $A_{N,\gep}$ , we have necessarily  
\begin{equation}
 \left\vert\left \{ z \in \lint 0, (N/L)-1 \rint \ : \ \exists x\in zL+ \lint 1, L\rint^d , |\phi(x)|  \le (1- \gep)u_h  \right\}\right\vert\, \ge\,  \gep \frac{3N^d}{4L^d},
\end{equation}
and hence on $A_{N,\gep}\cap F_{N,\gep}$ we have 
\begin{equation}
\label{leqwas}
 \left\vert\left\{ z \ : \ |H(z)|\ge (1-\gep/2)u_h \text{ and } \exists x\in zL+ \lint 1, L\rint^d , |\phi(x)|  \le (1- \gep)u_h  \right\}\right\vert
  \, \ge\,   \gep \frac{N^d}{4L^d} .
 \end{equation}
Note that for each $z$ which satisfies the property in the right-hand side \ of \eqref{leqwas} we can find $x_1\in zL+ \lint 1, L\rint^d$ such that $ |\phi(x_1)|\le (1- \gep)u_h$, and since $\phibar_L(z)$ is a weighted average of the values of $\phi$ on the boundary of $zL+ \lint 0, L\rint^d$,  there exists $x_2$ in the boundary of $zL+ \lint 0, L\rint^d$ such that $|\phi(x_2)|\le (1- \gep/2)u_h$.

\medskip

Considering a path of minimal length (which in this case has to be smaller than $dL$) between $x_1$ and $x_2$, we obtain that there exists a pair of neighbors $y_1\in zL+ \lint 0, L\rint^d$ $y_2\in zL+ \lint 1, L\rint^d$,  such that 
\begin{equation} |\phi(x_1)-\phi(x_2)| \ge \frac{\gep}{2dL} u_h.\end{equation}
%{HERE I AM USING THAT $x_2$ HAS A NEIGHBOR IN  $zL+ \lint 1, L\rint^d$) why this remark? You mean that $x_2$ is not in a corner? Yes or on an edge in dimension $3$ etc... but these vertices do not contribute to the harmonic mean so can be ignored. But maybe just mentionning the subject is going to confuse the reader. This is just something to ensure that we are considering distinct edges (which itself is not that important)}
Given that these edges are necessarily distinct for different values of $z$, we obtain that 
\begin{equation} A_{N,\gep}\cap F_{N,\gep} \subset \left\{ \phi \in  \bbR^{\mathring \gL_N} \ : \  \sumtwo{\{x,y\} \in \gL_N }{x\sim y}(\phi(x)-\phi(y))^2 \ge   \frac{\gep^3 u^2_h}{16d^2L^{d+2}} N^d \right\} \end{equation}
To conclude, we observe that 
${\gep^3 u^2_h}/(16dL^{d+1})$ can be made arbitrarily large by choosing $h$ small, and use that
by \cite[(B.8)]{cf:GL} one can find $C>0$ such that 
\begin{equation}
 \bP_{N} \left( \sumtwo{\{x,y\} \in \gL_N }{x\sim y}(\phi(x)-\phi(y))^2 \ge CN^d \right) \, \le \, e^{-N^d}.
\end{equation}
Therefore \eqref{eq:lwst63} holds and the proof of Proposition~\ref{lescontactz} is therefore complete.
\qed

\section{Proof of Theorem \ref{th:paths}: Upper bound on the height}
\label{sec:UBH}

In this Section we  prove inequality (B) in  \eqref{upperandlower}. We keep the definition in the previous section \eqref{def2uh} for the value of $u_h$.

\medskip

\begin{proposition}
\label{lescontactzB}
Given $\gep>0$, for all $h\le h_0(\gep)$, there exists $c=c(\gep,h)>0$ and $N_0(\gep, h, \go)$, with $N_0=N_0(\gep, h, \go)<\infty$ $\bbP(\dd \go)$-a.s., 
such that $\bbP(\dd \go)$-a.s. $N\ge N_0$  we have 
 \begin{equation}
 \label{eq:DNeps}
 \bP_{N, \go, \gb, h}\left( \cD (N , \gep) \right)\,  \le \, e^{-c N^d} \ \ \ \  \text{ 
with } \ 
  \cD (N , \gep)\, :=\, \left\{
 \sum_{x\in \mathring{\gL}_N} \ind_{\{ |\phi(x)| \ge (1+\gep)u_h\}}\ge \gep N^d \right\}\, .
 \end{equation}
\end{proposition}

\medskip

The proof will be achieved through various steps of which we give here a quick sketch that could be useful as a guideline:

\medskip

\begin{enumerate}
\item Section~\ref{sec:prepgrid}: \emph{construction of  a hierarchy of (almost) coverings of $\gL_N$}. 
To prove Proposition~\ref{lescontactzB}  we need to exploit the fact that the contact fraction is close to $p_h:= 2h \chi(\beta)$. However, a statement about the global density like Lemma \ref{lalem} is not sufficient. We need and will  show that if we divide $\gL_N$ into boxes (hyper-cubes) -- we will call them \emph{level-$0$ boxes} -- of volume roughly $h^{-2}$, the empirical contact density  in \emph{most of these boxes} is close to $p_h$. Such a statement is only about  level-$0$ boxes, but later on in the proof we will need a full hierarchy of boxes.
In such hierarchy,  level-$0$ is the lowest  level,
 the highest being the one of $\gL_N$ itself. 
 It is more practical to introduce from the start the full hierarchy even if up to
Subsection~\ref{sec:multiscale} only level-$0$ is used. For a part of the argument the level-$0$ boxes will be further  split 
into $6^d$ sub-boxes that will be called \emph{elementary boxes}.
\item Section~\ref{density}: \emph{control of contact density on elementary boxes}.
We will introduce an event $\cC(N, \gd)$, with $\gd>0$ a parameter that is simply going to be chosen proportional to $\gep$ in the end, on which the field has  approximately the correct contact fraction in most of the elementary boxes (that are just a further partition of each of  the 
level-$0$ boxes into a finite number, precisely $6^d$,   boxes). We will show that the $\bP_{N, \go, \gb, h}$ probability 
of the complement of  $\cC(N, \gd)$ is negligible, in the sense that it is $O(\exp(-c N^d))$ for some $c>0$. 
\item Section~\ref{sec:Decomp}: \emph{the main body of the argument}. We write $\phi= \psi + \overbar \psi$, with $\psi$ and $\overbar \psi$ independent. This decomposition is similar to the one made in Section \ref{sec:splitting}:   $\overbar \psi$ has small variance and contains the long range correlations, while 
the field $\psi$ has no correlations when we consider sites that belong to different level-$0$ boxes. We introduce at this stage 
one more event, called $\cB (N, \gd)$ that contains the requirements we demand on the $\overbar \psi$ and such that
the $\bP_{N, \go, \gb, h}$ probability of the complement of $\cB (N, \gd)$ is $O(\exp(-c N^d))$.
 In a nutshell, what we require on $\overbar \psi$ is that it is   
 close to being affine on most of the level-$0$ bowes and we do this by passing through second order discrete derivatives of $\phi$: arguments would have been much more straightforward if we were 
 able to show (with the proper exponential  probability estimate) that the field $\overbar \psi$ is small or
 that $\overbar \psi$ is almost flat, i.e. its gradient is small: the point is that we would like to get rid of 
 $\overbar{\psi}$ and exploit the independence properties of $\psi$.   We will however explain why we cannot prove this. 
Nevertheless, working with \emph{locally affine} $\overbar \psi$ will turn out to be sufficient 
to bound in a satisfactory way the $\bP_{N, \go, \gb, h}$ probability 
of $\cB (N, \gd)\cap \cC (N, \gd)\cap \cD (N, \gep)$. For readability we split this section into two and we 
devote a separate section to the probability estimates on the $\psi$ field inside the level-$0$ boxes.
\item Section~\ref{sec:level-0}:
  \emph{Level-$0$ estimates}. Here we 
  exploit the fact that we are in  $\cC(N , \gd)$, hence with a control on the contact density in most of the level-$0$ boxes, and in
   $\cB (N, \gd)$, hence with a strong control on the long range correlated field $\overbar \psi$ (in most of the level-$0$ boxes).
 We pick any of the good level-$0$ boxes 
  and  develop geometric arguments, coupled with probability bounds, that show that it is improbable 
  that the absolute value of the field goes above level $(1+ \gep)u_h$. 
\item Section~\ref{sec:multiscale}: the multiscale bound.
This section is devoted to bounding the probability of the complement of $\cB (N, \gd)$ and this is the step in which the multiscale construction introduced in Section~\ref{sec:prepgrid} is exploited. 
\end{enumerate}

\subsection{Construction of nested (almost) coverings of $\gL_N$}
\label{sec:prepgrid}

\medskip

Before stating the main result of this subsection we must introduce some notations.  
As announced, 
we want to \emph{cover}  $\gL_N$ with cubic boxes with volume   $c(\log(1/h)/h)^2$, $c$ a positive constant: we will call them \emph{level-0 boxes}. On top of this we want to construct a hierarchy of boxes: for each $i\ge 1$ we want to construct boxes of level $i$ which are obtained by grouping $2^d$ disjoint boxes at level $i-1$ (meaning that boxes at level $i$ of the hierarchy will have volume  asymptotically equivalent to   $2^i c (\log(1/h)/h)^2$).
Finally, on top of this we require that at each level of the hierarchy, some amount of free space is left between the boxes.
We stop the procedure once \emph{we reach}  $\gL_N$. Since we want to cover most of $\gL_N$ with level-0 boxes, in the sense that we want that the fraction of uncovered sites can be made arbitrarily small, it turns out to be more practical for the construction to start from the top level box, that is $\gL_N$,  and work down to when we get to level-$0$.

This structure will be of fundamental importance for the multiscale analysis introduced in Section \ref{sec:multiscale}. We introduce it  beforehand because the statement about the local contact density presented in Section \ref{density} needs to be formulated in terms of level-$0$ boxes in our hierarchy, but we stress that up to Section \ref{sec:multiscale} we are going to need only the level-0 of this construction. 

\medskip

 Set $\tilde N_0:= N$ and  define $\tilde N_j$ recursively for $j\ge 1$. Given $\varkappa\in (0, 1)$   we set also
\begin{equation}\label{larecurr}
\tilde N_j\, :=\, \left \lfloor \frac { \tilde N_{j-1}- 4 \left \lfloor \tilde N_{j-1}^{1-\varkappa} \right\rfloor } 2 \right\rfloor \, ,
\end{equation} 
and 
\begin{equation}
 J(h,N)\,:=\, \inf\left\{ j\ge 1 \ : \ \tilde N_{j+1}\le 7( \log(1/h)/h)^{2/d} \right\}\, .
\end{equation}
The parameter $\varkappa$ can be chosen arbitrarily in $(0,1)$, for example $\varkappa=1/2$, but this time readability is helped if we do not make the constant explicit. 
Note also that, if $h$ is sufficiently small and  $N>  7(\log(1/h)/h)^{2/d}$,  we have
\begin{equation}\label{tildenj}
7(\log(1/h)/h)^{2/d} \le \tilde N_J \le  15( \log(1/h)/h)^{2/d}\, . 
\end{equation}

Then we construct recursively a sequence $(\tilde \cC_j)_{0\le j \le J}$ of collections of  $2^{dj}$ disjoint boxes of edge length $\tilde N_j$ within $\gL_N$.

\begin{itemize}
 \item 
We let $\tilde \cC_j= \{ \tilde B_{j,k},  k\in \lint 1, 2^{dj}\rint \}$ denote the collection of boxes at step $j$.
We initiate with $\tilde B_{0,1}:=\gL_N$  and $\tilde \cC_0:= \{  \gL_N\}$.
 
\item Now for $j\ge 1$, given $B_{j-1,k_0}$ a generic box in $\tilde \cC_{j-1}$
 we introduce the $2^d$ disjoint hypercubes $\{\tilde B_{j,k}:\, 
k\in \lint 2^dk_0+ 1, 2^d(k_0+1)\rint \}$ of edge length $\tilde N_j$ satisfying $\tilde B_{j,k}\subset \tilde B_{j-1,k_0}$.
The cubes $\tilde B_{j,k}$  are placed inside $\tilde B_{j-1,k_0}$
as explained in Fig.~\ref{fig:multisc1}.
If we consider the $d$ hyperplanes that bisect $\tilde B_{j-1,k_0}$ and are orthogonal 
respectively to $e_1, e_2, \ldots , e_d$ we split $\tilde B_{j-1,k_0}$ into $2^d$ chambers and the $2^d$ disjoint hypercubes
we have just introduced are simply placed at the center of each chamber.
\item  We also introduce the \emph{conditioning grid} $\tilde G_j$: for every $\tilde B_{j-1,k_0}$ we build a portion of the grid
by considering the union 
of  the external boundary 
of  $\tilde B_{j-1,k_0}$ and of the portion in  $\tilde B_{j-1,k_0}$ of the $d$ bisecting hyperplanes we have introduced 
at the previous point (they are the boundaries of the $2^d$ chambers).  We repeat the procedure for each one of the $2^{d(j-1)}$ hypercubes $\tilde B_{j-1,k_0}$ and, by considering the union of the sets we have constructed we obtain 
$\tilde G_j$ (that has therefore $2^{d(j-1)}$ connected components). We add to this collection 
$\tilde G_0$ which is $\bbZ\setminus  \mathring{\gL}$: we could have added just the external boundary, but this is notationally convenient. 
In practice, it is more compact to work with the cumulative grid,  we define the cumulative grid for $j=0, 1, \ldots$ by
\begin{equation}
\tilde \bbG_j\,:=\, \bigcup_{j'=0}^j G_{j'}\, .
\end{equation}
\end{itemize}

\begin{rem}
 In the above construction, \emph{bisecting hyperplanes'} and \emph{at the center of each chamber} have to be considered after integer rounding if necessary (so that the hyperplanes  and the boxes $\tilde B_j$ are subsets of $\bbZ^d$).
\end{rem}

   \begin{figure}[htbp]
\centering
\includegraphics[width=14.5 cm]{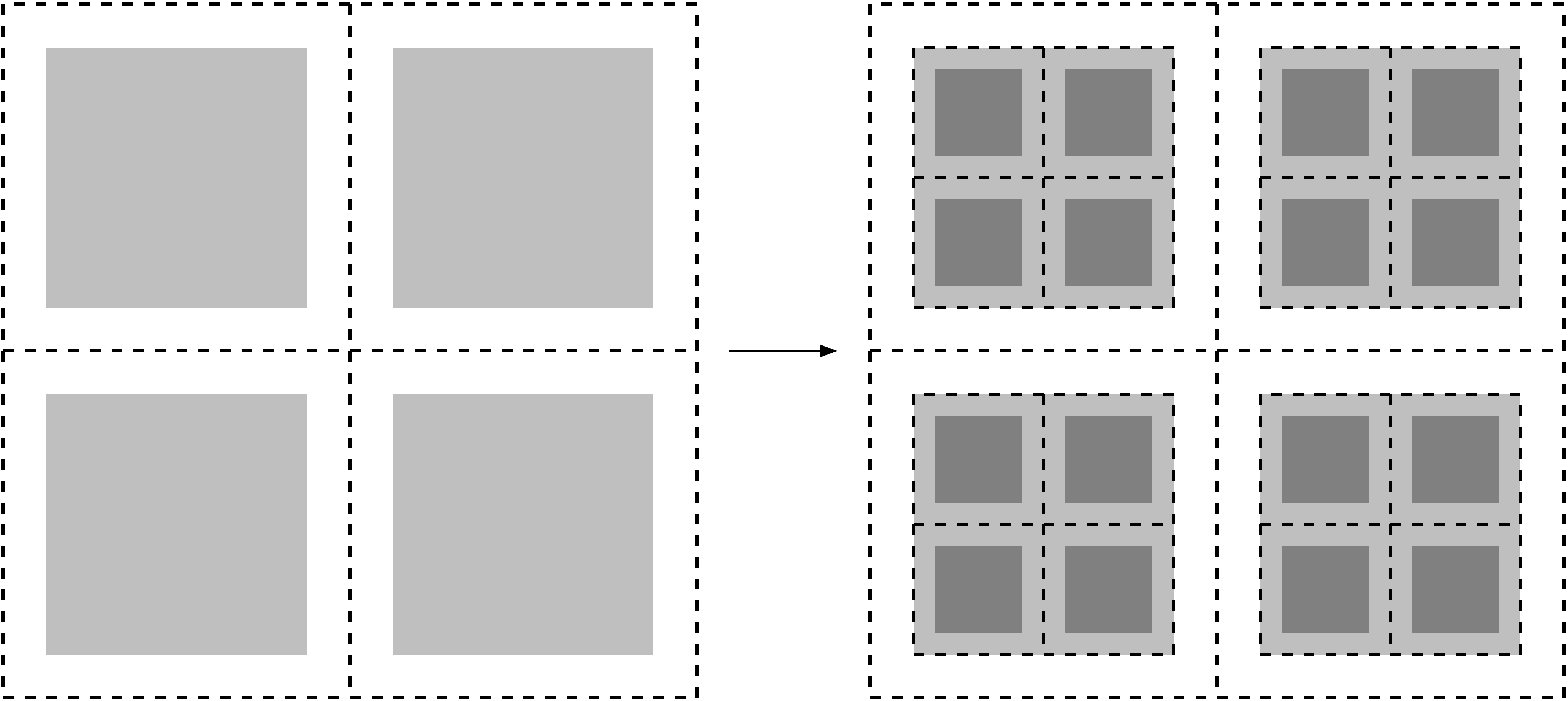}
\vskip-.2cm
\caption{\label{fig:multisc1} 
We draw two bisecting steps of the construction of the nested almost coverings of $\gL_N$, going say from level $j-1$ to $j$ and from level $j$ to $j+1$ in the preliminary version of the construction (that is before inversion: but of course we can view it the other way around ).
The dashed lines in the two figures mark a portion of the conditioning grid, but they are also the boundary of boxes:
notably on the left the largest cube (square in the figure) delimited by dashed line is one of the $\tilde B$ boxes at level $j-1$,
say $\tilde B_{j-1,k_0}$. This box contains $2^d$ boxes of level $j$:
they are denoted by $\tilde B_{j, k}$, $k\in \lint 2^d k_0+1, 2^d k_0+2^d\rint$. This operation is then repeated in the drawing in the right
and if we choose one of this boxes, say $\tilde B_{j, k_1}$, the next level boxes inside $\tilde B_{j, k_1}$ are 
$\tilde B_{j+1, k}$ with $k\in \lint 2^d k_1+1, 2^d k_1+2^d\rint$. 
}
\end{figure}

Now we introduce for $j=0,1, \ldots, J$ the decreasing sequence of sets (that are all unions of hypercubes)
with $\tilde D_0:= \gL_N$ and for $j=1, \ldots, J$
\begin{equation}
\tilde D_j\, :=\, \bigcup_{\tilde B \in \tilde \cC_j} \tilde B\, .
\end{equation}
Now we reverse the order by introducing for $j=0, \ldots , J$
\begin{equation}
N_j\, :=\, \tilde N_{J-j},  \ \ \ \    B_j\, :=\, \tilde B_{J-j}  , \ \ \ \ \ D_j:=  \tilde D_{J-j} , 
 \ \ \ \    G_j\, :=\, \tilde G_{J-j},  \ \ \ \    \bbG_j\, :=\, \tilde \bbG_{J-j} \, .
\end{equation}

Note that with this order reversing,  $N_j$ 
is close to $2^j N_0$ and $N_0$ does not depend much on $N$ (cf.\ \eqref{tildenj}), namely $h^{2/d}N_0/ (\log (1/h))^{2/d}\in [7,15]$. 
Furthermore, with our construction, the fraction of $\gL_N$ which is not covered by level zero boxes $B_{0,k}$ is small.
The content of this paragraph is made more precise and quantitative by:

\medskip

\begin{lemma}
\label{th:encadre}
With the notations specified above, for every $\varkappa>0$ there exists  $C_\varkappa>0$ and $h_0>0$ such that for every $j\ge 1$
and every  $h\in (0, h_0]$
we have 
\begin{equation}
\label{encadre}
 2^{j}N_0\, \le\, N_j \, \le\,  (1+ C_{\varkappa} N^{-\varkappa}_0) 2^{j}N_0\, . 
\end{equation}
In particular we have for  $C'_{\varkappa}=2dC_\varkappa$
\begin{equation}\label{conseq}
 \frac{\left \vert \bigcup_{k=1}^{2^{dJ}}B_{0,k}\right\vert }{N^d} \ge \left(1-   C'_{\varkappa} N_0^{-\varkappa}\right)  \, .
\end{equation}
\end{lemma}

\medskip

\begin{proof}
 The lower bound in \eqref{encadre} is immediate since \eqref{larecurr} implies $N_j\ge 2 N_{j-1}$.
 As for the upper bound in \eqref{encadre}, 
by definition $N_0= \tilde N_J\ge 7((1/h) \log (1/h))^{2/d}$ (recall \eqref{tildenj}),   \eqref{larecurr} implies that,  if $h$ is sufficiently small, for every $j\ge 1$
we have 
 \begin{equation}
 N_{j}\le  2N_{j-1}+  10 (N_{j-1})^{1-\varkappa}\, =\,  2 N_{j-1}\left( 1 + 5 N_{j-1}^{-\varkappa} \right)\,.
 \end{equation}
Hence iterating and using the lower bound we obtain
\begin{equation}
 \frac{N_j}{N_0}\le  2^j \prod_{i=1}^j \left(  1+ 5 N^{-\varkappa}_{i-1} \right)
 \le 2^j \prod_{i=0}^\infty \left(  1+ 5 (2^{i}N_0)^{-\varkappa} \right) \, \le \, 2^j \left(1+ C_\varkappa N_0^{-\varkappa}\right)\, ,
\end{equation}
with $C_\varkappa=6/(1-2^{-\varkappa})$. Therefore \eqref{encadre} is proven.
 The inequality \eqref{conseq} comes from the fact that
from \eqref{encadre} we have

\begin{equation}
  \frac{\left\vert \bigcup_{k=1}^{2^{dJ}}B_{0,k}\right\vert}{N^d}
  = \left(\frac{2^{J} N_0}{N_J}\right)^d\ge (1+C_{\varkappa} N^{-\varkappa}_0)^{-d} \, , 
\end{equation}
from which the result follows. 
\end{proof}

\medskip

\noindent Finally, we divide the level-$0$ boxes, whose edge length $N_0$ satisfies (because of \eqref{tildenj}) 
\begin{equation}
\label{eq:0boxsize}
7(\log(1/h)/h)^{2/d} \,\le\,   N_0 \,\le\,   15( \log(1/h)/h)^{2/d}\, ,
\end{equation} 
 into $6^d$ cubes of edge length 
 $\lfloor N_0/6 \rfloor$. More precisely $\overbar{B}_{6^dk+i}$, $i=1,\dots,6^d$ are obtained by dividing $B_{0,k}$, see Fig.~\ref{fig:6d}. We set $\overbar{N}_h=\lfloor N_0/6 \rfloor$ and set $\overbar \cC_0:=\{ \overbar{B}_l, l\in \lint 1, 6^d 2^{dJ}\rint \}$.  Note that \eqref{conseq} is also valid for $\bigcup_l \overbar{B}_{l}$, possibly increasing the value of $C'_\varkappa$.
Moreover $\overbar{N}_h$ depends on $N$ only mildly: in fact  from  \eqref{tildenj} we have
\begin{equation}
\label{eq:barboxsize}
(\log (1/h)/h)^{2/d}\,  \le \,  \overbar{N}_h \le 3 ( \log (1/h)/h)^{2/d} \, .
\end{equation}
We refer to $\overbar{B}_l$ as an \textit{elementary box} and to $B_{0,k}$ as a  level-$0$ box.

We let 
\begin{equation}
K\, := \, 2^{dJ} \ \textrm{ and } \overbar{K}\,:=\,  6^d 2^{dJ}\, , 
\end{equation}
denote the number of  as level-$0$ boxes and elementary boxes respectively.

\subsection{Control of the contact density in elementary boxes}
\label{density}

For $k\in \lint 1, \overbar{K} \rint$ let us define $\zeta(k)$ to be the contact fraction inside $\overbar{B}_{k}$.
\begin{equation}\label{lazeta}
 \zeta(k):= \frac 1{\overbar{N}_h^{d}}\sum_{x\in \overbar{B}_{k}} \ind_{[-1,1]}(\phi(x))\, ,
 \end{equation}
and let $\cC(N,\gd)$ be the event that most boxes have a contact fraction \emph{reasonably} close to the optimal value $p_h=2\chi(\beta) h$
\begin{equation}
\label{eq:cCset}
\cC(N,\gd):=\left\{ \#\left\{ k\in \lint 1, \overbar{K}\rint  \ : \ \frac{\zeta(k)}{\chi(\beta) h}\not \in [1, 3]  \right\} \,\le\,   \gd \overbar{K} \right\}\, . 
\end{equation}
In the end $\gd$ will be chosen proportional to $\gep$. 
Our first step is to prove that $\cC(N,\gd)$ has  probability close to one.

\medskip

\begin{lemma}\label{locadenz}
Choose an arbitrary value of $\gd>0$. Then there exists $h_0>0$ such that for $h \in (0, h_0]$ there exists $c=c(h)>0$
and $N_{h}>0$ such that for  $h \in (0, h_0]$ and $N\ge N_{h}$ we have
\begin{equation}
 \label{eq:step1p}
\bbE \left[ \bP_{N, \go, \gb, h} \left( \cC(N,\gd)^{\cc} \right) \right]\, = e^{-c(h)\gd N^d}.
\end{equation}
\end{lemma}
\medskip

\begin{rem}
 The result would be valid also replacing $[1,3]$ 
 by an arbitrarily  small interval centered at $2$, but this is useless for the rest of the proof and, with our choice, we   avoid introducing one more parameter.
\end{rem}

\medskip

\begin{proof}
 We are going   to prove an upper bound on $Z_{N, \go, \gb, h}\left( \cC(N,\gd)^{\cc} \right)$ which shows that it is typically much smaller than $Z_{N, \go, \gb, h}$.
 More precisely, by the Markov inequality it is sufficient to show that for every $h>0$ small there exists $\theta\in(0,1)$  such that  
 \begin{equation}
 \limsup_{N\to \infty} \frac{1}{\theta N^d}\log \bbE\left[Z_{N, \go, \gb, h}\left( \cC(N,\gd)^{\cc} \right)^{\theta}\right]<  \tf(\gb,h).
 \end{equation}
 We decompose $Z_{N, \go, \gb, h}\left( \cC(N,\gd)^{\cc} \right)$ according to the position of atypical density blocks.
 Given $I\subset \lint 1,\overbar{K}\rint$ we set 
  \begin{equation}
 \cA_I\, :=\,\left\{ \left\{ k \ : \ \left\vert \zeta(k)- 2 \chi(\beta) h \right\vert\, \le\, {\chi(\gb)} % \gd 
 h \right\}=I  \right\}\, ,
   \end{equation}
   so $ \cC(N,\gd)^{\cc}$ coincides with the union of the events $\cA_I$ with $\vert I \vert \ge \gd \overbar{K}$. 
 We then observe that 
  for $\theta \in (0,1]$
 \begin{equation}
  Z_{N, \go, \gb, h}\left( \cC(N,\gd)^{\cc} \right)^{\theta}
  \le \sum_{\{ I \ :  \ |I| \ge \gd \overbar{K}\}}  Z_{N, \go, \gb, h}\left( \cA_I \right)^{\theta}\, ,
 \end{equation}
 where we have used the elementary inequality $(\sum_j a_j)^\theta \le \sum_j a_j^\theta$ that holds for $a_j>0$. 
Therefore 
\begin{equation}
  \bbE\left[ Z_{N, \go, \gb, h}\left( \cC(N,\gd)^{\cc} \right)^{\theta}\right] \le 2^{\overbar{K}} \max_{|I| \ge \gd \overbar{K}  } \bbE \left[ Z_{N, \go, \gb, h}\left( \cA_I \right)^{\theta} \right].
\end{equation}
Now we can use the fact that $\overbar{K}\le (N/\overbar{N}_h)^d$ (simply because the boxes are disjoint), and we obtain that 
\begin{equation}
 \frac{1}{\theta N^d}  \log\bbE\left[ Z_{N, \go, \gb, h}\left( \cC(N,\gd)^{\cc} \right)^{\theta}\right]\le  \frac{\log 2}{\theta \overbar{N}^{d}_h }+      \frac{1}{\theta N^d} \max_{|I| \ge\gd  \overbar{K}  } \log \bbE \left[ Z_{N, \go, \gb, h}\left( \cA_I \right)^{\theta} \right].
\end{equation}
If one sets $\theta=\theta_h:= (\log 1/h)^{-1}$, the first term on the right-hand side \ is $O( h^2 (\log 1/h)^{-1})$. 
Hence we can conclude if we show that there exists $c>0$ such that  for $h$ sufficiently small
\begin{equation}
\label{aprouver0}
 \frac{1}{\theta_h N^d} \max_{|I| \ge  \gd \overbar{K}  } \log \bbE \left[ Z_{N, \go, \gb, h}\left( \cA_I \right)^{\theta_h} \right] \le  \chi(\beta)(1- c \gd)h^2\, .
\end{equation}
% after taking $\log$ and dividing by $\theta N^d$, the following additive term
% \begin{equation}
%  \theta^{-1} M^{-d} \log 2\, =\,  \theta^{-1} h^2 (\ell(h))^{-2}\log 2\, .
%  \end{equation}
%  Taking $\theta= \ell(h)^{-1}$ we see that this term is  $o(h^2)$. Hence we only have to find a good bound for $\bbE \left[ Z_{N, \go, \gb, h}\left( \cA_I \right)^{\theta} \right]$.
For this we use the following technical lemma, whose proof is postponed till the end of the proof we are developing.
 \medskip

 \begin{lemma} 
 \label{th:forub}
  Recall that
 $\xi_x=e^{h+\gb\go_x-\gl(\gb)}-1$.
  For any measure on  $\bP_{\gL}$ on $\{0,1\}^{\gL}$ and every $\theta \in (0, 1)$ 
we have for every $h\ge 0$ that
\begin{equation}
\label{eq:forub1}
\bbE\left[ \bE_{\gL}\left[\prod_{\{x \ : \ \gd_x=1\}} (1+ \xi_x) \right] ^\theta\right] \le \left(\max_{p\in(0,1]}\bbE\left[ (1+p\xi)^{\theta} \right]\right)^{\vert \gL\vert}.
\end{equation}
In the limit where $\theta$ and $h$ tend to zero we have
\begin{equation}
\label{padrao}
 \frac{1}{\theta}\max_{p\in(0,1)} \log \bbE\left[ (1+p\xi)^{\theta} \right]= \chi(\gb)h^2+O(h^2\theta)+O(h^3).
\end{equation}
If furthermore  $\bP_{\gL}$ is such that for $\eta>0$ and  $h>0$
\begin{equation}\label{bbloc}
\bP_{\gL} \left( 
\left| \frac{1}{|\gL|} \sum_{x\in\gL}\gd_x-2\chi(\gb)h\right|\ge \eta h \right) \,=\, 1 \,,
\end{equation}
then there exist  $C=C_\gb>0$, $h_0=h_{0,\gb , \eta}$ and  $\theta_0=\theta_{0,\gb}$  such that for all $h\in (0, h_0]$ and $\theta\in (0, \theta_0]$ and for all $\gL$
\begin{equation}
\label{eq:forub2}
\frac{1}{\theta |\gL|}\log \bbE\left[ \bE_{\gL}\left[\prod_{\{x \ : \ \gd_x=1\}}(1+ \xi_x) \right] ^\theta\right]\,  \le\,  
\chi(\gb)\left(1- \frac{\eta^2}{4 \chi(\gb)^2} \right)h^2 + \frac{\log 2}{\theta |\gL|} + C(h^2\theta+ h^3) \, .
\end{equation}  
 \end{lemma}
 
 \medskip
 
 \begin{rem}
 The inequality \eqref{eq:forub2} is valid for all size of boxes. However it provides a better bound than \eqref{padrao} only if  $\theta |\gL|$ is much larger  than $h^{-2}$.
 \end{rem}

 Consider $I\subset \lint 1,\overbar{K}\rint$ with $\vert I\vert \ge \gd \overbar{K} $. 
 Let $k_1\le  \dots\le  k_{|I|}$ denote the elements of $I$.
 We can prove by induction that for every $j\le |I|$ 
 \begin{equation}
 \label{leone}
 T_1(I, j)\, :=\,
  \frac{1}{\theta_h} \log \bbE \left[ \left(\bE_N \left[ e^{\sum_{x\in \bigcup_{1\le i\le j} \overbar{B}_{k_i}} (\beta \go_x-\gl(\beta)+h)\gd_x }    \ \bigg\vert \ \cA_I \right]\right)^{\theta_h}\right]\, \le\,  \frac j2 \chi(\gb)h^2 \overbar{N}_h^d.
 \end{equation}
 The result for $j=1$ is a direct consequence of \eqref{eq:forub2}: we work with $\eta=\chi(\gb)$ and then we use that  
 $ 1/(\overbar{N}_h^{d} \theta_h)= O(h^2/ \log (1/h))\ll  h^2$,
 so that the term $(\log 2)/(\theta_h|\gL|)$, as well as the term $C(h^2 \theta_h +h^3)$, can be absorbed into the leading order, for $h$ sufficiently small. 
 Now for $j> 1$, we just need to apply \eqref{eq:forub2} to the measure 
 $\mu$ defined by 
 \begin{equation}
  \mu(A):= \frac{\bE_N \left[ e^{\sum\limits_{x\in \left(\cup_{1\le i\le j-1} \overbar{B}_{k_i}\right)} (\beta \go_x-\gl(\beta)+h)\gd_x } \ind_{A}   \ \bigg\vert \ \cA_I \right]}{\bE_N \left[ e^{\sum\limits_{x\in \left(\cup_{1\le i\le j-1} \overbar{B}_{k_i}\right)} (\beta \go_x-\gl(\beta)+h)\gd_x }   \  \bigg\vert \ \cA_I \right]}\, . 
 \end{equation}
We obtain that 
\begin{multline}
\log  \bbE\left[ \left(\frac{\bE_N \left[ e^{\sum\limits_{x\in \left(\cup_{1\le i\le j} \overbar{B}_{k_i}\right)} (\beta \go_x-\gl(\beta)+h)\gd_x }  \ \ \bigg\vert\ \cA_I\right]}{\bE_N \left[ e^{\sum\limits_{x\in \left(\cup_{1\le i\le j-1} \overbar{B}_{0,k_i}\right)} (\beta \go_x-\gl(\beta)+h)\gd_x }    \ \ \bigg\vert \ \cA_I \right]}\right)^{\theta_h}\right]\\
  \\ =  \log  \bbE\left[ \mu\left(e^{\sum_{x\in \overbar{B}_{k_j}} (\beta \go_x-\gl(\beta)+h)\gd_x } \right)^{\theta_h}\right]\, 
   \le\,  \frac 12\theta_h \overbar{N}^d_h \chi(\gb)h^2
\end{multline}
where the inequality is obtained by applying  \eqref{eq:forub2}, in the same way as for the case $j=1$, when averaging with respect to $(\go_x)_{x\in \overbar{B}_{k_j}}$. This completes the induction argument and \eqref{leone} is proven.
 
Using the same trick and \eqref{eq:forub1}-\eqref{padrao} we obtain that 
\begin{multline}
\label{letwo}
T_2(I)\, :=\,
 \frac{1}{\theta_h} \log \bbE\left[ \left(\frac{\bE_N \left[ e^{\sum_{x\in \mathring \gL_N} (\beta \go_x-\gl(\beta)+h)\gd_x }  \ \ \Big\vert \ \cA_I\right]}{\bE_N \left[ e^{\sum_{x\in \bigcup_{i\in I} \overbar{B}_{k_i}} (\beta \go_x-\gl(\beta)+h)\gd_x }    \ \ \Big\vert \ \cA_I \right]}\right)^{\theta_h}\right]\\
 \le \,
  \left[(N-1)^d - |I| \overbar{N}^d_h \right]  (\chi(\gb) h^2 + C \theta_h h^2)\, .
\end{multline}

Finally, 
combining \eqref{leone} and \eqref{letwo}, by apply the Cauchy-Schwarz inequality  we obtain 
\begin{equation}
\begin{split}
  \frac{2}{\theta_h} \log \bbE \left[ Z_{N, \go, \gb, h}\left( \cA_I \right)^{(\theta_h/2)} \right]
  & \le \,  \frac{2}{\theta_h} \log  \bbE\left[ \left({\bE_N \left[ e^{\sum_{x\in \mathring \gL_N} (\beta \go_x-\gl(\beta)+h)\gd_x }  \ \ \Big\vert \ \cA_I\right]}\right)^{\theta_h/2}\right]
  \\
  &\le \, T_1(I, \vert I \vert) + T_2(I)\, \le \,  \chi(\gb)h^2 \left(N^d -\frac 12   |I| \overbar{N}^d_h\right) +C \theta_h h^2 N^d 
  \\ & \le \, 
  \left( 1-\frac \gd 3\right) \chi(\gb)h^2 N^d\, ,
  \end{split}
\end{equation}
 which is  \eqref{aprouver0} with $c=1/3$, except for the (clearly irrelevant) fact that $\theta_h$ is replaced by $\theta_h/2$.
 \end{proof}

 \medskip
 
\begin{proof}[Proof of Lemma \ref{th:forub}]
The inequality \eqref{eq:forub1} can be proven using the approach that lead to  \eqref{thegeneral}, but we give here a different proof.
We proceed  by induction on $N=|\gL|$. Note that we can assume without loss of generality that $\gL:=\lint 0,N\rint$ and write $\bE_N$ for $\bE_{\gL}$.
The result is obvious for $N=0$. Now given $N\ge 1$ set 
\begin{equation}
p_{N}(\xi_1,\dots, \xi_{N-1}):= \frac{\bE_{N}\left[\prod_{\{x\in\lint 1,N-1\rint \ : \ \gd_x=1\}}\left(1+\xi_x \right) \gd_N \right]}{\bE_{N}\left[\prod_{\{x\in\lint 1,N-1\rint \ : \ \gd_x=1\}} \left(1+\xi_x\right) \right]}\ ,
\end{equation}
together with 
\begin{equation} 
Z_N(\xi):= \bE_N \left[\prod_{\{x\in\lint 1,N\rint \ : \ \gd_x=1\}}(1+\xi_x) \right] \ \ \ \text{ and } \ \  \  \tilde Z_{N}(\xi):= \bE_N \left[\prod_{\{x\in\lint 1,N-1\rint \ : \ \gd_x=1\}}(1+\xi_x) \right].     
\end{equation}
Note that
\begin{equation}
 Z_N(\xi)
\,=\,  \tilde Z_{N}(\xi) (1+p_N \xi_N)\, ,
\end{equation}
and by raising both sides to the 
 power $\theta$ and taking the average with respect to $\xi_N$ we obtain
\begin{equation}
%\begin{split}
\bbE \left[ Z^{\theta}_N(\xi) \ \big | \ (\xi_{x})_{1\le x\le N-1}\right]= \tilde  Z_{N}^{\theta}(\xi) \bbE\left[ (1+p_N \xi_N)^{\theta}  \right]
\le
\tilde Z_{N}^{\theta}(\xi)  \max_{p\in(0,1] }\bbE\left[ (1+p \xi)^{\theta} \right]\, ,
%\end{split}
\end{equation}
and by  taking the average with respect to all other variables on both sides 
we obtain that 
\begin{equation}
\bbE \left[ Z^{\theta}_N(\xi)\right] \, \le \, \bbE \left[ \tilde Z^{\theta}_N(\xi)\right]  \max_{p\in(0,1] }\bbE\left[ (1+p \xi)^{\theta} \right]\, .
%\, =\, \sup_{\bP_{N-1}} \bbE \left[  Z^{\theta}_{N-1}(\xi)\right]  \max_{p\in(0,1] }\bbE\left[ (1+p \xi)^{\theta} \right]
\end{equation} 
We can now take the supremum over $\bP_N$ and we can conclude the induction step because 
\begin{equation}
\sup_{\bP_N}\bbE \left[ \tilde Z^{\theta}_N(\xi)\right] \, =\,\sup_{ \bP_{N-1}} \bbE \left[  Z^{\theta}_{N-1}(\xi)\right]\, .
\end{equation}
The bound \eqref{eq:forub1} is therefore established.

%\begin{equation} \bbE[ \overbar Z_{N}^{\theta}] \le \left(\max_{p\in(0,1] }\bbE\left[ (1+p \xi)^{\theta} \right]\right)^{N-1}.\end{equation}
\medskip

As for the optimizing problem \eqref{padrao} we can proceed as for Proposition~\ref{th:final-ub} using Taylor expansion and showing first that $p_{h,\theta}$ tends to $0$, then that $p_{h,\theta}\sim 2 \chi(\gb) h$, then $p_{h,\theta}= 2 \chi(\gb) h+ O(\theta h) +O(h^2)$ and then the final claim. We do not detail  these steps. 

\medskip

For what concerns \eqref{eq:forub2} (recall that we are therefore assuming \eqref{bbloc}), at the cost of loosing a factor two in the estimate on $\bbE[ Z^{\theta}_N]$ we can assume 
that one of the following holds
\begin{equation}
  \bP_{N} \left( 
\sum_{x=1}^N\gd_x  \ge  (2 \chi(\beta)+\eta) h N \right) \,=\, 1 \, \text{ or }
  \bP_{N} \left( 
\sum_{x=1}^N\gd_x  \le  (2 \chi(\beta)-\eta) h N \right) \,=\, 1 \,.
\end{equation}
Using H\"older inequality we have that for every positive random variable $g$
\begin{equation}
\label{eq:forforub1}
\bbE\left[ Z_{N}^\theta\right]\,
 \le\, 
  \bbE[Z_{N} g^{\theta-1}]^{\theta}\bbE[g^\theta]^{1-\theta}\, .
 \end{equation}
 Now we set $g=\prod_{x\in \lint 1, N\rint }(1+ q \xi')^\frac{1}{1-\theta}$ with $\xi'= e^{\gb\go -\gl(\gb)}-1$ (using $\xi$ instead of $\xi'$ would give an analogous result but computations are easier with $\xi'$).
 We are going to set $ q=q_{\pm}:= (2 \chi(\beta)\pm \eta)h$ depending on which assumption we have on the contact fraction.

The following asymptotic statements hold in the limit where both $q$ and $\theta$
go to zero
\begin{equation}
\begin{split}
 \bbE\left[(1+ q \xi')^{-1}\right]&=1+q^2 \Var (\xi')+O(q^3),\\
  \bbE\left[(1+\xi')(1+ q \xi)^{-1}\right]&=1-q\Var(\xi')+O(q^2),\\
  \bbE\left[(1+ q \xi')^{\frac{\theta}{1-\theta}}\right]&=1-\frac{\theta(1-2\theta)}{2(1-\theta)^2}\Var(\xi')q^2+O(\theta q^3).
 \end{split}
\end{equation}
Note in particular performing a second order expansion in $h$ that for fixed $\eta$ with our choice of $q$  we have for $h$ sufficiently small (depending 
on $\eta$ and on $\gb$) 
\begin{equation}\label{lesformols}
\begin{split}
 e^h\bbE\left[(1+\xi')(1+ q_- \xi')^{-1}\right]\, &\ge\,  \bbE\left[(1+ q_- \xi')^{-1}\right],\\
e^h  \bbE\left[(1+\xi')(1+ q_+ \xi')^{-1}\right]\, &\le\,  \bbE\left[(1+ q_+ \xi')^{-1}\right].
\end{split}
\end{equation}
% where the inequality of course holds for $h$ small. 
If $\sum_{x\in \lint 0,N\rint} \gd_x \le N q_-$, $\bP_N$-a.s.\ we can replace $\sum \gd_x$ by its upper bound $q_-|\gL|$ and altogether we obtain
\begin{multline}
\bbE\left[ Z_{N}^\theta\right] \le  \left( e^h\bbE\left[(1+\xi')(1+ q_- \xi)^{-1}\right] \right)^{q_-N \theta} 
\\ \times \left(\bbE\left[(1+ q_- \xi)^{-1}\right] \right)^{(1-q_-)N \theta} \left( \bbE\left[(1+ q_- \xi')^{\frac{\theta}{1-\theta}}\right]\right)^{N(1-\theta)}.
\end{multline}
Hence applying \eqref{lesformols} we obtain 
\begin{multline}
 \frac{1}{N\theta}\log  \bbE\left[ Z_{N}^\theta\right]
 \le q_-h -   \frac{q^2_-}{2}\Var(\xi')+O(\theta q^2)+O(q^3)\\
 =\, \left( \chi(\gb)-\frac{\eta^2}{4\chi (\beta)}\right) h^2+  O(h^2\theta)+O(h^3).
\end{multline}
In the same manner when  $\sum_{x\in \lint 0,N\rint} \gd_x \ge N q_+$, $\bP_N$-a.s. we have 
\begin{multline}
\bbE\left[ Z_{N}^\theta\right] \le  \left( e^h\bbE\left[(1+\xi')(1+ q_+ \xi)^{-1}\right] \right)^{q_+N \theta} 
\\ \times \left(\bbE\left[(1+ q_+ \xi)^{-1}\right] \right)^{(1-q_+)N \theta} \left( \bbE\left[(1+ q_+ \xi')^{\frac{\theta}{1-\theta}}\right]\right)^{N(1-\theta)}.
\end{multline}
and we obtain
\begin{multline}
 \frac{1}{N\theta}\log  \bbE\left[ Z_{N}^\theta\right]
 \le q_+ h -   \frac{q^2_+}{2}\Var(\xi')+O(\theta q^2)+O(q^3)\\
 =\, \left( \chi(\gb)-\frac{\eta^2}{4\chi (\beta)}\right) h^2+  O(h^2\theta)+O(h^3).
\end{multline}
This completes the proof of Lemma \ref{th:forub}.
\end{proof}

\subsection{The main body  of the argument (Proof of Proposition~\ref{lescontactzB})}
\label{sec:Decomp}

We start by an important preliminary result. 
By using the control we have on the local contact density and on the free energy, we are going to show that if 
the probability of a sequence of  events  decays exponentially in the volume size under $\bP_N$ with a  rate suitably controlled from below, then it also decays exponentially under  $\bP_{N, \go,\beta, h}$. Said differently, the result we are going to state allows to neglect the environment and to reduce the estimates to Gaussian field estimates.

\medskip

\begin{lemma}
\label{th:ctrlP}
Recall $\gep_d>0$ from  Propositon~\ref{lloobb}. Consider for $h\le h_0$,  $N$ larger than a suitable $N_0(h)$  and a sequence of
 events $A_N$ such
that $\bP_N(A_N) \le \exp(-h^{2+\eta}N^d)$, with $\eta \in (0, \gep_d)$.
Then there exists  $h'_0>0$ and $N'_0(h,\go)$ such that for every $h \le h'_0$ and $N\ge N'_0(h)$
\begin{equation}
 \bP_{N, \go,\beta, h} \left (A_N\right) \,\le \,  \exp(-N^d h^{2+ \eta}/4)\, .
\end{equation} 
\end{lemma}

\medskip

\begin{proof}
Because $N^{-d}\log Z_{N, \go,\beta, h}$ converges, it is sufficient to show that (recall \eqref{restriparti})
\begin{equation}
 \frac{1}{N^d}\log Z_{N, \go,\beta, h} \left (A_N\right)\, 
 \le\, 
  \tf(\beta,h)-h^{2+ \eta}/3 \, .
\end{equation}
We are going to show first that 
\begin{equation}\label{lespe}
   \frac{1}{N^d}\bbE \left[ \log Z_{N, \go,\beta, h} \left (A_N\right) \right] \le \tf(\beta,h)-h^{2+ \eta}/2.
\end{equation}
To see this, we note that we have (using Proposition \ref{th:ub} for $\bP_N(  \ \cdot \  | \ A_N)$ and keeping into account 
Remark~\ref{rem:forfuture})
\begin{multline}\label{sadsdf}
   \bbE \left[ \log Z_{N, \go} \left (A_N\right) \right]= \log \bP_N(A_N)
   + \bbE \log \bE_N \left[ e^{\sum_{x\in \mathring \gL_N}(\beta \go_x-\gl(\beta)+h)\gd_x} \ \big | \ A_N\right] \\ \le \, 
   - N^d h^{2+\eta}+  |\mathring \gL_N| (h^2\chi(\beta)+C_\beta h^3) \, .
\end{multline}
We have used $Z_N$ as a shortcut for $Z_{N, \go, \gb, h}$ and we recall that 
$\mathring{\gL}_N = 
\lint 1,N-1 \rint^d$, hence $ |\mathring \gL_N|\le N^d$.
Now in view of Proposition \ref{lloobb}, \eqref{sadsdf} implies that \eqref{lespe} 
holds for $h$ sufficiently small.

% On the other hand from Proposition \ref{lloobb} we have there exists $C'_\beta>0$ such that
%    \begin{equation}
%    \label{dsfsd}
%    \bbE \left[ \log Z_{N, \go}  \right]
%    \ge N^d (h^2\chi(\beta)- C'_\beta h^{2+\gep_d}).
%    \end{equation}
% Subtracting  \eqref{dsfsd} from \eqref{sadsdf} we obtain \eqref{lespe}.
Now to conclude we only need to show that $\log Z_{N, \go}(A_N)$ is concentrated around its mean. This follows from a standard concentration argument, for which we introduce a cut-off $M_N:=  N^{d/6}$ on the disorder variables 
$\overbar \go_x:= \go_x \ind_{\{|\go_x|\le M_N\}}$.
A union bound and the finiteness   of the exponential moment of all orders for $\go$ implies that for $N$ sufficiently large we have
\begin{equation}\label{diff}
\bbP\left( \exists x\in \gL_N, \ |\go_x| \ge M_N\right)\, \le\,  e^{-N^{d/6}}\, .
\end{equation}
Therefore,
 by Azuma's inequality applied to the martingale
$(\bbE[ \log \overbar Z_{N,\overbar \go}(A_N) \ | \ \cG_n ])^{|\mathring \gL_N|}_{n=0}$, where 
$\cG_n:=\sigma( \go_{x_i}, i\le n)$ and $x_1,\dots, x_{|\mathring \gL_N|}$ is an arbitrary enumeration of $\mathring \gL_N$, we obtain that there
\begin{equation}
\bbP\left(  |\log  Z_{N,\overbar \go}(A_N) - \bbE[ \log  Z_{N,\overbar \go}(A_N) ]|\ge u \right)\,  \le\,  
e^{-\frac{ c u^2}{N^d M^2_N}}\, ,
\end{equation}
with $c=1/(\gb+ \gl (\gb)+h)^2$. 
Applying this for $u=N^{3d/4}$ and using \eqref{diff} to bound the probability of
$\{Z_{N,\go}(A_N)\ne\{Z_{N,\overbar{\go}}(A_N)|\} $ we obtain (for $c'=c\wedge 1$)
\begin{equation}
 \bbP\left(\left \vert\log Z_{N,\go}(A_N)-\bbE[ \log  Z_{N,\overbar \go}(A_N) \right \vert > N^{3d/4}\right)\le 2 e^{-c' N^{d/6}}
\end{equation}
At this point wee can  conclude by observing that the difference between $\bbE[ \log  Z_{N,\overbar \go}(A_N) ]$ and 
$\bbE[ \log  Z_{N, \go}(A_N) ]$ is small: this follows by applying Lemma~\ref{th:theboundz} with $K=M_N$ that
yields 
\begin{equation}
\bbE \left[
\left \vert \log  Z_{N, \go}(A_N)-
\log  Z_{N,\overbar \go}(A_N)
\right \vert \right]\, \le \, \gb N^d \bbE \left[ \vert \go_x \vert ;\, \vert \go_x \vert > M_N\right] \, \le \,  \exp\left( -N^{d/6}\right)\, ,
\end{equation}
for $N$ sufficiently large,
again because all exponential moments of $\go$ are finite. This completes the proof of Lemma~\ref{th:ctrlP}.
\end{proof}

\medskip 

\noindent
\emph{Proof of Proposition~\ref{lescontactzB}.}
We aim at applying 
Lemma \ref{th:ctrlP} so to  reduce the proof of Proposition~\ref{lescontactzB} to proving a statement about 
$\bP_N$. However this fails if applied directly to $\cD(N,\gep)$ because
    $\bP_N(\cD(N,\gep))$ does not decrease exponentially in the volume (due to the massless character of the 
    LFF, see e.g. \cite{cf:BDZ}).
 However,    we will show that the event $\cC(N,\gd) \cap \cD(N,\gep)$ does not have this drawback and
  Lemma \ref{locadenz} assures that we can limit ourselves to studying this event. More formally, thanks to 
  Lemma \ref{locadenz} and Lemma \ref{th:ctrlP} it suffices to show that
 given $\gep>0$ there exist $\gd>0$ and  $\eta<\gep_d$ such that for $h$ sufficiently small for all $N$ sufficiently large, we have 
 \begin{equation}
 \label{eq:propoz} 
 \frac{1}{N^d}\log \bP_N\left( \cC(N,\gd) \cap \cD(N,\gep)\right)\, 
 \le\,   - h^{2+\eta}.
 \end{equation}

To prove \eqref{eq:propoz}  we are going to use a decomposition of $\phi$ into
a field $\psi$ which is independent in each level-$0$ boxes, and a field $\overbar \psi$ which displays long range correlation but has a small amplitude.

\medskip

Recall that $\bbG_0$ denotes a grid which isolates each of the boxes at level zero of the hierarchy.
By the Markov property of the LFF,
we can write 
\begin{equation}
\label{eq:2psi}
 \phi =  \phi_{\gL_N} \stackrel{\text{(law)}}{=} \psi +\overbar \psi \, ,  
 \end{equation}
where $\overbar \psi$ is the harmonic continuation of the restriction of $\phi$ to $\bbG_0$, that is the solution of 
\begin{equation}
\begin{cases}
 \Delta \overbar\psi(x)= 0, \text{ for all } x\in \mathring \gL_N \setminus \bbG_0,\\
  \overbar\psi(x)=\phi, \text{ for } x \in \bbG_0\, ,
  % \overbar\psi(x)=0, \text{ for } x \in \partial \gL_N\, .
   \end{cases}
\end{equation}
and $\psi$ a free field on $\mathring{\gL}_N \setminus \bbG_0$ with $0$ boundary condition: therefore $\psi$ is a collection of independent free fields on each of the level-$0$ boxes, and $\overbar \psi$ carries all the long range correlations of $\phi$.
We are going to reduce the proof of  \eqref{eq:propoz} to that of the two following facts about $\psi$ and $\overbar \psi$ (rigorously stated as Proposition \ref{smalllap} and Proposition \ref{unecas} below):
\begin{itemize}
 \item [(A)] With very large probability on most level-$0$ boxes $B_{0,k}$, the field $\overbar \psi$ is almost flat (in the sense that it is very close to an affine function). We mean by this that the probability of the complement is smaller than $e^{- N^d h^{b}}$ for a  value of $b<2$: this is largely sufficient because also a value  
 of $b$ slight larger than $2$ would have been sufficient (see Lemma~\ref{th:ctrlP}).
 \item [(B)] When $\overbar \psi$ is flat on the box $B_{0,k}$, the probability of having both the right number of contact (that is, about $2\overbar{N}^d_0 \chi(\beta)h$) in each of the $6^d$ elementary boxes inside $B_{0,k}$, and a density of high points is small.
\end{itemize}
To conclude from these two statements, we only need to use the fact that 
conditioned to $\overbar \psi$, the various level $0$ boxes are independent. 

\medskip

\begin{rem}
\label{rem:boundary-0}
The fact that the $0$ level boxes are separated from the grid of a distance equal (up to integer rounding) to $N_0^{1-\varkappa}$ and recalling that $N_0 \approx (\log(1/h)/h)^{2/d}$, more precisely 
\eqref{eq:0boxsize},  readily yields (use \eqref{eq:Gest0}) that the variance of $\psi_x$ can differ from  $\gs_d^2$ at most of a term $O(1/ N_0^{(d-2)(1-\varkappa)})$, 
that is by $O(h^c)$ with $c$ a positive constant.  
In view of the estimates we aim at, these variations of the variance turn out to be irrelevant: details will be given in due times but  
the reason is simply that for every $b>0$
\begin{equation}
\frac{P\left( (1+\epsilon(h))\cN \ge b \log (1/h) \right)} {P\left( \cN \ge b \log (1/h) \right)}-1 \, = \, O\left( h^{c'}\right)\, ,
\end{equation}
if $\epsilon(h)=O(h^c)$.
Moreover the distance from the conditioning grid is not only used at level-$0$, but at all levels. However for the higher levels it is used to  assure that 
harmonic extension  has a small variance uniformly in the box: the tool we use in that case is Lemma~\ref{th:var-est0lem}. 
\end{rem}

\medskip

Let us give a more quantitative version of $(A)$ and $(B)$.
First let us define for $e$ and $g$ two unit vectors (not necessarily distinct) of the canonical base of $\bbR^d$ 
$\nabla_{eg} \overbar \psi$ to  be the \textit{bi-gradient} of $\overbar \psi$ in directions $e$ and $g$
\begin{equation}
\label{eq:bigraddef}
\nabla_{eg} \overbar \psi(x)\, :=\,   \overbar \psi(x+e+g)-\overbar \psi(x+e)-\overbar \psi(x+g)+\overbar \psi(x)\,.\end{equation}
Now for $k\in \lint 1,K\rint$  (recall $K=2^{dJ}$ is the number of level-$0$ boxes) we set 
\begin{equation}
\|\nabla_{eg} \overbar \psi  \|_{\infty,k}\,:=\,  \max_{x\in B_{0,k}} |\nabla_{eg} \overbar \psi (x)|\, .
\end{equation}
Finally we let  $\cB(N,\gd)$ the event that all bi-gradients are small on most boxes, that is 
\begin{equation}
\label{eq:cBN}
 \cB(N,\gd)\, :=\, \left\{ \left \vert \{ k\in \lint 1,K\rint \ : \ \|\nabla_{eg} \overbar \psi  \|_{\infty,k} \ge N^{-2}_0  \}\right\vert \le  \gd K \right\} \, .
\end{equation}
Here is the result we need on the even $\cB(N,\gd)$: it is proven in Section~\ref{sec:multiscale}.

\medskip

\begin{proposition}
\label{smalllap}
For any $\gd>0$ and  for $h$ sufficiently small
\begin{equation}
 \bP_N\left( \cB(N,\gd)^{\cc} \right)\,\le\,  e^{- N^d h^{11/6}}\, .
\end{equation} 
\end{proposition}
\medskip

By putting together Lemma~\ref{th:ctrlP} 
and Lemma~\ref{smalllap} we have that for every $\eta>0$, $h$ sufficiently small and $N \ge N_0(\go, \gd, \eta,h)$
\begin{equation}
\label{eq:smalllap}
 \bP_{N,\go,\gb,h}\left( \cB(N,\gd)^{\cc} \right)\,\le\,  e^{- h^{2+ \eta} N^d }\, .
\end{equation}
But since $11/6<2$ we can avoid   Lemma~\ref{th:ctrlP} 
and obtain in a more direct way from Lemma~\ref{smalllap} that the right-hand side in \eqref{eq:smalllap}
can be improved to $ e^{- h^{2-\eta} N^d }$, any $\eta < 1/6$, but this is irrelevant for what follows.

\medskip

\begin{rem}
The reason for us to consider bi-gradient of of $\overbar \psi$ instead of $\overbar \psi$ or its gradient is that the spatial correlations of the bi-gradient are summable (which is not the case for neither $\overbar \psi$ nor for its gradient). 
The corresponding statement for $\psi$ or for the gradient of $\psi$ would not hold for this reason.
\end{rem}

\medskip

Given $k\in \lint 1,K\rint$, and recalling that  $\overbar{B}_{6^dk+i}$, $i=1,\dots,6^d$ are disjoint boxes of edge length $\overbar{N}_h$ located in $B_{0,k}$
we define   $\cE^{(1)}(k)$ to be the event 
that all elementary boxes in $B_{0,k}$ have a typical contact density.
Recalling \eqref{lazeta} we set
\begin{equation}
\label{eq:cE1}
  \cE^{(1)}(k):=  \bigcap_{i=1}^{6^d}\left\{ \frac{\zeta(6^dk+i)}{\chi(\beta) h}\in [1, 3]   \ \right\},
  \end{equation}
  We let  $\cE^{(2)}(k)= \cE^{(2)}(k,\gd,\gep)$ be the event that the box $B_{0,k}$ displays a density $\gd$ of high points
  (recall $u_h=\sigma_d\sqrt{2\log (1/h)}$)
  \begin{equation}
  \label{eq:cE2}
   \cE^{(2)}(k):=  \left\{
 \sum_{x\in B_{0,k}} \ind_{\{ |\phi(x)| \ge (1+\gep)u_h\}}\ge \gd  N^d_0\right\}
\end{equation}
Our second result is that  $\cE^{(1)}(k) \cap     \cE^{(2)}(k) $ is unlikely in boxes where the bi-gradient of $\overbar \psi$ are small: to the the proof is devoted Section~\ref{sec:level-0}.
 
\medskip

\begin{proposition}
\label{unecas}
Given $\eta>0$ there exists $ h_0(\eta,\gd,\gep)>0$ such that if $h\le h_0(\eta,\gd,\gep)$ and  $\max_{e,g}\|\nabla_{eg} \overbar \psi  \|_{\infty,k} \le N^{-2}_0$
we have 
\begin{equation}
\label{eq:unecas}
 \bP_N\left(  \cE^{(1)}(k) \cap     \cE^{(2)}(k) \ \Big \vert \  \overbar \psi \right)\, \le\,   \eta\, .
\end{equation} 
\end{proposition}

\medskip

 By Proposition \ref{smalllap}, in order to get to \eqref{eq:propoz} it is sufficient to show that 
 \begin{equation}
  \bP_N\left( \cC(N,\gd) \cap \cD(N,\gep) \cap \cB(N,\gd) \right)\, \le\,  \exp\left(- h^{2+\eta} N^d\right)\, , 
 \end{equation}
for some small $\gd>0$ (which is allowed to depend on $\gep$).
We are going to show that whenever  $\overbar \psi\in \cB(N,\gd)$ we have 
\begin{equation}\label{aprouver}
 \bP_N\left( \cC(N,\gd) \cap \cD(N,\gep) \ \big\vert  \  \overbar \psi \right) \, \le\,   \exp\left(- h^{2+\eta} N^d\right)\, .
\end{equation}
Now we record as a lemma the observation that when  $\cC(N,\gd) \cap \cD(N,\gep)$ holds then  $\cE^{(1)}(k) \cap     \cE^{(2)}(k)$ holds in a lot of boxes: also the proof of this result is delayed till the end of the main argument.

\medskip

\begin{lemma}
\label{top}
Given $\gep>0$ sufficiently small and $\gd\le \gep/(6^d+4)$ then we have
for all $N$ sufficiently large 
 \begin{equation}\cC(N,\gd) \cap \cD(N,\gep) \cap \cB(N,\gd)    \subset    \left\{
 \sum_{k=1}^K \ind_{\cE^{(1)}(k) \cap     \cE^{(2)}(k)} \ind_{\{ \max_{e,g}\|\nabla_{eg} \overbar \psi  \|_{\infty,k} < N^{-2}_0\}} \ge  \gd K\right\}\, .
  \end{equation}
\end{lemma}
According to Proposition \ref{unecas},  conditioning to  $\overbar \psi$ and assuming $\overbar \psi \in \cB(N,\gd)$ the sum

\begin{equation}\sum_{\{ k \ : \ \max_{e,g}\|\nabla_{eg} \overbar \psi  \|_{\infty,k} \le N^{-2}_0 \}} \ind_{ \cE^{(1)}(k) \cap     \cE^{(2)}(k)}\end{equation} is stochastically dominated by a binomial of parameters $\eta$ and $K$.
  We choose $\eta:= \gd/3$ and from the second bound in Lemma~\ref{th:binomial}  applied for $\Delta= 3$,
 we obtain that when $\overbar \psi \in \cB(N, \gd)$ 
 \begin{equation}
 \begin{split}
  \bP_N\left( \sum_{\{ k \ : \ \max_{e , g}\|\nabla_{eg} \overbar \psi  \|_{\infty,k} \le N^{-2}_0 \}} \ind_{ \cE^{(1)}(k) \cap     \cE^{(2)}(k)} \ge \gd K \ \bigg | \ \overbar \psi \right)\, &\le\, P\left( \mathrm{Bin}(K,\gd/3) \ge \gd K\right) 
  \\ 
  &\le\,   e^{-(\gd/2) K}\, .
  \end{split}
 \end{equation}
But  $K=2^{dJ}=( N/ N_0 )^d\ge c N^d h^2 (\log h)^{-2}$ for some $c>0$,  this is sufficient to conclude that \eqref{aprouver} holds. This completes the proof of \eqref{eq:propoz} and therefore  Proposition~\ref{lescontactzB} is established. 
\qed

\medskip 

\begin{proof}[Proof of Lemma \ref{top}]
We need to show that on the event $\cC(N,\gd) \cap \cD(N,\gep) \cap \cB(N,\gd)$ we have
\begin{equation}
\label{cprouve}
 \sum_{k=1}^K \ind_{\cE^{(1)}(k)^{\cc}}+ \sum_{k=1}^K   \ind_{\cE^{(2)}(k)^{\cc}}
 +  \sum_{k=1}^K \ind_{\{ \max_{e , g}\|\nabla_{eg} \overbar \psi  \|_{\infty,k} \ge N^{-2}_0\}} \le (1-\gd) K.
\end{equation}
First of all on  $\cC(N,\gd)$, recall \eqref{eq:cCset}, we have 
 \begin{equation}
 \sum_{k=1}^K \ind_{\cE^{(1)}(k)^{\cc}} \,\le\, \gd \overbar{K} \,=\,  6^d \gd K\, .
 \end{equation} 
 Moreover on $\cD(N,\gep)$
 \begin{equation}
 \gep N^d \le \sum_{x\in \mathring{\gL}_N} \ind_{\{ |\phi(x)| \ge (1+\gep)u_h\}}
 \,\le\,  N^d _0\sum_{k=1}^K \ind_{\cE^{(2)}(k)} 
 +\gd N_0^d K+ (N^d- KN^d_0)\, ,
 \end{equation}
 where the first term follows estimating from above with $N_0^d$ the number of points that are two high in a 
 box $B_{0,k}$ where there are at least $\gd N_0^d$ points that are two high, the second accounts  
 for the fact that in all other boxes (overestimated by \emph{all boxes} tout court) there at most $\gd N_0^d$
 high points, and the third accounts for the points that are not in the level $0$ boxes.
 Recalling that $(N/N_0)^d\ge K$, if $h$ is sufficiently small (so that the term
 $(N^d- KN^d_0)/N^d$ is small), we have
 \begin{equation}
  \sum_{k=1}^K \ind_{\cE^{(2)}(k)^{\cc}} \le (1+2 \gd-\gep) K\, .
 \end{equation}
Finally the third sum is smaller than $\gd K$ by the definition of 
 $\cB(N,\gd)$. Collecting all the estimates 
 we see that on $\cC(N,\gd) \cap \cD(N,\gep) \cap \cB(N,\gd)$ the 
the left-hand side of \eqref{cprouve} is bounded above by $(1+ 6^d +3\gd -\gep) K$. 
Since $\gd\le \gep/(6^d+4)$, 
  \eqref{cprouve} holds and the proof of Lemma \ref{top} is complete.
\end{proof}

\subsection{Level-$0$ estimates: the proof of Proposition \ref{unecas}}
\label{sec:level-0}
A first observation is the following 

\medskip

\begin{lemma}
\label{th:ctrldev}
 If $\max_{e, g}\|\nabla_{eg} \overbar \psi  \|_{\infty,k} \le N^{-2}_0$, then there exists $a\in \bbR$ and $b\in \bbR^d$ such that for all $x\in B_{0,k}$
 we have,
\begin{equation}
 \forall x\in B_k \ \ \  | \overbar \psi(x)-a+b\cdot x|\le d^2. 
\end{equation}
\end{lemma}

\medskip

\begin{proof}
 We can replace without loss of generality $B_{0,k}$ by $\lint -N_0/2, N_0/2\rint^d$ with appropriate rounding.
 We set 
 \begin{equation}
 \label{eq:a-and-b}
 a\, :=\, \overbar \psi(0)\ \text{ and } \ 
 b_i\,:=\,\overbar \psi(e)-\overbar \psi(0)\, .
 \end{equation}
 We have for every $x\in \lint -N_0/2, N_0/2\rint^d$ and $i\in \lint 1, d\rint$ 
\begin{equation}
 \left\vert\overbar
 \psi(x+e)-
 \overbar
 \psi(x)-b_i
 \right \vert \, \le \, d N^{-1}_0\, ,
 \end{equation}
simply because the term in the absolute value can be written as the sum of 
at most $N_0 d$ terms of the form $\nabla_{eg} \overbar \psi(z)$.
Using this we see that 
\begin{equation}
\left \vert\overbar \psi(y)-  \overbar \psi(0)-b\cdot y\right \vert\, \le\,  d^2\, ,
\end{equation}
 since the term in the absolute value can be written as the sum of at most $N_0d $ terms of the form $\overbar \psi(x+e)-  \overbar \psi(x)-b_i$.
 \end{proof}
 
\medskip 

We recall that if we work on  $\overbar{\psi} \in \cB(N, \gd)$, most level-$0$ boxes are good in the sense that they satisfy the bi-gradient requirement of Lemma~\ref{th:ctrldev}. 
 We assume that the (good) level-$0$ box $B_k$ we are considering is centered at $0$, that $\overbar \psi(0)\ge 0$ and that $b_i\ge 0$ for all $i$: by symmetry this yields no loss of generality, as we explain in the caption of Figure~\ref{fig:6d}. 
We set ($\cdot$ is the scalar product)
 \begin{equation}
 H_0(x)\, :=\,  a+b\cdot x\, ,
 \end{equation}
 so  Lemma~\ref{th:ctrldev} implies that  
 $\max_{x \in B_k}\vert \overbar{\psi}(x) -  H_0(x)\vert \le d^2$.
 We set $y_1:=(\overbar{N}_h, \overbar{N}_h,\dots, \overbar{N}_h)$, $y_2:=2y_1$, $y_3:=3y_1$ and we  let $\overbar{B}_i$ denote the cube of edge length $\overbar{N}_h$ whose maximal corner (for the lexicographic  order) is $y_i$, see Figure~\ref{fig:6d} and its caption.

\begin{figure}[htbp]
\centering
\includegraphics[width=14.5 cm]{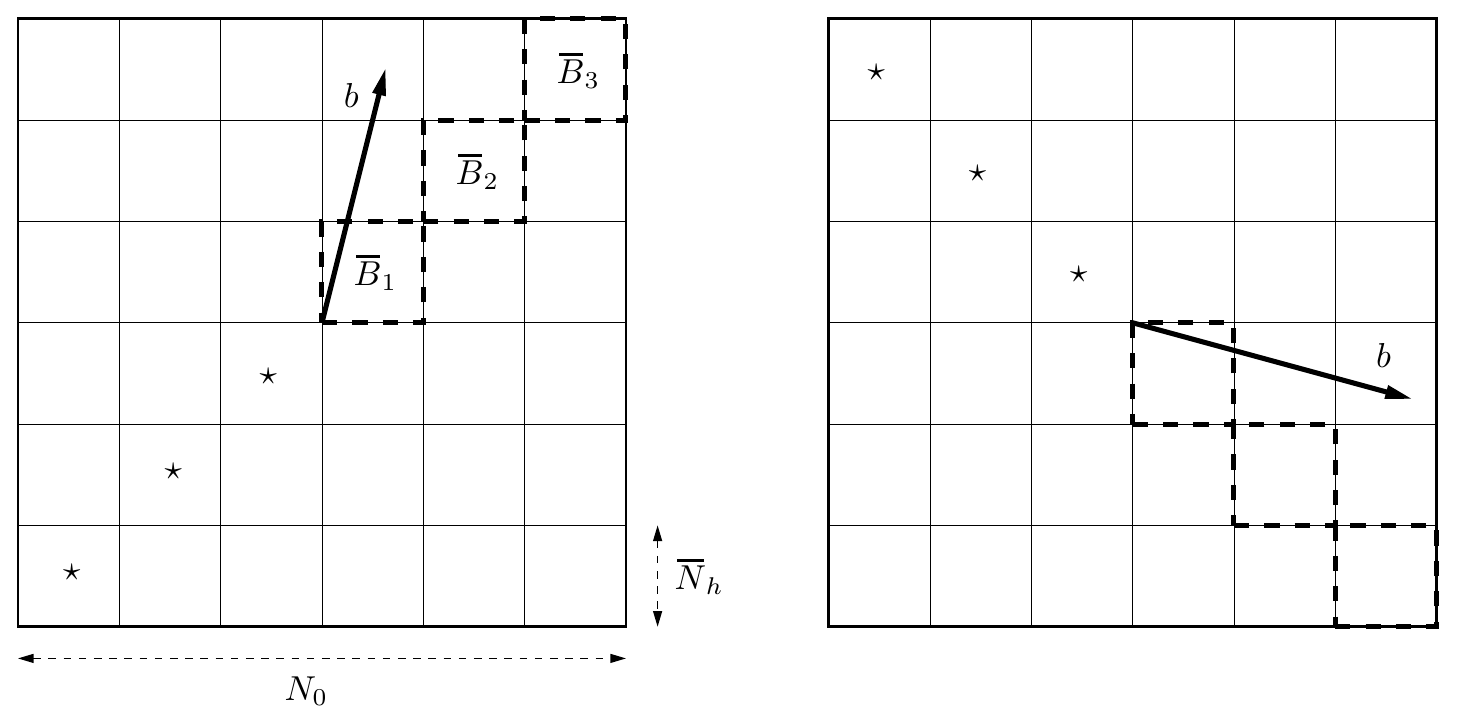}
\vskip-.2cm
\caption{\label{fig:6d} 
A level-$0$ box $B_{0,k}$, of edge length $N_0$, is partitioned into $6^d$ boxes (hypercubes) $\overbar{B}_1, \ldots, \overbar{B}_{6^d}$ of edge length $\overbar{N}_0$, up to integer rounding (so it may be that only a subset of 
$B_k$ is really covered (but what is left out is just a negligible fraction, $O(1/N_0)$, of the sites in $B_k$).
Note that we have relabeled the elementary boxes contained in $B_{0,k}$ both because they  were originally labeled with indexes going from 
$6^d k+1$ to $6^d (k+1)$ and because we have made a specific choice of the first three elementary boxes that minimizes notations.  
With reference to \eqref{eq:a-and-b}, the drawing considers the case $a\ge 0$ and, in the case on the left, $b=(b_1, \ldots,b_d)$ has non negative entries. In this case the three steps analysis -- see points $(i)$-$(iii)$ -- just focuses on 
three elementary boxes $\overbar{B}_1$ to $\overbar{B}_3$. Whenever one or more entries of $b$ are non positive it suffices to change of quadrant, like in the case on the right. And if $a \le 0$ it suffices to change the orientation and choose the elementary boxes marked by $\star$. So the $2^{d+1}$ cases that we need to analyse are all equivalent. 
}
\end{figure}

To prove Propositionı\ref{unecas}, we are going to distinguish three cases.

\medskip

\begin{itemize}
 \item [$(i)$\ ] $H_0(y_1)\le (1-\gd) u_h$. In this case   $H_0$ is small in  the whole $\overbar{B}_1$ and this yields, with large probability, too many contacts with respect to what  $\cE^{(1)}$ allows.
 \item [$(ii)$\,] $H_0(y_2)\ge (1+\gd) u_h$. In this case $H _0$ is large in $\overbar{B}_3$, which implies that, with large probability, there are too few contacts in this box with respect to what is required by  $\cE^{(1)}$. 
 \item[$(iii)$] $H_0(y_1)\ge (1-\gd) u_h$ and $H_0(y_2)\le (1+\gd) u_h$. In this case $H_0(y_3)\le (1+3\gd) u_h$, meaning in particular that $|H_0(x)|\le (1+3\gd) u_h$  on the full level $0$ box (because of our assumption $a$ and $b_i$ positive).
This makes $\cE^{(2)}$ unlikely if one chooses $\gd=\gep/6$ and $h$ sufficiently small.
\end{itemize}
\medskip

We will treat
cases $(ii)$ and $(iii)$  using first moment estimates and Markov inequality.
Case $(i)$  requires a less straightforward second moment computation.

\medskip

\subsubsection{The case (i)}
On $\overbar{B}_1$ we have  $-d^2\le \overbar \psi \le (1-\gd) u_h+d^2$. 
Hence we have for $h$ sufficiently small 
\begin{equation}
\label{eq:alm54}
 \bE_N\left[ \sum_{x \in \overbar{B}_1} \ind_{[-1,1]}(\phi(x)) \ \bigg| \ \overbar \psi \right]\, \ge \,
  \overbar{N}^d_h P \left( \gs_d \cN \in (1-\gd)u_h +d^2+ [0,2]\right)\, \ge \, 
  \overbar{N}^d_h h^{1-\gd},
\end{equation} 
where we have taken the worst case scenario for $\overbar \psi$ and we have added a shift of one to compensate 
for the boundary effects on the variance, cf.
Remark~\ref{rem:boundary-0}: the effect of the boundary on the variance is $O(h^c)$ so a shift of a unit is largely sufficient because $u_h=O(\sqrt{\log(1/h)})$.  
The probability in the intermediate term  turns out to be bounded below by $h^b$ for every $b> (1-\gd)^2$ by a direct Gaussian estimate. So the choice $b=1-\gd$ yields  the lower 
bound of $h^{1-\gd}$ and \eqref{eq:alm54} holds. We now proceed to a second moment estimate to get a concentration result on the number of contacts. 
The most technical estimate of this step is in the following lemma:

\medskip

\begin{lemma}
\label{th:alm54}
When $|x-y|\ge \vert\log h\vert ^{2}$ and $|v_1|, |v_2|\le u_h-1$ we have
\begin{multline}\label{linneq}
  \bE_N\left[ \ind_{[-1,1]}(\psi(x)+v_1)\ind_{[-1,1]}(\psi(y)+v_2) \right]\\
  \le   (1+C|\log h|^{-1})\bE_N\left[ \ind_{[-1,1]}(\psi(x)+v_1)  \right] \bE_N\left[\ind_{[-1,1]}(\psi(y)+v_2)  \right]\, ,
\end{multline}
with $C$ a constant that depends only on $d$. 
\end{lemma}

\medskip 

\begin{proof} 
 We need to bound 
  $\bE_N[ \ind_{[-1,1]}(\psi(y)+v_2) \ | \ \psi(x)]$ when $\psi(x)\in [-1,1]-v_1$.
 Let $\sigma^2_{y}$ denote the variance of $\psi(y)$ and $(\sigma^{(x)}_y)^2$ is the  variance of $\psi(y)$  conditional to $\psi(x)$:
 let us remark from now that these variances are bounded below by the variance of the field on one site conditioned to its nearest neighbors, that is  $1/(2d)$.
 The covariance of $\psi(x)$ and $\psi(y)$, that is $\bE_N[\psi(x) \psi(y)]$, is  denoted, as usual, by $G_N(x,y)$.
 Recall that, by standard computations on bi-variate Gaussian vectors, $\bE_N[\psi(y) \vert \psi(x)]= c_N(x,y) \psi(x)$ with
\begin{equation}
\label{eq:somctrlvar1}
 c_N(x,y)\,:=\,  \frac{G_N(x,y)}{\gs_x^2}\, \le\, 2d G(0, y-x) \, \le  C \vert\log h\vert ^{2 (2-d)}\, ,
 \end{equation}
 with $C$ a constant that depends only on $d$: in this proof we re-use  $C$ precisely in this sense, possibly updating its value. 
Moreover for the conditional variance we have the formula
\begin{equation}
\label{eq:somctrlvar2}
(\sigma^{(x)}_y)^2\,=\,  \left(1- \left(\frac{G_N(x,y)}{\gs_x \gs_y}\right)^2
\right) \gs_y^2 \, =\, \gs_y^2 - \left(c_N(x,y)\gs_x\right)^2\, .
 \end{equation}
 Therefore \eqref{eq:somctrlvar1} and \eqref{eq:somctrlvar2} yield
 \begin{equation}
 (\Delta \sigma)^2\, :=\,  \sigma^2_{y}- (\sigma^{(x)}_y)^2\, \, =  \, \left(c_N(x,y)\gs_x\right)^2 \, \le\,  G(0,0)C^2\,  \vert \log h\vert^{4 (2-d)}\,.
 \end{equation}
Now we remark that  we have 
 \begin{equation}\begin{split}
   \bE_N[ \ind_{[-1,1]}(\psi(y)+v_2) ]&
  = \frac{1}{\sqrt{2\pi}\sigma_{y}}\int_{[-1,1]} e^{ -\frac{(u-v_2 )^2}{2\sigma^2_{y}}}\dd u,\\
  \bE_N[ \ind_{[-1,1]}(\psi(y)+v_2) \ | \ \psi(x)=z]
   &= \frac{1}{\sqrt{2\pi}\sigma^{(x)}_{y}}\int_{[-1,1]} e^{ -\frac{(u-v_2- c_N(x,y)z)^2}{2 (\sigma^{(x)}_y)^2}}\dd u\, ,
 \end{split}\end{equation}
 with some small abuse of notation in the second identity. 
 We aim at controlling the ratio between these two quantities. For this it is sufficient to observe that for every $|z|, |v_2|\le u_h$ and $|u|\le 1$
\begin{equation}
\begin{split}
\frac{\sigma_{y}}{\sigma^{(x)}_{y}} e^{ \frac{(u-v_2 )^2}{2\sigma^2_{y}}-\frac{(u-v_2- c_N(x,y) z)^2}{2 (\sigma^{(x)}_y)^2}
  } \, &= \, \sqrt{1+ \left(\frac{\Delta \sigma}{\sigma^{(x)}_{y}}\right)^2} e^{ -\frac{(u-v_2 )^2 (\Delta \sigma)^2}{2\sigma^2_{y}(\sigma^{(x)}_{y})^2}+ \frac{(u-v_2 )^2- (u-v_2- c_N(x,y) z)^2 }
  {2(
\sigma^{(x)}_{y})^2}}\\
&\le \, 
\sqrt{1+ \left(\frac{\Delta \sigma}{\sigma^{(x)}_{y}}\right)^2} e^{  \frac{ \vert u-v_2\vert c_N(x,y) \vert z\vert }
  {(
\sigma^{(x)}_{y})^2}}\, 
\le\,  1+  \frac{C}{\vert \log h \vert}\, ,
%( u^2 _h (\log h)^{2(2-d)} ).
\end{split}
\end{equation}
where in the last inequality we have used \eqref{eq:somctrlvar1} and \eqref{eq:somctrlvar2}, obtaining thus an upper bound of
$1+C \max(\vert \log h\vert ^{4(2-d)}, u_h^2 \vert \log h\vert ^{2(2-d)})$ and the worst case estimate is for $d=3$. We have therefore completed the proof of Lemma~\ref{th:alm54}.
 \end{proof}
 \medskip

Using \eqref{linneq} when $|x-y|\ge \vert\log h\vert^2$ and $\ind_{[-1,1]}(\phi(x))\ind_{[-1,1]}(\phi(y))\le \ind_{[-1,1]}(\phi(x))$ when $|x-y|< (\log h)^2$
we obtain that 
\begin{multline}
 \bE_N\left[ \sum_{x,y\in \overbar{B}_1} \ind_{[-1,1]}(\phi(x))\ind_{[-1,1]}(\phi(y)) \ \bigg| \ \overbar \psi \right]  \\ \le (1 + C|\log h|^{-1}) \left(\bE_N\left[ \sum_{x\in\overbar{B}_1} \ind_{[-1,1]}(\phi(x)) \ \bigg| \ \overbar \psi \right] \right)^2 \\ +  (\log h)^{2d}
  \bE_N\left[ \sum_{x\in \overbar{B}_1} \ind_{[-1,1]}(\phi(x)) \ \bigg| \ \overbar \psi \right] \, .
\end{multline}
This inequality can be written in compact form if we call $V_N$ the conditional variance of $ \sum_{x\in \overbar{B}_1} \ind_{[-1,1]}(\phi(x))$ and $m_N$ the conditional mean, which, by \eqref{eq:alm54}, is bounded below by $\overbar{N}_h^d h^{1-\gd}\ge h^{-1-\gd}$ (by \eqref{eq:barboxsize}):
\begin{equation}
\frac{V_N}{m_N^2}\,  \le\,  \frac{C}{\vert \log h \vert} + \frac{\vert \log h\vert^{2d} }{m_N} \, \le \, \frac{2C}{\vert \log h \vert}\,.
\end{equation}
Therefore, using again \eqref{eq:alm54}, we see that for every $\eta>0$ 
\begin{equation}
 \bP_N\left( \sum_{x\in\overbar{B}_1} \ind_{[-1,1]}(\phi(x))\le  \overbar{N}^d_h h^{1-\gd/2} \, \bigg\vert\, \overbar \psi\right)  \, \le\,  \, \eta\,,
\end{equation}
by choosing $h$ small, and, by recalling  \eqref{eq:cE1}, we see that in the case $(i)$, inequality \eqref{eq:unecas} holds. 
\medskip

\subsubsection{The case {\rm (ii)}}
In this case it suffices to observe that 
\begin{multline}
\label{eq:alm55}
 \bE_N\left[ \sum_{x \in \overbar{B}_3} \ind_{[-1,1]}(\phi(x)) \ \bigg| \ \overbar \psi \right]\, \le\, 
  \overbar{N}^d_h P \left( \gs_d \cN \ge  (1+\gd)u_h -d^2\right)
  \\
  \le \, 
  \overbar{N}^d_h \exp\left( - \frac {\left((1+\gd)u_h -d^2\right)^2}{2 \gs_d^2}\right)\, 
  \le  \overbar{N}^d_h
  h^{1+\gd}\, ,
\end{multline} 
where, in the last step, we have used \eqref{def2uh}
and from this, applying Markov's inequality we have (recall   \eqref{eq:cE1})
\begin{equation}
\bP\left(\cE^{(1)} \ | \ \bar \psi\right) \le \bP_N\left( \sum_{x \in \overbar{B}_3} \ind_{[-1,1]}(\phi(x)) \ge h\chi(\beta)\overbar{N}^d_h
 \ \bigg| \ \overbar \psi \right)\, \le\, (\chi(\beta))^{-1}h^{-\gd}\, .
\end{equation}
% , is again incompatible  with $\cE^{(1)}$, i.e. with the fact that the contact fraction has to be bounded below by $h$ times a positive constant. 
So, also  in the case (ii) we have that \eqref{eq:unecas} holds.
\medskip

\subsubsection{The case  {\rm (iii)}} Here again just a first moment estimate suffices. As we pointed out 
$H_0(x)\le (1+3\gd) u_h$ for every $x$ in the level-$0$ box $B_k$ we are considering. So $\overbar \psi (x)\le  (1+3\gd) u_h +d^2$ for $x\in B_k$ and if we choose $\gd\le \gep/6$ we have that for $h$ sufficiently small
\begin{equation}
\bE_N\left[ \sum_{x \in B_k} \ind_{(-(1+ \gep)u_h, (1+\gep)u_h)^\complement}(\phi(x)) \ \bigg| \ \overbar \psi \right]\, \le\, 2 N_0^d 
P \left( \gs_d \cN \ge \gep u_h/2\right) \, \le \,  2 N_0^d h^{\gep^2/5}\, ,
\end{equation} 
and the Markov inequality immediately yields that the conditional probability of the event $\cE^{(2)}$, defined in  \eqref{eq:cE2}, can be made arbitrarily small, in particular smaller than $\eta$, by choosing $h$ smaller than a suitable constant (that depends on $\gep$). 
Therefore also  in the case $(iii)$ we obtain \eqref{eq:unecas}.

\medskip

This completes the proof of Proposition~\ref{unecas}. 
\qed

\subsection{The multiscale bound: proof of  Proposition~\ref{smalllap}}
\label{sec:multiscale}

Recall the two scale decomposition \eqref{eq:2psi} and that we need to control $\overbar \psi$. 
We start by a notational remark: we have  
$\overbar{\psi}= \bE\left[ \phi_{\gL_N}\vert \bbG _0\right]$ and $\psi= \phi_{\gL_N}- \bE\left[ \phi_{\gL_N}\vert \bbG _0\right]$, where
we set for conciseness 
$\bE[ \,\cdot\, \vert A]:=\bE[ \,\cdot\, \vert \cF_A]$.
We recall that we use the notation of $\phi_A$ for the free field with $0$ boundary conditions outside $A\subset \bbZ^d$.
In fact, in this section we avoid using $\bP_N$: the law of $\phi$ under $\bP_N$  just coincides with the law of $\phi_{\mathring{\gL}_N}$ (and the probability is just denoted by $\bP$).
By noticing that
$\bE[ \phi_{\gL_N}\vert  \bbG _J]=0$ simply because by definition $\bbG _J= \tilde \bbG_0= \gL_N^\complement$
we see that
\begin{equation}
\label{eq:multiscale}
\overbar \psi\, =\, \sum_{j=1}^J \psi_j \, , 
\end{equation}
where
\begin{equation}
\label{eq:multiscale-1}
  \psi_j\,:=\, \bE\left[\phi_{\gL_N}\, \big \vert\,   \bbG _{j-1}\right]-
 \bE\left[\phi_{\gL_N}\, \big \vert\,   \bbG _{j}\right]\, .
\end{equation}
By constructions, the field $\psi_j(x)$ and $\psi_j(y)$ are independent if $x$ and $y$ belong to different level-$j$ boxes.
And of course $(\psi_j)_{1, \dots, J}$ is a family of independent fields. 
This independence is going to play a crucial role.

Let us recall  now the definitions
\eqref{eq:bigraddef}-\eqref{eq:cBN}.
From \eqref{eq:multiscale} of course we have
\begin{equation}
\label{eq:fmsc-1}
\left\vert \nabla_{e g} \overbar \psi(x)\right\vert \, \le \, 
\sum_{j=1 }^J \left\vert \nabla_{eg} \psi_j(x)\right\vert\, .
\end{equation}

\medskip

\noindent
\emph{Proof of proof of  Proposition~\ref{smalllap}.}
With reference to \eqref{eq:fmsc-1}: we need to estimate the left-hand side and we are going to  do so 
by estimating every level-$k$ separately. We start by estimating the variance of 
$\nabla_{eg} \psi_j(x)$ and for this we use
\begin{equation}
\label{eq:varbound5}
\Var\left( \nabla_{eg} \psi_j(x)\right)\,=\, \bE \left[
 \Var_{ \bbG _{j}}\left(  \bE\left[\nabla_{eg} \phi_{\gL_N}(x)\, \big \vert\,   \bbG _{j-1}\right] \right)\right]
  \, \le \, \Var\left(  \bE\left[\nabla_{eg} \phi_{\gL_N}(x)\, \big \vert\,   \bbG _{j-1}\right] \right)\,,
\end{equation}
where $ \Var_{ \bbG _{j}} ( \cdot) $ is the variance with respect to $\bP(\cdot \vert  \bbG _{j})$.
But we can go even farther: in fact 
\begin{equation}
\label{eq:varbound6}
\Var\left( \nabla_{eg} \psi_j(x)\right)\, \le \, \Var\left(  \bE\left[\nabla_{eg} \phi(x)\, \big \vert\,   \bbG _{j-1}\right] \right)\,,
\end{equation}
where $\phi$ is now the infinite volume field (note the difference between the right(most)-hand sides in \eqref{eq:varbound5}
and \eqref{eq:varbound6}). This is simply because, using the same trick we have repeatedly used up to now,  
$\phi $ can be written as sum of two independent fields: the first one is the trace of the infinite volume LFF over the complement of $\gL_N$ harmonically continued in $\gL_N$ and the second one is a free field with Dirichlet boundary conditions outside $\gL_N$. Hence the variance in right-hand side of \eqref{eq:varbound6} can be written as the sum of two variances and one of them coincides with the right-hand side of \eqref{eq:varbound5} and \eqref{eq:varbound6}
holds.

At this point we can apply Lemma~\ref{th:d+2}, together with Lemma~\ref{th:encadre},  to the right-hand side of  \eqref{eq:varbound6}
and we see that for every $j\in \lint 1, 2^J\rint$
\begin{equation}
\label{eq:varbound7}
\max_{e, g} \max_{k\in \lint 1, 2^{d(J-j+1)}\rint}\max_{x \in B_{j-1,k}}\Var\left( \nabla_{eg} \psi_j(x)\right)\, \le \, \frac{C}{N_{j}^{d+3/2}}\, ,
%\,\le \, \frac{C'}{N_{0}^{d+3/2}2^{j(d+3/2)}}\, ,
\end{equation}
where the exponent $3/2$ has been chosen arbitrarily in $(1,2)$    (we could even choose $2$ if we introduce a 
logarithmic  
correction, see  Lemma~\ref{th:d+2}), but any number larger than $1$ suffices for our purposes. $C$ is just a $d$ dependent constant
and it has been chosen also to compensate the fact that for readability we replaced $N_{j-1}$ with $N_j$. 
From \eqref{eq:varbound7} we directly obtain  the Gaussian tail estimate
\begin{equation}
\bP\left( \left \vert  \nabla_{eg} \psi_j(x)\right\vert \, \ge \, \frac{h^{4/d}}{j^2 {(\log(1/h))^2} }\right) \,
\le \, 2 \exp\left(- \frac{h^{8/d}N_j^{d+3/2}}{2C j^4  (\log(1/h))^4}
\right)\, ,
\end{equation}
uniformly in $e$, $g$ and in $x$ in all the level-$(j-1)$ boxes. Therefore, by a union bound, we have that
if $\check B_j$ is the union of the $2^d$ level-$(j-1)$ boxes contained in the level-$j$ box $B_j$
\begin{multline}
\bP\left( \max_{e, g}
\max_{x \in \check B_j} 
\left \vert  \nabla_{eg} \psi_j(x)\right\vert \, \ge \, 
 \frac{h^{4/d}}{j^2 (\log(1/h))^2 }\right) 
 \,
\le 
\\
 2d^2 N_j^d \exp\left(- \frac{h^{8/d}N_j^{d+3/2}}{2C j^4 ( \log(1/h))^4}
\right)\,=: \, p_j(h)\, .
\end{multline}
Now we remark that  \eqref{eq:0boxsize} implies
\begin{equation}
h^{8/d}N_0^{d+3/2}\, \ge \, h^{(8/d)-(2/d)(d+3/2)}\, \ge\, h^{-1/3}\, , 
\end{equation}
and for $h$ sufficiently small we have for every $j\in \lint 1,  J\rint$
\begin{equation}
  \log p_j(h)\, \le \, -\frac{h^{-1/3} 2^{j(d+3/2)}}{2Cj^4 (\log(1/h))^4} + d \log N_j +\log (2 d^2)
\, \le \,  -\frac{h^{-1/4} 2^{jd} \exp(j)}{j^4 }\, ,
\end{equation}
because $d \log N_j \le j d \log 2+ \log (1/h)$ for $h$ small and in the last step we have used $2^{3/2} > e$.
We record explicitly the bound that we will use
\begin{equation}
\label{eq:pjbound}
p_j(h)\, \le \, \exp\left(- \frac{h^{-1/4} 2^{jd}\exp(j)}{j^4 } \right)\, ,
\end{equation}
and that we choose $h$ sufficiently small to guarantee that $\sup_j p_j(h)<1$. But we are going to choose $h$ small to
satisfy also the stronger requirement that for given  $\eta>0$ we have $p_j(h)<(1/2)^{2j^2/\eta}$ for every $j$: 
we are going to choose $\eta:= 6\gd/ \pi^2$, with $\gd$ the constant  entering the definition \eqref{eq:cBN} of $\cB(N, \gd)$, but this is going to be irrelevant till the very last steps of the proof. 
 We can then 
apply the binomial bound in Lemma~\ref{th:binomial2} and obtain
\begin{multline}
\bP\left(
\#
\left \{
k \in \lint1, 2^{d(J-j+1)}\rint:\, \max_{e, g} \max_{x\in \check B_{j,k}}
\left \vert  \nabla_{eg} \psi_j(x)\right\vert 
\, \ge \, \frac{h^{4/d}}{j^2 (\log(1/h))^2 }
\right\}
 \, \ge \, \frac{\eta}{j^2}  2^{d(J-j+1)}
\right)
\\
\le\, \left(p_j(h)\right)^{\eta 2^{d(J-j+1)}/j^2}\, \le \, 
 \exp\left(- \eta \frac{h^{-1/4} 2^{d(J+1)}\exp(j)}{j^6 } \right)\, .
\end{multline}
Using the independence of $(\psi_j)_{j=1,\ldots,J}$ we can control  all levels at the same time:
\begin{equation*}
\bP\left(\exists j \in \lint 1, J\rint \textrm{ s.t. }
\#
\left \{
k :\, \max_{e, g} \max_{x\in \check B_{j,k}}
\left \vert  \nabla_{eg} \psi_j(x)\right\vert 
\, \ge \, \frac{h^{4/d}}{j^2 (\log(1/h))^2}
\right\}
 \, \ge \, \frac{\eta}{j^2}  2^{d(J-j+1)}
\right)
\end{equation*}
\begin{equation}
\begin{split}
\label{eq:estpr41}
%\phantom{movemovemove}
 &\le \, 1- \prod_{j=1}^J \left(
1-  \exp\left(- \eta \frac{h^{-1/4} 2^{d(J+1)}\exp(j)}{j^6 } \right)
\right)
\\
& \le \, 1- \prod_{j=1}^J \left(
1-  \exp\left(- 2\eta h^{-1/5} 2^{d(J+1)}j 
\right) \right)\\
&\le \, 1- \exp\left(\frac 12\sum_{j=1}^\infty 
 \exp\left(- 2\eta h^{-1/5} 2^{d(J+1)}j 
\right) \right)
\\
& \le \,
 \exp\left(- \eta h^{-1/5} 2^{d(J+1)}
\right) \, \le \,
 \exp\left(-  h^{2-\frac 16} N^d
\right) \, ,
\end{split}
\end{equation}
Now we recall \eqref{eq:fmsc-1} and the definition \eqref{eq:cBN} of the event $\cB(N, \gd)$
and we see that on the complementary of the event whose probability is estimated in \eqref{eq:estpr41}
we have that 
\begin{equation}
\max_{e, g}\Vert \nabla _{eg} \overbar \psi \Vert_{\infty, k}\,\le\, \frac{h^{4/d}}{(\log(1/h))^2} \sum_{j=1}^\infty \frac 1{j^2} \, \le \,
\frac 1{N_0^2}\, ,
\end{equation}
except for at most a fraction $\gd=\eta \sum_{j\in \bbN} j^{-2}$ of the $K=2^{dJ}$ level-$0$ boxes.
This completes the proof of Proposition~\ref{smalllap}.
\qed

 \appendix
 
 \section{Technical estimates}

\subsection{Standard binomial bounds}
 $\textrm{Bin}(n,p)$ denotes a binomial random variable of parameters $n$ and
$p$.
\medskip

\begin{lemma}
\label{th:binomial}
For $\gD \in [0,1/6]$ we have  
\begin{equation}
\label{eq:binleft}
P\left(\mathrm{Bin}(n,p)  \, \le\,  p \gD n \right)\, \le\, \exp(-np/2)\, ,
\end{equation}
and for $\gD \ge 3$
\begin{equation}
 \label{eq:binleft2}
P\left(\mathrm{Bin}(n,p)  \, \ge\,   p \gD n \right)\, \le\, \exp(-np/2)\, .
\end{equation}
\end{lemma}

\medskip

\noindent
\emph{Proof.}
For $\gD \in [0, 1/p]$  set
with \begin{equation} 
f(p, \gD)\, :=\, p \gD \log \gD +(1- p\gD) \log ((1- p \gD)/(1-p)).
\end{equation}
By the exponential Markov inequality  we have that for every  $\gD \in [0,1]$  
\begin{equation}
P\left(\mathrm{Bin}(n,p)  \, \le  p \gD n \right) \, \le\,   \exp \left( - n f(p, \gD)  \right)\, ,
\end{equation}  
and for every  $\gD \in [1,1/p]$
\begin{equation}
P\left(\mathrm{Bin}(n,p)  \, \ge  p \gD n \right)\,  \le\,   \exp \left( - n f(p, \gD)  \right)\, .
\end{equation}  
The convexity of
$p \mapsto f(p, \gD)$ yields $f(p, \gD) \ge p \partial_p f(0, \gD)= p (1- \gD(1+ \log(1/\gD)))$, so the result 
follows because $1- \gD(1+ \log(1/\gD))\ge 1/2$ for $\gD\le 1/6$ and $\gD\ge 3$.
\qed

\medskip

\begin{lemma}
\label{th:binomial2}
Given $\eta>0$ such that
$p\le (1/2)^{2/\eta}$, we have $P(\mathrm{Bin}(n,p) \ge \eta n)\le p^{\eta n/2}$.
\end{lemma}

\medskip

\begin{proof}
If follows by remarking that
\begin{equation}
P(\mathrm{Bin}(n,p) \ge \eta n)\, =\, \sum_{k\ge \eta n} \left(\begin{array}{c}n \\ k\end{array}\right) p^k(1-p)^k 
\,\le\, p^{\eta n} \sum_{k\ge \eta n} \left(\begin{array}{c}n \\ k\end{array}\right)\, \le \, p^{\eta n}2^n\, \, .
\end{equation} 
\end{proof}

\subsection{A disorder cut-off estimate}
\label{sec:cut-off}
Let us consider a finite subset $A$ of $\bbZ^d$ and an arbitrary measure $\bP$ on $\bbR^{\bbZ^d}$ and the corresponding 
product $\gs$-algebra, with $\cB$ a measurable event in this $\gs$-algebra. We set 
$Z_{A, \go}(\cB):= \bE[ \exp(\sum_{x \in A}(\gb\go_x-\gl(\gb)+h) \gd_x) \ind_\cB]$
with $\gd_x$ a random variable on the probability space we just introduced taking only values $0$ or $1$. For $K>0$ we introduce also
$\overbar \go_x := \go_x \ind_{\vert \go_x\vert \le K}$ and
\begin{equation}
L_K\, :=\, \bbE\left[ \vert \go_x \vert ; \, \vert \go_x \vert >K\right]\,.
\end{equation}

\medskip

\begin{lemma}
\label{th:theboundz}
We have 
\begin{equation}
\label{eq:theboundz}
     \big\vert\bbE \left[\log Z_{A,  \go}(\cB)\right]-
          \bbE \left[\log  Z_{A,\overbar\go}(\cB)\right]\big\vert\, \le\,  \beta |A| L_K.
 \end{equation}
\end{lemma}

\medskip

 \begin{proof}
 We can assume $\cP(\cB)>0$.
 We have 
 \begin{equation}\label{eq:alaf}
       \bbE \left[\log \Z_{A,  \go}\right]-
          \bbE \left[\log Z_{A,\overbar \go}\right] 
          =  \bbE \log \bE_{A, \overbar \go} \left[ \exp\left(
  \gb \sum_{x\in A} \go_x \ind_{\{ |\go_x|> K\}} \gd_x
  \right) \right] \,, 
 \end{equation}
  where we have introduced the measure $\bP_{A, \overbar \go}$ associated to the partition 
  function $Z_{A, \overbar \go}(\cB)$ neglecting $\cB$ in the notation for conciseness.
   As we trivially have (for every realization of $\phi$)
   \begin{equation}
   \sum_{x\in A} \go_x  \ind_{\{ \go_x<-  K\}} \,\le\,   \sum_{x\in A} \go_x \ind_{\{ |\go_x|> K\}} \gd_x\,  \le\,  \sum_{x\in A} \go_x  \ind_{\{ \go_x> K\}} \, ,
   \end{equation}
the right-hand side of \eqref{eq:alaf} is smaller, respectively larger, than 
\begin{equation}
\beta \bbE\left[ \sum_{x\in A} \go_x \ind_{\{ \go_x> K\}}\right]\, , \quad \text{ respectively }  \beta \bbE\left[ \sum_{x\in A} \go_x \ind_{\{\go_x<- K\}}\right]\, ,
\end{equation}
both of which in absolute value are smaller  than 
\begin{equation} 
\beta |A|  \bbE\left[ |\go_x|\ind_{\{ \vert \go_x\vert > K\}} \right]\, =\, \beta |A|  L_K\, .
\end{equation}
\end{proof}

\subsection{Harmonic extension estimates}
For the next result $B\neq \emptyset$ is an arbitrary finite connected subset of $\bbZ^d$ and $H=H_B$ is the harmonic extension of the trace
on $\bbZ^d \setminus B$  of 
a LFF $\phi$, that is $\Delta H(x)=0$ for every $x\in B$ and $H(x)=\phi(x)$ for $x  \in \bbZ^d \setminus B$. 
We recall for this section that $G(\cdot, \cdot)$ is the Green function of the walk in $\bbZ^d$, cf. Section~\ref{sec:building}.
\medskip

\begin{lemma}
\label{th:var-est0lem}
For $d\ge 3$ for every  $x\in B$ we have 
\begin{equation}
\label{eq:var-est0lem}
\Var \left(H(x)\right)\, \le\, C_d \left(\left(\dist\left(x, \bbZ^d \setminus B\right)+1\right)^{2-d}\right)\,,
\end{equation}
with $C_d>0$ the constant appearing in \eqref{eq:Gest-tail}. 
\end{lemma}

\medskip

\begin{proof}
We use  that for $x\in B$ 
we have
\begin{equation}
\label{eq:var-formula}
\Var \left(H(x)\right)\,=\, \Var \left(\phi(x) \right) -  \Var \left(\phi^B_x \right)\, ,
\end{equation}
with  $\phi^B$ a free field with zero boundary conditions in $\bbZ^d\setminus B$. A proof of \eqref{eq:var-formula} is for example in \cite[p.~1676]{cf:BDG}.
If we reinterpret this formula in random walk terms we have that the variance under analysis is the expected number of visits to $x$ (by the walk that starts from $x$) after hitting the external boundary of $B$. If we call $p_x(z)$, $z \in \bbZ^d\setminus B$,  the hitting probability of the walk issued from $x$ -- of course $p_x(z)>0$ only if dist$(z,B)=1$ -- then, by using \eqref{eq:Gest-tail}, we  get to 
%\begin{multline}
\begin{equation}
\label{eq:var-formula2}
\Var \left(H(x)\right)\,=\, \sum_{z} p_x(z) G(z,x)\, \le \, 
%\\
C_d \sum_{z} \frac{p_x(z)}{(\vert z-x\vert +1)^{d-2}} \, \le \, \frac{C_d}{\left(\dist (x,\bbZ^d\setminus B)+1 \right)^{d-2}}\, .
\end{equation}
%\end{multline}
\end{proof}

Next we consider the case of $B= \mathring{\gL}_M
=\lint 1, M-1 \rint^d$ 
of (the trace of) and infinite volume LFF on $\bbZ^d \setminus \mathring{\gL}_M$.
So now $H=H_{\mathring{\gL}_M}$.
The covariance of $H$ can be written for $x, y\in \gL_M$ as 
\begin{equation}
\overbar G(x,y)\, :=\, \bE_x \left[ \sum_{n=\tau_M}^{\infty} \ind_{\{X_n=y\}} \right]\, ,
\end{equation} 

with $\tau_M$ the hitting time of the internal boundary of $\gL_M$ as in  \eqref{zoop}. 

\medskip

\begin{lemma}
\label{th:d+2}
For every $d\ge 3$ and every $\kappa> 1$ there is $C>0$ such that for every $r\ge 2$
\begin{equation}
\label{eq:d+2}
\sup_{x:\, \dist(x, \gL_M^\complement)\ge r }
\Var \left( H(x)-H(x+e)-H(x+g)+  H(x+e+g) \right)\, \le\,  C \frac{(\log (r))^\kappa}{r^{d+2}}\,.
\end{equation}
\end{lemma}

\medskip

\noindent
\emph{Proof.}
Setting $p_M(x,y):= P_x( X_{\tau_{M}}=y)$
we have 
\begin{multline}
\Var ( H(x)-H(x+e)-H(x+g)+  H(x+e+g) )\,=
\\
\sum_{y\in \partial \gL_M} \left[p_M(x,y)-p_M(x+e,y)- p_M(x+g,y)+ 
p_M(x+e+g,y) \right ]  \\
\times ( G(y,x)-G(y,x+e)-
G(y,x+g)+ G(y,x+e+g))\,.
\end{multline} 
The proof therefore follows if we show that for every $i,j\in \{1, \ldots, d\}$ and every $x$ and $y$
\begin{equation}
\label{eq:tbs-A1}
\left\vert G(y,x)-G(y,x+e)-
G(y,x+g)+ G(y,x+e+g)\right \vert \, \le \, C/\vert x-y \vert^d\, ,
\end{equation}
with $C$ a positive constant that depends only  on $d$, and that for every $x\in \gL_M$ at distance at least $r$ from  $\gL_M^\complement$ we have
\begin{equation}
\label{eq:tbs-A2}
\sum_{y\in \partial \gL_M} \left \vert p_M(x,y)-p_M(x+e,y)- p_M(x+g,y)+ 
p_M(x+e+g,y)\right \vert \, \le\,   C r^{-2 }(\log r)^\kappa\, ,
\end{equation}
with $C$ a positive constant that depends only  on $d$ and $\kappa$.

\medskip

For what concerns \eqref{eq:tbs-A1}, the bound follows directly from the asymptotic  Green function estimate in the limit when  $|x| \to \infty$
\cite[Th.~4.3.1]{cf:LL}, 
\begin{equation}
 G(0,x) - \frac{c_d}{\vert x \vert ^{d-2}}  \, = \, O\left(\frac{1}{\vert x \vert ^d}\right)\, ,  
\end{equation}
where $c_d$ is a positive constant. 

\medskip

For what concerns \eqref{eq:tbs-A2}
we observe   the left-hand side of \eqref{eq:tbs-A2}
coincides with four times the total variation distance between the distribution of $X_{\tau_M}$ for the simple random walk  starting with respective initial condition 
$\frac 1 2 (\gd_x+\gd_{x+e+g})$ and $\frac 1 2 (\gd_{x+e}+\gd_{x+g})$. The factor four comes from 
the usual factor two that relates the $L^1$ norm and the total distance and the fact two that comes from normalizing the initial measures. 
We let $P_1$ and $P_2$ denote the respective law of these two walks.
Set $t:= r^2/(\log r)^{\kappa}$. We have
\begin{multline}
\label{eq:Ablt}
 \left \Vert P_1( X_{\tau_{M}}\in \cdot)- P_2(X_{\tau_{M}}\in \cdot) \right\Vert_{TV}
 \\ \le  \, \left \Vert P_1( X_{t}\in \cdot)- P_2( X_{t}\in \cdot )\right\Vert_{TV}
 +P_1\left(\tau_{M}< t\right)+ P_2\left(\tau_{M} t\right)\,.
\end{multline}
The validity of \eqref{eq:Ablt} follows by observing that
\begin{multline}
\label{eq:forAblt}
\sum_{y\in \partial \gL_M} \left \vert P_1( X_{\tau_{M}}= y, \,
\tau_M\ge t)- P_2(X_{\tau_{M}}=y, \, \tau_M\ge t)\right\vert\, =\\
\sum_{y\in \partial \gL_M} \left \vert \sum_{x \in \gL_M}\big(P_1( X_{\tau_{M}}= y, \,X_t=x, \,
\tau_M\ge t)- P_2(X_{\tau_{M}}=y, \,X_t=x,\, \tau_M\ge t)\big)\right\vert \\
=\,
\sum_{y\in \partial \gL_M} \left \vert \sum_{x \in \gL_M}P_x(X_{\tau_{M}}= y)
\big(P_1(  X_t=x, \,
\tau_M\ge t)- P_2(X_t=x,\, \tau_M\ge t)\big)\right\vert
\\
\le \, 
 \sum_{x \in \gL_M} \left(\sum_{y\in \partial \gL_M}P_x(X_{\tau_{M}}= y)\right)
\big\vert P_1(  X_t=x, \,
\tau_M\ge t)- P_2(X_t=x,\, \tau_M\ge t)\big\vert
\\
\le \, 2 \left \Vert P_1( X_{t}\in \cdot)- P_2( X_{t}\in \cdot )\right\Vert_{TV} + P_1\left(\tau_{M}< t\right)+ P_2\left(\tau_{M}< t\right)\, ,
\end{multline}
where $P_x$ is used in the obvious sense. From \eqref{eq:forAblt} one easily obtains \eqref{eq:Ablt}

\medskip

Now we estimate the right-hand side of \eqref{eq:Ablt}: for the last two terms we observe that, with $P=P_1$ or 
$P=P_2$, we can apply directly \cite[Prop.~2.1.2, part 2]{cf:LL}  and we have that there exists a constant $C$ that depends only on $d$ such that 
\begin{equation}
P\left(\tau_{M}< t\right)\, \le \, C \exp\left(-r^2/(Ct)\right)\, =\, 
C \exp\left(-(\log (t))^{\kappa}\right)/C\, ,
\end{equation}
which is much smaller than any power than $1/t$ because $\kappa>1$.
 
 We can therefore focus on the  the right-hand side of \eqref{eq:Ablt}.
We have (translating every coordinate by $x$ and using the symmetries),
\begin{equation}
  \| P_1( X_{n}\in \cdot)- P_2( X_{n}\in \cdot )\|_{TV}\,=\, \frac{1}{4}\sum_{z\in \bbZ^d} \left \vert p_t(z)-p_{t}(z+e)-p_{t}(z+g)+p_{t}(z+e+g)  \right\vert\, ,
\end{equation}
where $p_t(z)=P_0(X_t =z)$.
Using the Local Central Limit Theorem \cite[Theorem 2.1.1 and Theorem 2.1.3]{cf:LL}
we have for some $d$ dependent constant $C$
\begin{equation}
 \left | p_t(z)-\frac{1}{(\sqrt{2\pi t})^{d/2}}e^{-\frac{|z|^2}{2t}} \right|\,\le\,
    C \left( t^{-\frac{d+2}{2}}\left(\frac{|z|^4}{t^2}+1\right) e^{-\frac{|z|^2}{2t}}+ 
 t^{-\frac{d+3}{2}} \right)\, .
\end{equation}
On the other hand, it is elementary to see that $p_t(z)\le \exp(- z^2/t)$ 
so  we choose $b\in (0,1/(2d))$ and we have
\begin{multline}
\label{eq:Ajdgev}
 \sum_{z \in \bbZ^d} \left \vert p_t(z)-p_{t}(z+e)-p_{t}(z+g)+p_{t}(z+e+g)  \right\vert \, \le \\ 
  C t^{-d/2}
 \sum_{|z|\le t^{b+1/2}} \left\vert e^{-\frac{|z|^2}{2t}}+ e^{-\frac{|z+e+g|^2}{2t}}-
 e^{-\frac{|z+e|^2}{2t}}- e^{-\frac{|z+g|^2}{2t}}\right\vert
 \\ 
 + C t^{-d/2}
 \sum_{|z|\le t^{b+1/2}} \left(\frac{1}{t}\left(\frac{|z|^4}{t^2}+1\right) e^{-\frac{|z|^2}{2t}}+  \frac1{t^{3/2}} \right) +
 2 \sum_{|z|> t^{b+1/2}} \exp(- z^2/t)\, .
\end{multline}
It is straightforward to see that the last line if bounded by $C/t$. We are therefore left with controlling the second line of 
\eqref{eq:Ajdgev}: we use Taylor formula and we obtain that there exists $C>0$ such that for every choice of $i$ and $j$
and for every $|z|\le t^{b+1/2}$ we have
\begin{equation}
 \left \vert e^{-\frac{|z|^2}{2t}}+ e^{-\frac{|z+e+g|^2}{2t}}-
 e^{-\frac{|z+e|^2}{2t}}- e^{-\frac{|z+g|^2}{2t}}\right \vert
 \,\le\, \frac C{2 t} \left(1+\frac{z_i^2+z_j^2}{t} \right) e^{-\frac{|z|^2}{2t}}\, \le \, \frac C t\, ,
\end{equation}
where the last inequality holds for $t$ sufficiently large. 
Altogether we obtain that 
\begin{equation}
   \| P_1[ X_{t}\in \cdot]- P_2[ X_{t}\in \cdot ]\|_{TV}\le \frac{C}{t}\, ,
\end{equation}
for a $C$ depending only on $d$ and $\kappa$ and the proof of \eqref{eq:tbs-A2} is complete.

The proof of Lemma~\ref{th:d+2} is thus completed.
\qed

\end{document}